\documentclass[onecolumn,draftclsnofoot,12pt]{IEEEtran}
\usepackage{aliascnt}
\usepackage{cite}
\usepackage{wrapfig}
\usepackage{gensymb}

\usepackage{array,makecell}

\IEEEoverridecommandlockouts

\usepackage{tikz}
\usetikzlibrary{arrows.meta, positioning, fit, calc}
\usepackage{multirow}
\usepackage{enumitem}
\usepackage{pifont}
\usepackage{amsmath,amssymb,amsfonts,amsthm}
\usepackage{algorithmic}
\usepackage[ruled,linesnumbered]{algorithm2e}
\usepackage{graphicx}
\usepackage{textcomp}
\usepackage{caption}
\usepackage[hidelinks]{hyperref}
\usepackage{float}
\usepackage[utf8]{inputenc}
\usepackage{tabularx}
\usepackage{booktabs}
\usepackage{cases}
\usepackage{subfigure}
\usepackage{bm}
\def\BibTeX{{\rm B\kern-.05em{\sc i\kern-.025em b}\kern-.08em
    T\kern-.1667em\lower.7ex\hbox{E}\kern-.125emX}}

\usepackage{longtable}
\usepackage{array}
\usepackage{makecell}

\usepackage{graphicx}
\usepackage{mathrsfs}
\usepackage{bm}

\newtheoremstyle{slanted}
{0em plus 0em minus 0em}
  {0em plus 0em minus 0em}
  {\em}
  {}
  {\bfseries}
  {.}
  { }
  {}
\theoremstyle{slanted}

\theoremstyle{slanted}

\theoremstyle{slanted}
\newtheorem{theorem}{Theorem}

\theoremstyle{slanted}

\theoremstyle{slanted}
\newtheorem{remark}{Remark}

\theoremstyle{slanted}
\newtheorem{proposition}{Proposition}

\theoremstyle{slanted}
\newtheorem{lemma}{Lemma}

\theoremstyle{slanted}

\title{A Unified Successive Approximation Framework for General Coupled Multiplicative and Fractional Optimization: HM-GM-AM-QM Transforms and Applications}

\author{\IEEEauthorblockN{Liangxin Qian, \emph{Member, IEEE}, Wenhan Yu, \emph{Member, IEEE}, Peiyuan Si, \emph{Member, IEEE}, Jun Zhao, \emph{Member, IEEE}, Ping Yang, \emph{Senior Member, IEEE}, Xuanyu Cao, \emph{Senior Member, IEEE}, and Kwok-Yan Lam, \emph{Senior Member, IEEE}}\vspace{-10pt}\thanks{Liangxin Qian, Wenhan Yu, Peiyuan Si, Jun Zhao, and Kwok-Yan Lam are all with Nanyang Technological University, Singapore. Ping Yang is with University of Electronic Science and Technology of China, China. Xuanyu Cao is with Washington State University, USA. Email: qian0080@e.ntu.edu.sg, wenhan002@e.ntu.edu.sg, peiyuan001@e.ntu.edu.sg, junzhao@ntu.edu.sg, yang.ping@uestc.edu.cn, xuanyu.cao@wsu.edu, kwokyan.lam@ntu.edu.sg.
}
}
\begin{document}
\maketitle

\begin{abstract}

Optimization problems in communication networks and information systems often contain coupled multiplicative or fractional terms, such as sum-of-products, sum-of-ratios, and logarithmic product-ratio structures. These problems are generally non-convex and difficult to solve, which motivates the development of tractable transformation and approximation techniques. In this paper, we propose an inequality-based transform framework for handling multiplicative and fractional terms involving an arbitrary number of coupled functions. The proposed framework is built upon the harmonic-mean, geometric-mean, arithmetic-mean, and quadratic-mean inequalities, and yields lower-bound and upper-bound surrogates for product-type terms. We derive the corresponding auxiliary-variable updates in closed form and show that the constructed surrogates are tight and first-order consistent at the current iterate. Based on these properties, we develop a class of successive approximation (SA) methods for sum-of-products/ratios minimization and maximization problems. When the transformed surrogate is convex for minimization or concave for maximization, the proposed method reduces to a standard successive convex approximation (SCA) method. When such convexity or concavity is not guaranteed, we further develop gradient-based SA variants and establish their sublinear convergence to an $\epsilon$-stationary point under standard smoothness and boundedness assumptions. We also discuss extensions to logarithmic product-ratio objectives and non-convex constraints. Numerical studies and application examples, including transmit-energy minimization, age-of-information minimization, semantic utility maximization, reliability-aware routing, cooperative edge caching, and product-loss learning, demonstrate the versatility and effectiveness of the proposed transform framework.
\end{abstract}



\section{Introduction}
Optimization is crucial in communication networks, where the primary goal is to enhance various performance metrics, e.g., throughput \cite{ju2013throughput,che2015spatial,yang2015throughput,shen2017fplinq,bhaskaran2010maximizing,razaviyayn2012linear,li2024maximizing}, efficiency \cite{ng2013wireless,miao2009cross,xu2013energy,wang2021resource,zhang2023towards}, and network utility \cite{palomar2006tutorial,huang2013priority,qian2024user,dehghan2019utility,papandriopoulos2009scale,li2021data}. A distinctive feature of these optimization problems is the frequent appearance of coupled multiplicative and fractional terms. These terms arise from the modeling of communication systems and protocols, reflecting the complex interactions between network parameters \cite{jeruchim2006simulation,pioro2004routing,rappaport2015wideband,lin2022centralized}. From the system level, the system performance is typically measured by functions containing several multiplicative or fractional terms from multiple links in a communication network \cite{shen2018fractional}. A summation operation is generally used for the overall system performance as well. This optimization family is called sum-of-products (ratios) optimization, which is generally non-convex and NP-hard \cite{freund2001solving,mei2018maximum,xing2012simple}. Existing studies either focused on optimizing a single ratio or product term \cite{zappone2015energy,isheden2012framework} or the highly complex algorithm to solve sum-of-ratios or products, e.g., the branch and bound algorithm \cite{benson2010branch,beale1979branch,konno2000branch}. 

To reduce the algorithm complexity, several transformation-based methods have been proposed. Jong \cite{jong2012efficient} proposes a transformation to solve the minimization of sum-of-ratios by analyzing its Karush–Kuhn–Tucker (KKT) conditions. Although a global optimum can be found by the transformation proposed in \cite{jong2012efficient}, it applies to the optimization only consisting of the sum-of-ratios term. In \cite{shen2018fractional}, Shen \textit{et al.} propose the quadratic transformation for decoupling the fractional terms in the maximization problem, which contains the sum-of-ratios and other general terms. Specifically, for maximizing the sum of $\frac{\text{concave}}{\text{convex}}$ terms, the quadratic transformation guarantees the transformed parametric convex optimization converges to a stationary point. However, the quadratic transformation of \cite{shen2018fractional} can't be used to solve the minimization problem containing the sum-of-ratios term. To fix it, \hbox{Zhao \textit{et al.} \cite{zhao2024}} propose another transformation and successfully convert the sum-of-ratios minimization to a parametric convex optimization. Nevertheless, the transform in \cite{zhao2024} mainly applies to the sum-of-ratios minimization terms in the $\frac{\text{convex}}{\text{concave}}$ form, and its intrinsic mathematical rationale is not explicitly explained.

\textbf{Motivation:}
Transformation-based methods have become important tools for reducing the complexity of non-convex fractional optimization problems. However, most existing transforms are designed for specific fractional structures. For example, the quadratic transform is effective for sum-of-ratios maximization, while the transform in \cite{zhao2024} targets sum-of-ratios minimization with a convex-over-concave form. These methods are powerful within their respective settings, but their applicability is limited when the objective contains more general multiplicative couplings or products involving more than two functions.

Such general multiplicative structures arise frequently in communication and information systems. For instance, end-to-end reliability can be expressed as a product of link success probabilities; cache-miss probability depends on the product of miss events across accessible caches; information freshness and semantic utility may involve coupled rate, delay, reliability, and semantic-quality factors; and resource-efficiency objectives often contain ratios that can be viewed as products with reciprocal terms. These examples suggest that ratio-type optimization is only one instance of a broader class of product-structured optimization problems. Therefore, a unified treatment of products and ratios is needed.

An additional difficulty is that useful transformations should be compatible with both minimization and maximization problems. For minimization, one typically needs a tight upper-bound surrogate to ensure descent, whereas for maximization, one needs a tight lower-bound surrogate to ensure ascent. Existing fractional programming transforms do not directly provide a general mechanism for constructing both types of surrogates for arbitrary multi-function products. This motivates the use of classical mean inequalities, which naturally relate harmonic, geometric, arithmetic, and quadratic means and thus provide a principled way to construct lower and upper bounds for product-type terms.

However, constructing a tight bound is only part of the challenge. The transformed surrogate is not always convex or concave, even if it is tight at the current iterate. Therefore, a rigorous algorithmic framework must distinguish between cases where the surrogate subproblem is convex or concave and cases where it is not. In the former case, the update can be interpreted as a standard successive convex approximation (SCA) step. In the latter case, the surrogate can still guide the optimization through gradient descent or ascent, leading to a gradient-based successive approximation (SA) method.

These observations motivate the development of an inequality-based transform framework for multiplicative and fractional optimization. Such a framework should be able to decouple products with an arbitrary number of coupled functions, provide suitable lower or upper surrogates depending on the problem direction, clarify when the resulting method is a standard SCA method, and offer gradient-based alternatives when convexity or concavity is not guaranteed.

\textbf{Our Contributions:}
The contributions of this paper are summarized as follows.
\begin{itemize}
    \item We propose transforms based on the harmonic, geometric, arithmetic, and quadratic means (HM-GM-AM-QM) for decoupling general multiplicative and fractional terms with an arbitrary number of coupled functions. The proposed framework constructs an HM lower bound and AM/QM upper bounds from the mean-inequality chain. We also derive the closed-form expression of the introduced auxiliary variables in (\ref{y_n_closedform}) by mathematical induction. These auxiliary variables ensure that the constructed surrogates are tight at the current iterate.

    \item We develop SA methods for sum-of-products minimization and maximization problems. For minimization problems, the AM/QM bounds provide upper-bound surrogates, while for maximization problems, the HM bound provides a lower-bound surrogate. We further clarify that the resulting method is a standard SCA method only when the transformed surrogate is convex or concave. When this condition does not hold, we propose gradient-based SA methods based on projected gradient descent or ascent.

    \item We analyze the convergence and complexity of the proposed methods. For the exact SA methods, stationary-point convergence follows under standard surrogate approximation conditions. For the gradient-based SA variants, we establish sublinear convergence to an $\epsilon$-stationary point under smoothness, boundedness, tightness, and first-order consistency assumptions. The complexity of both exact and gradient-based implementations is also discussed.

    \item We show that the proposed transform framework can be applied not only to objective functions but also to non-convex constraints. Product-type constraints can be conservatively approximated using the proposed upper- or lower-bound surrogates. We also discuss extensions to logarithmic product-ratio objectives by combining the monotonicity of the logarithm with the proposed product surrogates.

    \item We demonstrate the applicability of the proposed framework through numerical examples and applications in communication networks, information systems, and machine learning. These include transmit-energy minimization, age-of-information minimization, semantic information utility maximization, reliability-aware multi-hop information delivery, cooperative edge caching, multi-modal product-loss learning, and clean-adversarial product-loss learning.
\end{itemize}


\begin{figure}[t]
\centering
\footnotesize
\begin{tikzpicture}[
    >=Latex,
    every node/.style={align=left},
    sec2box/.style={
        draw,
        rounded corners=1.8mm,
        fill=green!8,
        text width=0.92\columnwidth,
        inner sep=4pt
    },
    sec3box/.style={
        draw,
        rounded corners=1.8mm,
        fill=orange!8,
        text width=0.92\columnwidth,
        inner sep=4pt
    },
    sec4box/.style={
        draw,
        rounded corners=1.8mm,
        fill=purple!7,
        text width=0.92\columnwidth,
        inner sep=4pt
    },
    sec5box/.style={
        draw,
        rounded corners=1.8mm,
        fill=red!6,
        text width=0.92\columnwidth,
        inner sep=4pt
    },
    line/.style={thick, -{Latex[length=2mm]}}
]

\node[sec2box] (sec2) {
\textbf{Section \ref{sec_proposed_transforms}. Proposed HM-GM-AM-QM Transforms}\\
\textbf{Subsections:}\\
\ref{sec_existing_method}: Existing fractional programming technique, supplementary rationale for intrinsic construction logic behind the transform, and its extension to multiplicative terms.
\\
\ref{sec_infactSCA}: A new finding of existing iterative methods based on fractional programming techniques.
\\
\ref{sec_multi_twofunctions}: Proposed HM-GM-AM-QM transforms for multiplicative term coupling of two functions.\\
\ref{sec_extension_arbitrary_functions}: Proposed HM-GM-AM-QM transforms for multiplicative terms with an arbitrary number of coupled functions.
};

\node[sec3box, below=6mm of sec2] (sec3) {
\textbf{Section \ref{sec_AMbound_SCA}. Proposed SA Methods to Solve Sum-of-Products Optimization by Using the HM, AM, and QM Surrogates}\\
\textbf{Subsections:}\\
\ref{sec_sop_min}: Sum-of-products minimization including cases where the transformed functions are convex and non-convex.
\\
\ref{sec_max_sca}: Sum-of-products maximization including cases where the transformed functions are convex and non-convex.
\\
\ref{sec_sum_sa}: Summary of SA methods based on HM-GM-AM-QM transforms.
\\
\ref{sec_extension_log}: Extension to logarithmic problems.
\\
\ref{sec_comp_existingmethods}: Comparison to existing methods.\\
\ref{sec_app_constraint}: Application in optimization constraints.\\
\ref{sec_simple_results}: Simple numerical experiments.
};

\node[sec4box, below=6mm of sec3] (sec4) {
\textbf{Section \ref{sec_application}. Applications and Numerical Studies}\\
\textbf{Subsections:}\\
\ref{sec_minimization_transmission_energy}: Minimization of transmission energy in wireless communications.\\
\ref{sec_aoi}: Age-of-information minimization in status update networks.\\
\ref{sec_semantic}: Semantic information utility maximization in wireless networks.\\
\ref{sec_reliability}: Reliability-aware multi-hop information delivery.\\
\ref{sec_cache}: Cooperative edge caching for information availability.\\
\ref{sec_ml_prodloss}: Multi-modal product-loss learning.\\
\ref{sec_ml_clean}: Clean-adversarial product-loss learning.\\
\ref{sec_summary_application}: Summary of representative application problems that can be handled by the proposed SA methods.
};

\draw[line] (sec2.south) -- (sec3.north);
\draw[line] (sec3.south) -- (sec4.north);

\end{tikzpicture}
\vspace{-2mm}
\caption{Organization of the main technical contents of this paper.}
\label{fig_paper_overview}
\end{figure}

To improve readability, Fig.~\ref{fig_paper_overview} summarizes the organization of the main technical contents of this manuscript. The remainder of this paper is structured as follows. We first present our main theoretical result, i.e., HM-GM-AM-QM transforms, in Section \ref{sec_proposed_transforms}. Specifically, in Section \ref{sec_rationale_zhao_transformation}, we use the GM-AM inequality to illustrate the core idea behind the transform in \cite{zhao2024}. Under the same inequality-based perspective, we then derive HM-GM-AM-QM bounds for multiplicative terms with an arbitrary number of coupled functions in Section \ref{sec_extension_arbitrary_functions}. Section \ref{sec_AMbound_SCA} develops the proposed SA methods for sum-of-products minimization and maximization, including exact and gradient-based variants, logarithmic extensions, comparisons with representative FP methods, and constraint transformations. Section \ref{sec_application} presents application examples and numerical studies. Section \ref{sec_limitation} discusses the limitations of the proposed transforms. Finally, Section \ref{secConclusion} concludes the paper.

\section{Proposed HM-GM-AM-QM Transforms}\label{sec_proposed_transforms}
Fractional and multiplicative programming are typical optimization problems where multiple fractional and multiplicative terms are generally coupled. We will review recent research on solving the sum-of-ratios minimization problem, with each ratio containing two coupled functions in \cite{zhao2024} and its extended form in solving the sum-of-products minimization problem. Since there are no rationales for the transformation proposed in \cite{zhao2024}, we will present a detailed rationale behind the transformation from the perspective of the GM-AM inequality, where GM and AM are short for ``geometric mean'' and ``arithmetic mean'', respectively. We then propose our novel decoupling techniques and bounds for solving the sum-of-products minimization and maximization problems with an arbitrary number of coupled functions in each product. We focus on the sum-of-products problem after Section \ref{sec_extension_multiplicative_terms} because the fractional terms are included in the multiplicative terms.

In the rest of the paper, we use $\mathbb{R}$ (resp., $\mathbb{R}_{+}$, and $\mathbb{R}_{++}$) to 
denote the set of \textit{real} (resp., \textit{non-negative} and \textit{strictly positive})  numbers. 

\subsection{Existing Technique in \cite{zhao2024}}\label{sec_existing_method}
 Let the optimization variable $\bm{x}$ belong to a compact convex set $\mathcal{X}\subseteq \mathbb{R}^M$. Define $\mathcal{N}:=\{1,2,\cdots,N\}$. With functions $A_n(\bm{x}): \mathbb{R}^M \rightarrow \mathbb{R}_{++}$, $B_n(\bm{x}): \mathbb{R}^M \rightarrow \mathbb{R}_{++}$, $\forall n \in \mathcal{N}$, and $G(\bm{x}): \mathbb{R}^M \rightarrow \mathbb{R}$, we begin by considering the following minimization problem:
\begin{subequations}\label{prob1}
\begin{align}
\min\limits_{\bm{x}}\quad &G(\bm{x})+\sum_{n=1}^N \frac{A_n(\bm{x})}{B_n(\bm{x})}\\
\text{s.t.} \quad & \bm{x} \in \mathcal{X}.
\end{align}
\end{subequations}
The above minimization problem contains a general function $G(\bm{x})$ and sum-of-ratios term $\sum_{n=1}^N \frac{A_n(\bm{x})}{B_n(\bm{x})}$, which makes the optimization challenging, where the sum-of-ratios term is often non-convex \cite{freund2001solving}. Although only the $\frac{\text{convex}}{\text{concave}}$ case is considered in \cite{zhao2024}, the transformation proposed in \cite{zhao2024} can also be used in decoupling general fractional terms, with each term containing two functions. However, a detailed rationale for constructing that transformation is not given in \cite{zhao2024}. Next, we will delve into the intrinsic construction logic of the transformation in \cite{zhao2024} from the perspective of the GM-AM inequality.

\subsubsection{Supplementary Rationale for Intrinsic Construction Logic Behind the Transform in \cite{zhao2024}}\label{sec_rationale_zhao_transformation}
Recall the GM-AM inequality:
\begin{equation}
    ab\leq\frac{a^2+b^2}{2},
\end{equation}
for $a,b \in \mathbb{R}_{+}$, where the equal sign holds if and only if (iff) $a=b$. It is easy to get that
\begin{equation}
    \frac{A_n(\bm{x})}{B_n(\bm{x})} \leq \frac{\left(A_n(\bm{x})\right)^2 + \frac{1}{\left(B_n(\bm{x})\right)^2}}{2},
\end{equation}
where iff $A_n(\bm{x}) = \frac{1}{B_n(\bm{x})}$, the equal sign holds. Obviously, $\frac{1}{2} \left(\left(A_n(\bm{x})\right)^2 + \frac{1}{\left(B_n(\bm{x})\right)^2}\right)$ is a global upper bound for $\frac{A_n(\bm{x})}{B_n(\bm{x})}$. If there exists one point $\bm{x}^\prime$ that satisfies $A_n(\bm{x}^\prime) = \frac{1}{B_n(\bm{x}^\prime)}$, these two curves are tangent to $\bm{x}^\prime$. The current form of the upper bound is not good enough for iterative optimization, because it is not necessarily tight at a chosen feasible point. To obtain a useful surrogate, the upper bound should include the information of the current iterate or the chosen feasible point. Specifically, we first introduce an auxiliary variable $\bm{y}:=[y_n]|_{n \in \mathcal{N}}$. Given an initial feasible point $\bm{x}^\prime$, the resulting upper bound with fixed $\bm{y}^\prime$ is tangent to the original function at $\bm{x}^\prime$. Then, by minimizing the surrogate with fixed $\bm{y}^\prime$, we obtain the next iterate $\bm{x}^{\prime\prime}$. Repeating this procedure yields a successive approximation process toward a stationary point. We illustrate how to achieve this by modifying $\frac{1}{2} \left(\left(A_n(\bm{x})\right)^2 + \frac{1}{\left(B_n(\bm{x})\right)^2}\right)$ in the following.

We define a new function $g_n (\bm{x}; y_n(\bm{x}^\prime))$
as the new surrogate, where $y_n(\bm{x}^\prime)$ is fixed based on the chosen feasible point $\bm{x}^\prime$. Here, $y_n(\bm{x}^\prime)$ is used to record the chosen feasible point $\bm{x}^\prime$ information. Specifically, for this feasible point $\bm{x}^\prime$, it is updated to make the surrogate $g_n (\bm{x}; y_n(\bm{x}^\prime))$ is tight at $\bm{x}^\prime$. If $y_n$ is not fixed, we denote the surrogate as $g_n (\bm{x}, y_n)$. Note that initially, $g_n(\cdot)$ was a function of the variables $\bm{x}$ and $y_n$. Only when we first optimize $y_n$ and obtain its optimal value, which includes the variable $\bm{x}$, do we then write $y_n$ as $y_n(\bm{x})$.

 Similar to the properties of the quadratic transform proposed in \cite{shen2018fractional}, we also list the construction principles (CPs) of the new surrogate as follows:
\begin{itemize}[leftmargin=*]
    \item CP1 \textit{(Equivalent Solution):} $\bm{x}^{\ast}$ is the optimal or local optimal solution of the problem $\min \frac{A_n(\bm{x})}{B_n(\bm{x})}$ if and only if $y_n^{\ast}(\bm{x}^*)$ minimizes $ g_n (\bm{x},y_n)$ with the same $\bm{x}^{\ast}$.
    \item CP2 \textit{(Equivalent Objective):} For fixed $\bm{x}$, if $y_n^{\ast}(\bm{x})$ is the optimal solution to $g_n (\bm{x}, y_n)$, then $g_n (\bm{x}; y_n^*(\bm{x})) = \frac{A_n(\bm{x})}{B_n(\bm{x})}$.
    \item CP3 \textit{(Convexity):} For fixed $\bm{x}$, $g_n (\bm{x}, y_n)$ is convex over $y_n$. Note that the CP3 is not compulsory as we can directly obtain the optimal value of $y_n$ by the equality conditions of the HM-GM-AM-QM inequality, which we will illustrate in detail in Section \ref{sec_extension_arbitrary_functions}.
\end{itemize}
CP1 and CP2 guarantee the equivalent solution, and objective function and value. CP3 guarantees that we can find the optimal value of $y_n$ with fixed $\bm{x}$ based on the surrogate function. There is also a decoupling property that the transformation should have, but we don't need it here due to the fact that GM-AM inequality itself ensures decoupling.
\begin{remark} \label{remark1}
In CP1, while it is assured that the solution $\bm{x}^{\ast}$ to the transformed optimization using the surrogate $g_n (\cdot)$ matches the solution to the original optimization, the optimality of these solutions differs significantly. Specifically, in the original optimization, $\bm{x}^{\ast}$ is intended to be the global optimum. However, in the transformed scenario, we may only achieve a stationary or KKT point under certain conditions, implying that the solutions might be suboptimal. This shift in optimality and its implications will be explored in depth in Section \ref{sec_AMbound_SCA}. 
\end{remark}
In the preceding paragraphs, we supplement the motivation and rationale not found in the theoretical parts in \cite{zhao2024}. In \cite{zhao2024}, it is proved that the quadratic transform in \cite{shen2018fractional} can't be used in the minimization problem. Thus, the authors in \cite{zhao2024} propose a bound for the minimization problem as
\begin{equation}
    g (\bm{x}; y_n(\bm{x})) = y_n(\bm{x}) \left(A_n(\bm{x})\right)^2  + \frac{1}{4y_n(\bm{x})\left(B_n(\bm{x})\right)^2 },
\end{equation}
where $y_n (\bm{x})= \frac{1}{2 A_n(\bm{x}) B_n(\bm{x})}$. It is obvious that this bound satisfies three construction principles, and we easily find that 
\begin{align}
    \frac{A_n(\bm{x})}{B_n(\bm{x})} 
    &= \sqrt{2 y_n \left(A_n(\bm{x})\right)^2 \cdot \frac{1}{2 y_n \left(B_n(\bm{x})\right)^2}} \nonumber \\
    &\leq \frac{1}{2} \left(2 y_n \left(A_n(\bm{x})\right)^2 + \frac{1}{2 y_n( \left(B_n(\bm{x})\right)^2}\right)\nonumber \\
    &= y_n \left(A_n(\bm{x})\right)^2  + \frac{1}{4y_n\left(B_n(\bm{x})\right)^2 },
\end{align}
where if and only if (iff) $y_n= \frac{1}{2 A_n(\bm{x}) B_n(\bm{x})}$, the equal sign holds. If we choose one feasible point $\bm{x}^\prime$ of the original problem (\ref{prob1}) and set $y_n (\bm{x}^\prime)= \frac{1}{2 A_n(\bm{x}^\prime) B_n(\bm{x}^\prime)}$, the upper bound function and the original function are tangent to the point $(\bm{x}^\prime, \frac{A_n(\bm{x}^\prime)}{B_n(\bm{x}^\prime)})$. From the GM-AM inequality perspective, we can derive the same transformation in \cite{zhao2024}. But we give a more detailed and clearer intrinsic construction logic behind that transformation.

Based on the above discussion, Problem (\ref{prob1}) can be transformed to the following Problem (\ref{prob1_transformed}):
\begin{subequations}\label{prob1_transformed}
\begin{align}
\min\limits_{\bm{x},\bm{y}}\quad &G(\bm{x})+\sum_{n=1}^N \left(y_n \left(A_n(\bm{x})\right)^2  + \frac{1}{4y_n\left(B_n(\bm{x})\right)^2}\right) \\
\text{s.t.} \quad & \bm{x} \in \mathcal{X},
\end{align}
\end{subequations}
which is convex over $\bm{y}$ with fixed $\bm{x}$. Obviously, the current upper bound satisfies CP1, CP2 and CP3, and it is better than $\frac{1}{2}\left(\left(A_n(\bm{x})\right)^2 + \frac{1}{\left(B_n(\bm{x})\right)^2}\right)$.

\subsubsection{Extension to Multiplicative Terms}\label{sec_extension_multiplicative_terms}
For a new minimization problem with sum-of-products and each product containing two functions coupled together:
\begin{subequations}\label{prob2}
\begin{align}
\min\limits_{\bm{x}}\quad &G(\bm{x})+\sum_{n=1}^N A_n(\bm{x})B_n(\bm{x})\\
\text{s.t.} \quad & \bm{x} \in \mathcal{X}. 
\end{align}
\end{subequations}
In this case, the upper bound surrogate would be
\begin{align}
    g_n(\bm{x}; y_n(\bm{x})) = \left(A_n(\bm{x})\right)^2 y_n(\bm{x}) + \frac{\left(B_n(\bm{x})\right)^2}{4 y_n(\bm{x})},
\end{align}
where $y_n(\bm{x}) = \frac{B_n(\bm{x})}{2 A_n(\bm{x})}$. Thus, Problem (\ref{prob2}) can be transformed to the following Problem (\ref{prob2_transformed}):
\begin{subequations}\label{prob2_transformed}
\begin{align}
\min\limits_{\bm{x},\bm{y}}\quad &G(\bm{x})+\sum_{n=1}^N \left(A_n(\bm{x})\right)^2 y_n + \frac{\left(B_n(\bm{x})\right)^2}{4 y_n} \\
\text{s.t.} \quad & \bm{x} \in \mathcal{X},
\end{align}
\end{subequations}
which is also convex over $\bm{y}$ with fixed $\bm{x}$.

\subsection{A New Finding of Proposed Iterative Approaches in \cite{shen2018fractional}~\cite{zhao2024}}\label{sec_infactSCA}
The iterative approaches based on the proposed transformations in \cite{shen2018fractional}~\cite{zhao2024} are actually SCA methods. 
These iterative optimization methods are block coordinate descent (BCD) methods at first glance. But, the proposed transformations in \cite{shen2018fractional}~\cite{zhao2024} are actually upper bounds and lower bounds of the original multiplicative coupled terms. The current-iterate information is all included in $\bm{y}$. The iterative optimization of $\bm{x}$ and $\bm{y}$ is, in fact, the SCA iteration by using bound functions to find the optimum.

In the following sections, we only focus on multiplicative terms since fractional terms can be included in multiplicative terms. We further analyze the product of two terms by using the HM-GM-AM-QM inequality to de rive the lower bound and two upper bounds of it.

\subsection{Extension to the Multiplicative Term Coupling of Two Functions Based on the HM-GM-AM-QM Inequality}\label{sec_multi_twofunctions}
With HM, GM, AM, and QM standing for ``harmonic mean'', ``geometric mean'',  ``arithmetic mean'', and ``quadratic mean'', respectively,
the well-known HM-GM-AM-QM inequality is given by
\begin{align}
    \frac{2}{\frac{1}{a}+\frac{1}{b}}\leq\sqrt{ab}\leq\frac{a+b}{2}\leq\sqrt{\frac{a^2+b^2}{2}},
\end{align}
where $a$, $b \in \mathbb{R}_{++}$ and if and only if (iff) $a=b$, the equality holds. Similarly, by modifying the forms of $a$ and $b$ to $a y$ and $\frac{b}{y}$, we obtain that
\begin{align}
     \frac{2}{\frac{1}{ay}+\frac{y}{b}}\leq\sqrt{ay \cdot \frac{b}{y}}\leq\frac{ay+\frac{b}{y}}{2}\leq\sqrt{\frac{a^2 y^2+\frac{b^2}{y^2}}{2}},
\end{align}
iff $ay=\frac{b}{y}$, i.e., $y=\sqrt{\frac{b}{a}}$, the equal sign can be achieved. Therefore, for the multiplicative term $A_n(\bm{x}) B_n(\bm{x})$, we can know that
\begin{align}
    \frac{2}{\frac{1}{\left(A_n(\bm{x})\right)^2 y_n}+\frac{y_n}{\left(B_n(\bm{x})\right)^2}}&\leq\sqrt{\left(A_n(\bm{x})\right)^2 y_n \cdot \frac{\left(B_n(\bm{x})\right)^2}{y_n}}\nonumber \\  
    &\leq\frac{\left(A_n(\bm{x})\right)^2 y_n+\frac{\left(B_n(\bm{x})\right)^2}{y_n}}{2}\nonumber \\  
    &\leq\sqrt{\frac{\left(A_n(\bm{x})\right)^4 y_n^2+\frac{\left(B_n(\bm{x})\right)^4}{y_n^2}}{2}},\label{eq_inequality_two_functions}
\end{align}
where iff $\left(A_n(\bm{x})\right)^2 y_n = \frac{\left(B_n(\bm{x})\right)^2}{y_n}$, i.e., $y_n = \frac{B_n(\bm{x})}{A_n(\bm{x})}$, the equality holds. Here, we get one lower bound and two upper bounds of $A_n(\bm{x}) B_n(\bm{x})$. Note that the AM upper bound is convex over $y_n$, satisfying three construction principles, i.e., CP1, CP2, and CP3. Functions $\frac{2}{\frac{1}{\left(A_n(\bm{x})\right)^2 y_n}+\frac{y_n}{\left(B_n(\bm{x})\right)^2}}$ and $\sqrt{\frac{\left(A_n(\bm{x})\right)^4 y_n^2+\frac{\left(B_n(\bm{x})\right)^4}{y_n^2}}{2}}$ don't satisfy CP3, i.e., they may be non-convex over $y_n$. But it doesn't matter since we can also obtain the optimal value of $y_n$ based on the equality condition. We can still use the HM and QM bounds as the successive surrogates of the original problems, and solve them by gradient descent or other methods. We briefly discuss cases where the proposed HM and QM bounds have good properties.

\subsubsection{Cases Where the HM Bound has good properties}
For some cases, $\frac{2}{\frac{1}{\left(A_n(\bm{x})\right)^2 y_n}+\frac{y_n}{\left(B_n(\bm{x})\right)^2}}$ \vspace{1pt} may show some monotonicity properties. If it is not difficult to obtain the minimum of the lower bound in the given intervals, we can use this HM lower bound with the AM upper bound, i.e., $\frac{\left(A_n(\bm{x})\right)^2 y_n+\frac{\left(B_n(\bm{x})\right)^2}{y_n}}{2}$, to obtain the global minimum of the original optimization with the branch and bound method.

\subsubsection{Cases Where the QM Bound has good properties}
From 
\begin{align}
    \sqrt{\left(A_n(\bm{x})\right)^2 y_n \cdot \frac{\left(B_n(\bm{x})\right)^2}{y_n}}\leq \sqrt{\frac{\left(A_n(\bm{x})\right)^4 y_n^2+\frac{\left(B_n(\bm{x})\right)^4}{y_n^2}}{2}}, 
\end{align}
it is known that
\begin{align}
    \left(A_n(\bm{x})\right)^2 y_n \cdot \frac{\left(B_n(\bm{x})\right)^2}{y_n} \leq 
    \frac{\left(A_n(\bm{x})\right)^4 y_n^2+\frac{\left(B_n(\bm{x})\right)^4}{y_n^2}}{2}.
\end{align}
It shows that when we encounter the coupled terms $\left(A_n(\bm{x})\right)^2\left(B_n(\bm{x})\right)^2$, the function 
\begin{equation}
    \frac{\left(A_n(\bm{x})\right)^4 y_n^2+\frac{\left(B_n(\bm{x})\right)^4}{y_n^2}}{2} \nonumber
\end{equation}
is convex over $y_n$ and would be a useful bound to decouple the multiplicative term for the minimization problem.

\subsection{Extension to Multiplicative Terms with an Arbitrary Number of Coupled Functions Based on the HM-GM-AM-QM Inequality}
\label{sec_extension_arbitrary_functions}

We now present the main transform result for multiplicative terms with an arbitrary number of coupled functions. Let \(\mathcal{K}:=\{1,2,\ldots,K\}\). For any \(a_k\in\mathbb{R}_{++}\), the HM-GM-AM-QM inequality gives
\begin{equation}
    \frac{K}{\sum_{k=1}^{K}\frac{1}{a_k}}
    \leq
    \left(\prod_{k=1}^{K}a_k\right)^{\frac{1}{K}}
    \leq
    \frac{1}{K}\sum_{k=1}^{K}a_k
    \leq
    \sqrt{\frac{1}{K}\sum_{k=1}^{K}a_k^2},
    \label{eq:hmgmamqm_basic}
\end{equation}
where equality holds if and only if \(a_1=a_2=\cdots=a_K\).

The key idea is to apply \eqref{eq:hmgmamqm_basic} to properly scaled positive terms whose geometric mean is exactly the original product term. Consider the product $\prod_{k=1}^{K} f_n^{(k)}(\bm{x})$,
where \(f_n^{(k)}(\bm{x}):\mathbb{R}^{M}\rightarrow\mathbb{R}_{++}\). For any
$
\bm{y}_n=
\left[
y_n^{(1)},y_n^{(2)},\ldots,y_n^{(K-1)}
\right]^{\intercal}
\in\mathbb{R}_{++}^{K-1}$,
we construct the following \(K\) scaled terms:
\begin{align}
    \left(f_n^{(1)}(\bm{x})\right)^K
    \prod_{j=1}^{K-1}y_n^{(j)}, 
    \left(f_n^{(k)}(\bm{x})\right)^K
    \frac{\prod_{j=k}^{K-1}y_n^{(j)}}
    {\left(y_n^{(k-1)}\right)^{k-1}},
    \frac{\left(f_n^{(K)}(\bm{x})\right)^K}
    {\left(y_n^{(K-1)}\right)^{K-1}},
    \label{eq:scaled_terms_no_phi}
\end{align}
where $k=2,\ldots,K-1$.
The product of all \(K\) scaled terms in \eqref{eq:scaled_terms_no_phi} is
$
\left(\prod_{k=1}^{K} f_n^{(k)}(\bm{x})\right)^K
$
. Therefore, the geometric mean of the scaled terms is exactly
$
\prod_{k=1}^{K} f_n^{(k)}(\bm{x})$.

\textbf{HM-GM-AM-QM Transforms:} We consider optimization problems where the term \newline$\sum_{n=1}^N \left(\prod_{k=1}^K f^{(k)}_n(\bm{x})\right)$ is included in the objective function. 
For any $f_n^{(k)}(\bm{x})\in\mathbb{R}^{M}\rightarrow\mathbb{R}_{++}$, and \(\bm{y}_n\in\mathbb{R}_{++}^{K-1}\), the following inequalities hold:
\begin{align}
    &F_n^{ HM}(\bm{x},\bm{y}_n)
    \leq
    \prod_{k=1}^{K} f_n^{(k)}(\bm{x})
    \leq
    F_n^{ AM}(\bm{x},\bm{y}_n)
    \leq
    F_n^{ QM}(\bm{x},\bm{y}_n),
    \label{bounds_K}
\end{align}
where
\begin{align}
    F_n^{ HM}(\bm{x},\bm{y}_n)
    =
    \frac{K}{
    \frac{1}{
    \left(f_n^{(1)}(\bm{x})\right)^K
    \prod_{j=1}^{K-1}y_n^{(j)}
    }
    +
    \sum_{k=2}^{K-1}
    \frac{1}{
    \left(f_n^{(k)}(\bm{x})\right)^K
    \frac{\prod_{j=k}^{K-1}y_n^{(j)}}
    {\left(y_n^{(k-1)}\right)^{k-1}}
    }
    +
    \frac{
    \left(y_n^{(K-1)}\right)^{K-1}
    }{
    \left(f_n^{(K)}(\bm{x})\right)^K
    }
    },
    \label{eq:HM_bound_no_phi}
\end{align}
\begin{align}
    F_n^{ AM}(\bm{x},\bm{y}_n)
    =
    \frac{1}{K}
    \bigg[
    &
    \left(f_n^{(1)}(\bm{x})\right)^K
    \prod_{j=1}^{K-1}y_n^{(j)}
    +
    \sum_{k=2}^{K-1}
    \left(f_n^{(k)}(\bm{x})\right)^K
    \frac{\prod_{j=k}^{K-1}y_n^{(j)}}
    {\left(y_n^{(k-1)}\right)^{k-1}}
    +
    \frac{\left(f_n^{(K)}(\bm{x})\right)^K}
    {\left(y_n^{(K-1)}\right)^{K-1}}
    \bigg],
    \label{eq:AM_bound_no_phi}
\end{align}
and
\begin{align}
    F_n^{ QM}(\bm{x},\bm{y}_n)
    =
    \Bigg\{
    \frac{1}{K}
    \Bigg[
    &
    \left(
    \left(f_n^{(1)}(\bm{x})\right)^K
    \prod_{j=1}^{K-1}y_n^{(j)}
    \right)^2
    +
    \sum_{k=2}^{K-1}
    \left(
    \left(f_n^{(k)}(\bm{x})\right)^K
    \frac{\prod_{j=k}^{K-1}y_n^{(j)}}
    {\left(y_n^{(k-1)}\right)^{k-1}}
    \right)^2
    \nonumber\\
    &+
    \left(
    \frac{\left(f_n^{(K)}(\bm{x})\right)^K}
    {\left(y_n^{(K-1)}\right)^{K-1}}
    \right)^2
    \Bigg]
    \Bigg\}^{\frac{1}{2}}.
    \label{eq:QM_bound_no_phi}
\end{align}
Moreover, equality in \eqref{bounds_K} holds if and only if
\begin{align}
    &
    \left(f_n^{(1)}(\bm{x})\right)^K
    \prod_{j=1}^{K-1}y_n^{(j)}
    =
    \left(f_n^{(2)}(\bm{x})\right)^K
    \frac{\prod_{j=2}^{K-1}y_n^{(j)}}{y_n^{(1)}}
    =
    \left(f_n^{(3)}(\bm{x})\right)^K
    \frac{\prod_{j=3}^{K-1}y_n^{(j)}}{\left(y_n^{(2)}\right)^2}
    \nonumber\\
    &=
    \cdots
    =
    \left(f_n^{(K-1)}(\bm{x})\right)^K
    \frac{y_n^{(K-1)}}{\left(y_n^{(K-2)}\right)^{K-2}}
    =
    \frac{\left(f_n^{(K)}(\bm{x})\right)^K}
    {\left(y_n^{(K-1)}\right)^{K-1}}.
    \label{eq:equality_condition_no_phi}
\end{align}
For a given point \(\bm{x}^\prime\), the equality condition is satisfied by the recursive update
\begin{align}
    y_n^{(1)}(\bm{x}^\prime)
    &=
    \left(
    \frac{
    f_n^{(2)}(\bm{x}^\prime)
    }{
    f_n^{(1)}(\bm{x}^\prime)
    }
    \right)^{\frac{K}{2}},
    y_n^{(k-1)}(\bm{x}^\prime)
    &=
    \left[
    \left(y_n^{(k-2)}(\bm{x}^\prime)\right)^{k-2}
    \left(
    \frac{
    f_n^{(k)}(\bm{x}^\prime)
    }{
    f_n^{(k-1)}(\bm{x}^\prime)
    }
    \right)^K
    \right]^{\frac{1}{k}},
    \quad k=3,\ldots,K.
    \label{eq:y_update_k_no_phi}
\end{align}
Note that we also denote HM, AM, and QM surrogates as $F_n^{HM}(\bm{x},\bm{y}_n)$, $F_n^{AM}(\bm{x},\bm{y}_n)$, and $F_n^{QM}(\bm{x},\bm{y}_n)$, respectively, with unfixed $\bm{y}_n$. For fixed $\bm{y}_n(\bm{x}^\prime)$ with the chosen feasible point $\bm{x}^\prime$, we denote them as $F_n^{HM}(\bm{x};\bm{y}_n(\bm{x}^\prime))$, $F_n^{AM}(\bm{x};\bm{y}_n(\bm{x}^\prime))$, and $F_n^{QM}(\bm{x};\bm{y}_n(\bm{x}^\prime))$, respectively.

Here we have presented a unified way to construct lower and upper surrogates for multiplicative terms with an arbitrary number of coupled functions. The HM bound gives a lower-bound surrogate, while the AM and QM bounds give upper-bound surrogates. These bounds are useful for optimization because they are not arbitrary approximations: after updating the auxiliary variables according to the equality condition, the surrogate is tight at the current iterate and preserves the local first-order information of the original multiplicative term. For minimization problems, a tight upper-bound surrogate can be minimized to obtain a descent update, while for maximization problems, a tight lower-bound surrogate can be maximized to obtain an ascent update. Therefore, the HM transform is naturally suited to maximization problems, whereas the AM and QM transforms are suited to minimization problems, provided that the resulting surrogate subproblem is tractable or is solved approximately by a gradient-based successive approximation step. Choose the feasible point $\bm{x}^\prime$, and then, based on Equation \eqref{eq:y_update_k_no_phi}, the constructed bounds are tight at \(\bm{x}^\prime\), i.e.,
\begin{align}
    F_n^{HM}(\bm{x}^\prime;\bm{y}_n(\bm{x}^\prime))
    =
    \prod_{k=1}^{K} f_n^{(k)}(\bm{x}^\prime)
    =
    F_n^{\rm AM}(\bm{x}^\prime;\bm{y}_n(\bm{x}^\prime))
    =
    F_n^{\rm QM}(\bm{x}^\prime;\bm{y}_n(\bm{x}^\prime)).
\end{align}
This tightness property is the basis of the successive approximation methods developed later.

\begin{remark}
Fractional terms can also be handled by this framework. For example, a ratio \(A(\bm{x})/B(\bm{x})\) can be written as \(A(\bm{x})B(\bm{x})^{-1}\). More generally, product-ratio terms can be handled by treating numerator functions and reciprocal denominator functions as positive coupled components.
\end{remark}

\begin{remark}
The proposed bounds do not automatically guarantee convexity or concavity of the transformed surrogate. Therefore, the resulting algorithm is a standard successive convex approximation method only when the transformed upper-bound surrogate is convex for minimization, or the transformed lower-bound surrogate is concave for maximization. Otherwise, the same tight surrogate can be used in a gradient-based successive approximation method.
\end{remark}

\begin{remark}[Relationship to geometric programming]
The proposed transforms are related to product-structured optimization, but they are different from classical geometric programming (GP). GP applies to problems whose objectives and constraints are composed of monomials and posynomials, and it can be converted equivalently into a convex problem through a logarithmic change of variables. Therefore, when the considered problem satisfies the monomial/posynomial structure required by GP, GP provides an exact and efficient solution method. In contrast, this paper considers general positive component functions, such as \(A_n(\bm{x})\), \(B_n(\bm{x})\), and \(f_n^{(k)}(\bm{x})\), which may represent rate, delay, reliability, utility, queueing, or learning-loss functions and are not necessarily monomials or posynomials. In these cases, the standard GP transformation is not directly applicable. The proposed HM-GM-AM-QM framework instead constructs tight inequality-induced surrogates and solves the resulting problem through exact or gradient-based successive approximation. Thus, GP can be viewed as an exact convex reformulation for a special algebraic class, while the proposed method targets broader product-structured non-convex optimization problems.
\end{remark}

With the equality condition \eqref{eq:y_update_k_no_phi}, we can obtain the closed-form expression of $y^{(k)}_n$ as
    \begin{equation}
         y^{(1)}_n = \sqrt{\frac{f^{(2)}_n(\bm{x})}{f^{(1)}_n(\bm{x})}},\quad
        y^{(2)}_n = (y^{(1)}_n)^{\frac{1}{3}}\cdot \left(\frac{f^{(3)}_n(\bm{x})}{f^{(2)}_n(\bm{x})}\right)^{\frac{1}{3}},\nonumber
    \end{equation}
    \begin{equation}
        y^{(k-1)}_n=(y_n^{(1)})^{\prod_{i=1}^{k-2} \frac{i}{i+2}} \cdot \left(\prod_{i=2}^{k-2} \left(\frac{f_n^{(i+1)}(\bm{x})}{f_n^{(i)}(\bm{x})}\right)^{\frac{1}{i+1}\cdot \prod_{j=i+2}^k \frac{j-2}{j}}\right) \cdot \left(\frac{f_n^{(k)}(\bm{x})}{f_n^{(k-1)}(\bm{x})}\right)^{\frac{1}{k}}, \forall k\in \{4,\cdots,K\}.\label{y_n_closedform}
    \end{equation}
Based on the above Equation (\ref{y_n_closedform}), we can get the closed form of the optimal $\bm{y}_{n,\ast}$. Note that this $\bm{y}_{n,\ast}$ is the global optimum of the HM/AM/QM bound with fixed $\bm{x}$. The auxiliary variables do not introduce approximation error when they are updated according to the equality condition (\ref{eq:y_update_k_no_phi}). Instead, they encode the current-iterate information needed to construct a tight upper-bound surrogate. This property is very beneficial for us to conduct the SA procedure to find a good solution to the sum optimization $\min/\max\limits_{\bm{x}} G(\bm{x}) + \sum_{n=1}^N \left(\prod_{k=1}^K f^{(k)}_n(\bm{x})\right)$.


\section{Successive Approximation Techniques to Solve Sum-of-Products Optimization by Using the HM, AM, and QM Surrogates}\label{sec_AMbound_SCA}

In this section, we present how the proposed surrogates can be used in successive approximation (SA) techniques to find a stationary point of the sum-of-products minimization and maximization problems with arbitrary multiplicative terms.

\subsection{Sum-of-Products Minimization}\label{sec_sop_min}
\textbf{Problem Statement}:
We consider $f^{(k)}_n(\bm{x}): \mathbb{R}^M \rightarrow \mathbb{R}_{++}$, where $k \in \mathcal{K}$, $\bm{x}$ is within a compact convex set $\mathcal{X}$, and a general coupled multiplicative term $\prod_{k=1}^K f^{(k)}_n(\bm{x})$ under the sum-of-products minimization problem:
\begin{subequations}\label{prob3}
\begin{align}
\min\limits_{\bm{x}} &\quad J(\bm{x}) + \sum_{n=1}^N \left(\prod_{k=1}^K f^{(k)}_n(\bm{x})\right)
\\
\text{s.t.} &\quad  \bm{x} \in \mathcal{X},
\end{align}
\end{subequations}
where $J(\bm{x}): \mathbb{R}^M \rightarrow \mathbb{R}$, and $J(\bm{x})$ is convex. Denote the objective function as $\Phi(\bm{x})$. This sum of multiplicative optimization is generally non-convex and NP-complete. 
\begin{remark}
Note that $J(\bm{x})$ need not be convex for the transformation itself to remain valid. When possible, we may apply additional techniques to reformulate a non-convex $J(\bm{x})$ as a convex term, and then optimize the resulting transformed problem. To keep the discussion focused and avoid unnecessary complications, we assume in what follows that $J(\bm{x})$ is convex.
\end{remark}
Based on Equation (\ref{bounds_K}), the AM upper bound of the sum of multiplicative terms in Problem (\ref{prob3}) is 
\begin{align}
    \frac{\left(f^{(1)}_n(\bm{x})\right)^K \cdot \prod_{k=1}^{K-1} y^{(k)}_n + \frac{\left(f^{(K)}_n(\bm{x})\right)^K}{\left(y^{(K-1)}_n\right)^{K-1}} +\sum_{k=2}^{K-1} \left(f^{(k)}_n(\bm{x})\right)^K \frac{\prod_{i=k}^{K-1} y^{(i)}_n}{\left(y^{(k-1)}_n\right)^{k-1}}}{K}.\nonumber 
\end{align}
Based on the AM upper bound $F_n^{AM}(\bm{x},\bm{y}_n)$, we can convert the original optimization (\ref{prob3}) to the following optimization:
\begin{subequations}\label{prob4}
\begin{align}
\min\limits_{\bm{x},\bm{y}} &\quad J(\bm{x}) + \sum_{n=1}^N F_n^{AM}(\bm{x},\bm{y}_n)
\\
\text{s.t.} &\quad  \bm{x} \in \mathcal{X}.
\end{align}
\end{subequations}
\begin{remark}
When the AM upper bound is convex, the proposed SA method reduces to an SCA method. The QM upper bound can also be used to transform problem (\ref{prob3}). Similar to the AM upper bound, it provides a tight upper-bound surrogate at every iterate. Therefore, solving the QM surrogate exactly yields a descent update for the original minimization problem, and solving it approximately leads to the corresponding gradient-based SA method. If the QM surrogate is convex, the resulting update is a standard SCA step. Since the AM-GM and QM-GM transforms follow the same surrogate approximation principle, we use the AM-GM transform as the representative example here.
\end{remark}

Before introducing the proposed SA method, we first present some useful properties of the transformed objective function.

\begin{proposition}\label{prop_property}
Let
\begin{equation}
    \Phi(\bm{x})
    =
    J(\bm{x})
    +
    \sum_{n=1}^{N}
    \prod_{k=1}^{K} f_n^{(k)}(\bm{x}),
\end{equation}
and define the AM-based surrogate anchored at a feasible point $\bm{x}^\prime$ as
\begin{equation}
    Q(\bm{x};\bm{x}^\prime)
    =
    J(\bm{x})
    +
    \sum_{n=1}^{N}
    F_n^{\rm AM}(\bm{x};\bm{y}_n(\bm{x}^\prime)),
\end{equation}
where $\bm{y}_n(\bm{x}^\prime)$ is obtained by Equation (\ref{y_n_closedform}) of the AM bound at $\bm{x}^\prime$. Suppose that all involved functions are positive and differentiable over the feasible region. Then, for any feasible points $\bm{x}$ and $\bm{x}^\prime$, we have
\begin{align}
    &\Phi(\bm{x}) \leq Q(\bm{x};\bm{x}^\prime),\\
    &\Phi(\bm{x}^\prime) = Q(\bm{x}^\prime;\bm{x}^\prime),\\
    &\nabla_{\bm{x}}\Phi(\bm{x}^\prime)
    =
    \nabla_{\bm{x}}Q(\bm{x}^\prime;\bm{x}^\prime).
\end{align}
\end{proposition}

\begin{proof}
    Please refer to Appendix \ref{append_proof_prop_property}.
\end{proof}
\begin{remark}
    It's obvious that Proposition \ref{prop_property} also holds for HM and QM bounds.
\end{remark}
Next, we will analyze the cases where the transformed problems are convex and non-convex, respectively. We then present that if the gradient-based SA method is utilized, it is applicable to both of these situations.

\subsubsection{Cases where the Transformed Function is Convex}\label{sec_AM_convex}

We next discuss when the AM-based surrogate becomes convex with respect to the optimization variable. This condition is not required for constructing the AM upper bound itself. But it identifies the cases where the proposed SA update reduces to a standard SCA update.

For fixed auxiliary variables $\bm{y}_n$, the AM upper bound can be written as a nonnegative weighted sum of terms in the form
$\left(f_n^{(k)}(\bm{x})\right)^K$.
Therefore, a sufficient condition for $F_n^{\rm AM}(\bm{x},\bm{y}_n)$ to be convex is that
$\left(f_n^{(k)}(\bm{x})\right)^K$
is convex for all $n$ and $k$. This condition is satisfied, for example, when $f_n^{(k)}(\bm{x})$ is nonnegative and convex over the feasible set, since $t^K$ is convex and nondecreasing over $t\ge 0$, and the composition of a convex nondecreasing function with a convex function is convex. Thus, this condition is not introduced merely for analytical convenience. It characterizes a practically relevant class of problems for which the transformed AM surrogate is convex and can be solved by an SCA step.

Several application examples satisfying this condition are summarized in Table~\ref{tab:convex_am_comm_examples}. In these cases, the AM upper bound gives a convex surrogate for minimization. If this condition is not satisfied, the AM bound remains a valid tight upper-bound surrogate, but the resulting subproblem may not be convex. In that case, we use the gradient-based SA method described in Section~\ref{sec_am_nonconvex}.

\begin{table*}[t]
\centering
\footnotesize
\setlength{\tabcolsep}{3.5pt}
\renewcommand{\arraystretch}{1.18}
\caption{Representative optimization terms where the AM surrogate is convex.}
\label{tab:convex_am_comm_examples}
\begin{tabular}{p{0.20\textwidth}|p{0.25\textwidth}|p{0.25\textwidth}|p{0.25\textwidth}}
\toprule
Application & Product term & Component function & Why $\left(f_n^{(k)}(x)\right)^K$ is convex \\
\midrule

Cooperative edge caching &
Cache-miss probability:\newline
$\prod_{m\in\mathcal{A}_u}(1-q_{m,f})$ &
$f_{u,f}^{(m)}(q)=1-q_{m,f}$, with\newline
$0\le q_{m,f}\le 1-\epsilon_q$ &
Positive affine function; its $K$-th power is convex on the nonnegative domain. \\
\hline

Multi-connectivity outage minimization &
All-link failure probability:\newline
$\prod_{m\in\mathcal{A}_u} e^{-a_{u,m}p_{u,m}}$ &
$f_{u,m}(p)=e^{-a_{u,m}p_{u,m}}$,
$a_{u,m}>0$ &
$\left(f_{u,m}(p)\right)^K=e^{-K a_{u,m}p_{u,m}}$ is exponential and convex. \\
\hline

Multi-hop packet-loss minimization &
End-to-end loss risk:\newline
$\prod_{\ell\in\mathcal{L}_r} e^{-a_{\ell}p_{\ell}}$ &
$f_{\ell}(p)=e^{-a_{\ell}p_{\ell}}$,
$p_{\ell}\ge 0$ &
The power of an exponential-affine function remains convex. \\
\hline

Distributed sensing miss-detection minimization &
Joint miss-detection probability:\newline
$\prod_{s\in\mathcal{S}} e^{-b_s \tau_s}$ &
$f_s(\tau)=e^{-b_s\tau_s}$,
$\tau_s\ge 0$ &
$\left(e^{-b_s\tau_s}\right)^K=e^{-K b_s\tau_s}$ is convex. \\
\hline

Freshness-risk minimization &
Age-risk product:\newline
$\Delta_u(x)e^{-a_up_u}$ &
$f_u^{(1)}(x)=\Delta_u(x)$,
$f_u^{(2)}(p)=e^{-a_up_u}$ &
If $\Delta_u(x)$ is positive affine or convex, then $\Delta_u(x)^K$ is convex; the exponential term is also convex. \\
\hline

Service-margin risk minimization &
Queueing-risk product:\newline
$\prod_{j\in\mathcal{J}_u}\frac{1}{\mu_j-r_j}$ &
$f_j(r)=\frac{1}{\mu_j-r_j}$,
$0\le r_j<\mu_j$ &
$\left(\mu_j-r_j\right)^{-K}$ is convex over $r_j<\mu_j$. \\
\hline

Energy-delay product minimization &
Coupled cost:\newline
$E_u(x)D_u(x)$ &
$f_u^{(1)}(x)=E_u(x)$,
$f_u^{(2)}(x)=D_u(x)$ &
If $E_u(x)$ and $D_u(x)$ are positive affine or convex cost models, their powers are convex. \\
\hline

Computation-offloading failure minimization &
Joint failure risk:\newline
$e^{-a_up_u}e^{-b_uc_u}$ &
$f_u^{(1)}(p)=e^{-a_up_u}$,
$f_u^{(2)}(c)=e^{-b_uc_u}$ &
Both powered terms remain exponential-affine and convex. \\

\bottomrule
\end{tabular}
\end{table*}

\begin{remark}
    The QM upper bound can be analyzed similarly. If the resulting QM surrogate is convex, the corresponding SA update is also an SCA update. Otherwise, the QM surrogate can still be used within the proposed gradient-based SA framework.
\end{remark}


\textbf{Proposed SCA Method to Solve Problem (\ref{prob4}):}
Based on the AM upper bound $F_n^{AM}(\bm{x},\bm{y}_n)$, we can convert the original problem (\ref{prob3}) to the problem (\ref{prob4}). If we fix $\bm{y}$, this transformed optimization is a convex optimization over $\bm{x}$. Thus, we can use a novel SCA method to solve it and find a stationary point solution to the original optimization. The SCA algorithm is detailed in Algorithm \ref{algo:SCA}. Since the QM upper bound $F_n^{QM}(\bm{x},\bm{y}_n)$ can also be used to solve the optimization (\ref{prob4}), we add it in Algorithm \ref{algo:SCA}. For the sake of simplicity, we denote $\bm{y}(\bm{x}^i)$ as $\bm{y}^i$.

Note that this novel SCA algorithm is like the BCD algorithm, approximating the stationary point under the SCA procedure. The current-iterate information is contained in $\bm{y}$. Therefore, Algorithm \ref{algo:SCA} is actually an SCA method, as we illustrated in Section \ref{sec_infactSCA}.

\begin{algorithm}
\caption{A Novel SCA Method to Solve \mbox{Problem (\ref{prob4})}.}
\label{algo:SCA}

Initialize $i \leftarrow -1$ and a feasible point $\bm{x}^{0}$;

Obtain the AM upper bound $F_n^{AM}(\bm{x};\bm{y}_n^0)$ or the QM upper bound $F_n^{QM}(\bm{x};\bm{y}_n^0)$ by (\ref{bounds_K});

Replace $\prod_{k=1}^K f^{(k)}_n(\bm{x})$ by $F_n^{AM}(\bm{x};\bm{y}_n^0)$ or $F_n^{QM}(\bm{x};\bm{y}_n^0)$;

\Repeat{the value of function in optimization (\ref{prob4}) converges}{
Let $i \leftarrow i+1$;

Update $\bm{y}^{i}$ by (\ref{y_n_closedform}) with the feasible point $\bm{x}^{i}$;

Update $\bm{x}^{i+1}\in\arg \min\limits_{\bm{x}\in\mathcal{X}} \left(J\left(\bm{x}\right) + \sum_{n=1}^N F_n^{AM}\left(\bm{x};\bm{y}_n^{i}\right)\right)$ or 
$\bm{x}^{i+1}\in\arg \min\limits_{\bm{x}\in\mathcal{X}} \left(J\left(\bm{x}\right) + \sum_{n=1}^N F_n^{QM}\left(\bm{x};\bm{y}_n^{i}\right)\right)$
;
}
\end{algorithm}

\textbf{Convergence Analysis:}
We will show that Algorithm \ref{algo:SCA} converges to a stationary point of Problem (\ref{prob4}). Before proving this, we first recall the convergence lemma of the maximum block improvement (MBI) in~\cite{chen2012maximum}.

\begin{lemma}
\label{lemma_mbi}
Consider the following block optimization problem
\begin{equation}
    \min_{\bm{z}_1,\ldots,\bm{z}_d}
    \Psi(\bm{z}_1,\ldots,\bm{z}_d)
    \quad
    \emph{s.t.}\quad
    \bm{z}_i\in \mathcal{S}_i,\quad i=1,\ldots,d,
\end{equation}
where $\Psi$ is continuous and each $\mathcal{S}_i$ is compact. At iteration $k$, for each block
$i$, define the best block response
\begin{equation}
    \widehat{\bm{z}}_i^{k+1}
    \in
    \arg\min_{\bm{z}_i\in\mathcal{S}_i}
    \Psi
    \left(
    \bm{z}_1^k,\ldots,\bm{z}_{i-1}^k,
    \bm{z}_i,
    \bm{z}_{i+1}^k,\ldots,\bm{z}_d^k
    \right),
\end{equation}
and the corresponding objective value
\begin{equation}
    v_i^{k+1}
    =
    \Psi
    \left(
    \bm{z}_1^k,\ldots,\bm{z}_{i-1}^k,
    \widehat{\bm{z}}_i^{k+1},
    \bm{z}_{i+1}^k,\ldots,\bm{z}_d^k
    \right).
\end{equation}
Let
    $i_k
    \in
    \arg\min_{1\le i\le d}
    v_i^{k+1}$, and the MBI update is given by
    $\bm{z}_{i_k}^{k+1}=\widehat{\bm{z}}_{i_k}^{k+1},
    \bm{z}_{i}^{k+1}=\bm{z}_{i}^{k}$, for $i\ne i_k$.
Then any cluster point
    $\bm{z}^*
    =
    \left(
    \bm{z}_1^*,\ldots,\bm{z}_d^*
    \right)$
of the generated sequence is a blockwise stationary point, i.e.,
\begin{equation}
    \bm{z}_i^*
    \in
    \arg\min_{\bm{z}_i\in\mathcal{S}_i}
    \Psi
    \left(
    \bm{z}_1^*,\ldots,\bm{z}_{i-1}^*,
    \bm{z}_i,
    \bm{z}_{i+1}^*,\ldots,\bm{z}_d^*
    \right),
    \quad
    i=1,\ldots,d.
\end{equation}
\end{lemma}
The original MBI result in~\cite{chen2012maximum} is stated for maximization, and the above lemma follows by applying the result to minimization. The core idea of Algorithm \ref{algo:SCA} can be interpreted as an exact two-block MBI procedure.

\begin{theorem}[Convergence of the proposed AM/QM-SCA method]
\label{theorem_sca_convergence}
Suppose that $\mathcal{X}$ is compact and convex, $J(\bm{x})$ is convex, and that $J(\bm{x})$ and
$f_n^{(k)}(\bm{x})$ are continuously differentiable on an open set containing
$\mathcal{X}$, with $f_n^{(k)}(\bm{x})>0$ for all $\bm{x}\in\mathcal{X}$.
Assume that, at each iteration, the auxiliary variables $\bm{y}$ are updated according
to Equation (\ref{y_n_closedform}), and the resulting transformed
surrogate is convex over $\mathcal{X}$ and solved exactly. Then the objective
sequence $\{\Phi(\bm{x}^{i})\}$ generated by Algorithm~\ref{algo:SCA} is
non-increasing and convergent. Moreover, every limit point of
$\{\bm{x}^{i}\}$ is a stationary point of Problem~\eqref{prob3}.
\end{theorem}

\begin{proof}
    Please refer to Appendix \ref{appdix_proof_theorem_sca_convergence}.
\end{proof}

\textbf{Complexity Analysis:}
Updating $\bm{y}$ takes $NK$ operations. As for updating $\bm{x}$, the actual operations depend on the way to derive $\bm{x}$, e.g., closed-form solution of $\bm{x}$, CVX solver. We denote the complexity of computing $\bm{x}$ at one particular iteration as $\mathcal{C}_{\bm{x}}$. Assuming that it takes $I$ iterations in total, the complexity of Algorithm \ref{algo:SCA} is $ \mathcal{O}(INK + I \mathcal{C}_{\bm{x}})$. 


\subsubsection{Cases where the Transformed Function is Non-Convex}\label{sec_am_nonconvex}
We then present how to use the SA method with gradient descent to solve the optimization (\ref{prob4}) when the transformed objective function is non-convex. Moreover, in Algorithm \ref{algo:SCA}, the point $\bm{x}^{i+1}$ is obtained by solving $\bm{x}^{i+1}\in\arg \min\limits_{\bm{x}\in\mathcal{X}} \left(J(\bm{x}) + \sum_{n=1}^N F_n^{AM}(\bm{x};\bm{y}_n^i)\right)$, which can be computationally demanding. To reduce the computational burden, we can replace this minimization with several gradient steps on the objective function.

\textbf{Proposed Efficient Gradient Descent SA Method to Solve Problem (\ref{prob4}):}
We first denote $Q(\bm{x};\bm{x}^i):=J(\bm{x}) + \sum_{n=1}^N F_n^{AM}(\bm{x};\bm{y}_n^i)$ and we will use this notation for the transformed objective function of the HM, AM, or QM bound in the remaining part of this paper. For each outer iteration $i$, perform $M_i \geq 1$ gradient steps:
\begin{equation}
    \bm{x}_{j+1}^i = \Pi_{\mathcal{X}}
    \left(\bm{x}_j^i - \alpha_{i,j} \nabla Q(\bm{x};\bm{x}_j^i)\right), j=0,1,\cdots,M_i-1,
\end{equation}
where $\alpha_{i,j}$ is the step size of the outer loop $i$ and the inner loop $j$, and $\Pi_{\mathcal{X}}(\cdot)$ denotes the Euclidean projection onto $\mathcal{X}$. Each inner step only needs a gradient evaluation of the transformed surrogate $Q(\bm{x};\bm{x}^i_j)$, rather than solving the whole convex subproblem exactly. Thus, the proposed efficient gradient-descent SA method is shown in Algorithm \ref{algo:SCA_gd}.

\begin{algorithm}[t]
\caption{An Efficient Gradient-based SA Method to Solve \mbox{Problem (\ref{prob4})}.}
\label{algo:SCA_gd}

Initialize $i \leftarrow -1$ and a feasible point $\bm{x}^{0}$;

Obtain the AM upper bound $F_n^{AM}(\bm{x};\bm{y}_n^{0})$ or the QM upper bound $F_n^{QM}(\bm{x};\bm{y}_n^0)$ by (\ref{bounds_K});

Replace $\prod_{k=1}^K f^{(k)}_n(\bm{x})$ by $F_n^{AM}(\bm{x};\bm{y}_n^0)$ or $F_n^{QM}(\bm{x};\bm{y}_n^0)$ and construct the AM/QM surrogate $Q(\bm{x};\bm{x}^0)$;

\Repeat{the value of function in optimization (\ref{prob4}) converges}{
Let $i \leftarrow i+1$;

Update $\bm{y}^{i}$ by (\ref{y_n_closedform}) with the feasible point $\bm{x}^{i}$;

Initialize $j = -1$ and set $\bm{x}_0^{i} = \bm{x}^{i}$;

\Repeat{$j = M_i-1$}{
Let $j \leftarrow j+1$;

Update $\bm{x}_{j+1}^{i} = \Pi_{\mathcal{X}}
    \left(\bm{x}_j^{i} - \alpha_{i,j} \nabla Q(\bm{x};\bm{x}_j^{i})\right)$;
}

Update $\bm{x}^{i+1} = \bm{x}_{M_i}^{i}$;
}
\end{algorithm}

\textbf{Convergence Analysis:}
Based on the AM/QM-GM transform and Proposition \ref{prop_property},
we obtain the following theorem:
\begin{theorem}\label{theorem_sca_gd}
    Suppose that the original objective $\Phi(\bm{x})$ is bounded below, the functions $J(\bm{x})$ and $f_n^{(k)}(\bm{x})$ are continuously differentiable, the surrogate $Q(\bm{x})$ is $L$-smooth, i.e., 
    \begin{equation}
        \|\nabla Q(\bm{x})-\nabla Q(\bm{x}^\prime)\|\leq L\|\bm{x} - \bm{x}^\prime\|, \forall \bm{x}, \bm{x}^\prime,
    \end{equation}
    and $0<\alpha_{i,j}\leq 1/L$. The iterates generated by Algorithm \ref{algo:SCA_gd} satisfy
    \begin{equation}
        \Phi(\bm{x}^{i+1})\leq \Phi(\bm{x}^i) - \frac{\alpha_{i,0}}{2}\left\|\nabla \Phi(\bm{x}^i)\right\|^2.
    \end{equation}
    Consequently, 
    \begin{equation}
        \min_{0\leq i \leq I-1} \|\nabla \Phi(\bm{x}^i)\|^2 \leq \frac{2\left(\Phi(\bm{x}^0)-\Phi_{\text{inf}}\right)}{\sum_{i=0}^{I-1}\alpha_{i,0}},
    \end{equation}
    where $\Phi_{\text{inf}}$ is the lower bound of $\Phi(\bm{x})$. In particular, if $\alpha_{i,0} = \alpha$ for all $i$, then 
    \begin{equation}
        \min_{0\leq i \leq I-1} \|\nabla \Phi(\bm{x}^i)\|^2 \leq \frac{2\left(\Phi(\bm{x}^0)-\Phi_{\text{inf}}\right)}{\alpha I}.
    \end{equation}
    Therefore, an $\epsilon$-stationary point can be obtained after at most 
    \begin{equation}
        I \geq \frac{2\left(\Phi(\bm{x}^0)-\Phi_{\text{inf}}\right)}{\alpha \epsilon^2}
    \end{equation}
    outer iterations.
\end{theorem}
\begin{proof}
    Please refer to Appendix \ref{appdix_proof_theorem_sca_gd}.
\end{proof}

\textbf{Complexity Analysis:}
Updating $\bm{y}$ takes $NK$ operations. As for updating $\bm{x}$, suppose that $M_i$ gradient steps are performed at the $i$-th outer iterations, and let the complexity of computing one gradient with respect to $\bm{x}$ be denoted by $\mathcal{C}_{\nabla x}$. Then the complexity of updating $\bm{x}$ at the $i$-th outer iteration is $M_i \mathcal{C}_{\nabla x}$. Let $I$ denote the number of outer iterations required to obtain an $\epsilon$-stationary point. Since $I = \mathcal{O}(\epsilon^{-2})$ based on the \textbf{Theorem \ref{theorem_sca_gd}}, the total complexity of the proposed gradient-based SA method is $\mathcal{O}\left(\epsilon^{-2}NK+\sum_{i=1}^I M_i \mathcal{C}_{\nabla x}\right)$. If a fixed number $M$ of gradient steps is used at each outer iteration, this complexity becomes $\mathcal{O}\left(\frac{NK+M\mathcal{C}_{\nabla x}}{\epsilon^2}\right)$. In particular, when only one gradient step is used for updating $\bm{x}$ at each outer iteration, the complexity reduces to $\mathcal{O}\left(\frac{NK+\mathcal{C}_{\nabla x}}{\epsilon^2}\right)$.

\subsection{Sum-of-Products Maximization}\label{sec_max_sca}
\textbf{Problem Statement}:
We consider $f^{(k)}_n(\bm{x}): \mathbb{R}^M \rightarrow \mathbb{R}_{++}$, where $k \in \mathcal{K}$, $\bm{x}$ is within a compact convex set $\mathcal{X}$, and a general coupled multiplicative term $\prod_{k=1}^K f^{(k)}_n(\bm{x})$ under the sum-of-products maximization problem:
\begin{subequations}\label{prob12}
\begin{align}
\max\limits_{\bm{x}} &\quad J(\bm{x}) + \sum_{n=1}^N \left(\prod_{k=1}^K f^{(k)}_n(\bm{x})\right)
\\
\text{s.t.} &\quad  \bm{x} \in \mathcal{X},
\end{align}
\end{subequations}
where $J(\bm{x}): \mathbb{R}^M \rightarrow \mathbb{R}$, and $J(\bm{x})$ is concave. Denote the objective function as $\Phi(\bm{x})$. This sum of multiplicative optimization is also generally non-convex and NP-complete. From Equation (\ref{bounds_K}), we get the HM lower bound of multiplicative terms in Problem (\ref{prob12}) as
\begin{equation}
    \frac{K}{\frac{1}{f^{(1)}_n(\bm{x}) \cdot \prod_{k=1}^{K-1} y^{(k)}_n} + \frac{1}{\frac{f^{(K)}_n(\bm{x})}{\left(y^{(K-1)}_n\right)^{K-1}}} + \sum_{k=2}^{K-1}\frac{1}{f^{(k)}_n(\bm{x}) \cdot \frac{\prod_{i=k}^{K-1} y^{(i)}_n}{\left(y^{(k-1)}_n\right)^{k-1}}}},
\end{equation}
which is denoted as $F_n^{HM}(\bm{x},\bm{y}_n)$. Based on the HM lower bound $F_n^{HM}(\bm{x},\bm{y}_n)$, we can convert the original optimization (\ref{prob12}) to the following optimization:
\begin{subequations}\label{prob13}
\begin{align}
\max\limits_{\bm{x},\bm{y}} &\quad J(\bm{x}) + \sum_{n=1}^N F_n^{HM}(\bm{x},\bm{y}_n)
\\
\text{s.t.} &\quad  \bm{x} \in \mathcal{X}.
\end{align}
\end{subequations}
The relevant theoretical analysis in this part is similar to that under the minimization problem mentioned earlier, so we will not elaborate on it here. We directly present the standard SA algorithm and the gradient-based SA algorithm by using the HM-GM transform as follows. Based on Algorithms \ref{algo:SCA_HM} and \ref{algo:SCA_HM_gd}, we can obtain the stationary point or $\epsilon$-stationary point of the original maximization problem (\ref{prob12}).

\begin{algorithm}
\caption{A Novel SA Method to Solve \mbox{Problem (\ref{prob13})}.}
\label{algo:SCA_HM}

Initialize $i \leftarrow -1$ and a feasible point $\bm{x}^{0}$;

Obtain the HM lower bound $F_n^{HM}(\bm{x};\bm{y}_n^0)$ by (\ref{bounds_K});

Replace $\prod_{k=1}^K f^{(k)}_n(\bm{x})$ by $F_n^{HM}(\bm{x};\bm{y}_n^0)$;

\Repeat{the value of function in optimization (\ref{prob13}) converges}{
Let $i \leftarrow i+1$;

Update $\bm{y}^{i}$ by (\ref{y_n_closedform}) with the feasible point $\bm{x}^{i}$;

Update $\bm{x}^{i+1}\in\arg \max\limits_{\bm{x}\in\mathcal{X}} \left(J\left(\bm{x}\right) + \sum_{n=1}^N F_n^{HM}\left(\bm{x};\bm{y}_n^{i}\right)\right)$
;
}
\end{algorithm}

\begin{algorithm}
\caption{An Efficient Gradient-based SA Method to Solve \mbox{Problem (\ref{prob13})}.}
\label{algo:SCA_HM_gd}

Initialize $i \leftarrow -1$ and a feasible point $\bm{x}^{0}$;

Obtain the HM upper bound $F_n^{HM}(\bm{x};\bm{y}_n^{0})$ by (\ref{bounds_K});

Replace $\prod_{k=1}^K f^{(k)}_n(\bm{x})$ by $F_n^{HM}(\bm{x};\bm{y}_n^{0})$ and construct the HM surrogate $Q(\bm{x};\bm{x}^{0})$;

\Repeat{the value of function in optimization (\ref{prob13}) converges}{
Let $i \leftarrow i+1$;

Update $\bm{y}^{i}$ by (\ref{y_n_closedform}) with the feasible point $\bm{x}^{i}$;

Initialize $j = -1$ and set $\bm{x}_0^{i} = \bm{x}^{i}$;

\Repeat{$j = M_i-1$}{
Let $j \leftarrow j+1$;

Update $\bm{x}_{j+1}^{i} = \Pi_{\mathcal{X}}
    \left(\bm{x}_j^{i} - \alpha_{i,j} \nabla Q(\bm{x};\bm{x}_j^{i})\right)$;
}

Update $\bm{x}^{i+1} = \bm{x}_{M_i}^{i}$;
}
\end{algorithm}

\begin{table*}[t]
\centering
\footnotesize
\setlength{\tabcolsep}{5pt}
\renewcommand{\arraystretch}{1.18}
\caption{Summary of the proposed SA methods.}
\label{tab:sca_summary}

\begin{tabularx}{\textwidth}{lXX}
\toprule
Method & Studied problem & Surrogate at feasible point $\bm{x}^\prime$ \\
\midrule
(Gradient) AM/QM-SA &
$\min\limits_{\bm{x}\in\mathcal{X}} J(\bm{x}) + \sum_{n=1}^N \left(\prod_{k=1}^K f^{(k)}_n(\bm{x})\right)$ &
$\min\limits_{\bm{x}\in\mathcal{X}} J(\bm{x})+\sum_{n=1}^{N}F_n^{\rm AM/QM}(\bm{x};\bm{y}_n(\bm{x}^\prime))$ \\
(Gradient) HM-SA &
$\max\limits_{\bm{x}\in\mathcal{X}} J(\bm{x}) + \sum_{n=1}^N \left(\prod_{k=1}^K f^{(k)}_n(\bm{x})\right)$ &
$\max\limits_{\bm{x}\in\mathcal{X}} J(\bm{x})+\sum_{n=1}^{N}F_n^{\rm HM}(\bm{x};\bm{y}_n(\bm{x}^\prime))$ \\
\bottomrule
\end{tabularx}

\vspace{1.5mm}

\begin{tabularx}{\textwidth}{lXXX}
\toprule
Method & Solution & Convergence rate & Complexity \\
\midrule
AM/QM-SA &
Stationary point &
-- &
$\mathcal{O}(INK+I\mathcal{C}_{\bm{x}})$ \\

Gradient AM/QM-SA &
$\epsilon$-stationary point &
Sublinear &
$\mathcal{O}\!\left((NK+M\mathcal{C}_{\nabla\bm{x}})/\epsilon^2\right)$ \\

HM-SA &
Stationary point &
-- &
$\mathcal{O}(INK+I\mathcal{C}_{\bm{x}})$ \\

Gradient HM-SA &
$\epsilon$-stationary point &
Sublinear &
$\mathcal{O}\!\left((NK+M\mathcal{C}_{\nabla\bm{x}})/\epsilon^2\right)$ \\
\bottomrule
\end{tabularx}
\end{table*}
\subsection{Summary of SA Methods based on HM-GM-AM-QM Transforms}\label{sec_sum_sa}
Table~\ref{tab:sca_summary} summarizes the proposed SA methods based on different inequality-induced surrogates. For minimization problems, the GM-AM and GM-QM inequalities provide upper-bound surrogates, which lead to the AM/QM-SA methods. For maximization problems, the HM-GM inequality provides a lower-bound surrogate, which leads to the HM-SA method. In the exact SA methods, the transformed subproblem with respect to $\bm{x}$ is solved at each iteration, and the resulting solution is guaranteed to converge to a stationary point under standard SA conditions. However, the per-iteration complexity depends on the cost $\mathcal{C}_{\bm{x}}$ of solving the transformed subproblem.

To improve computational efficiency and solve optimization problems when the transformed objective function is non-convex, we further propose gradient-based SA variants, where the exact update of $\bm{x}$ is replaced by several gradient steps on the corresponding surrogate. These methods avoid repeatedly solving the transformed subproblem exactly and only require gradient evaluations. Under smoothness, boundedness, tightness, and first-order consistency conditions, the gradient-based AM/QM-SA and HM-SA methods converge to an $\epsilon$-stationary point with sublinear convergence rate and iteration complexity $\mathcal{O}(\epsilon^{-2})$. Therefore, the gradient-based variants provide a more efficient implementation of the proposed inequality-based SA framework, especially when the transformed subproblem is difficult to solve exactly.

\textbf{Role of the proposed SA/SCA framework.}
The proposed framework is not intended to replace standard SCA or to universally outperform direct gradient descent/ascent on every smooth instance. Direct gradient methods can also approach stationary points under suitable smoothness and stepsize conditions. However, they treat the original multiplicative or fractional objective as a generic non-convex function and do not exploit its product structure. In contrast, the proposed HM-GM-AM-QM transforms provide a systematic way to construct product-aware surrogates with the correct bound direction, tightness, and first-order consistency. For minimization problems, the AM/QM upper-bound surrogates yield descent-oriented updates; for maximization problems, the HM lower-bound surrogate yields ascent-oriented updates. More importantly, when the transformed surrogate is convex or concave, the update becomes a standard SCA step and can exploit convex solvers, KKT conditions, or even closed-form updates, which is beyond a single direct gradient step on the original non-convex objective. The proposed transforms can also be applied to product-type constraints, where direct projected gradient would require projection onto a generally non-convex feasible set, while the proposed method can construct conservative convex surrogate constraints. Therefore, the main contribution is not that the proposed method is always faster than direct gradient descent, but that it provides a unified, structure-exploiting, and convergence-analyzable surrogate construction for coupled product, ratio, logarithmic product-ratio, and product-constraint optimization problems.

\subsection{Extension to Logarithmic Problems}\label{sec_extension_log}
The proposed inequality-based transform can also be extended to logarithmic objectives with general coupled terms. Let
\begin{equation}
    S_n(\bm{x})=\prod_{k=1}^{K} f_n^{(k)}(\bm{x})>0,
\end{equation}
where the functions \(f_n^{(k)}(\bm{x})\) may include reciprocal terms, so fractional structures are also covered. For the maximization problem
\begin{equation}
    \max_{\bm{x}\in\mathcal{X}}\; J(\bm{x})+\sum_{n=1}^{N}w_n\log\left(1+S_n(\bm{x})\right),
\end{equation}
where $w_n$ is a positive constant, 
we construct a lower-bound surrogate \(F_n^{\rm HM}(\bm{x},\bm{y}_n)\le S_n(\bm{x})\) by the HM-GM inequality. Since \(\log(1+z)\) is increasing, we have
\begin{equation}
    \log\left(1+F_n^{\rm HM}(\bm{x},\bm{y}_n)\right)
\le
\log\left(1+S_n(\bm{x})\right).
\end{equation}
Thus, the logarithmic objective can be minorized by
\begin{equation}
    R(\bm{x})=J(\bm{x})+\sum_{n=1}^{N}w_n
\log\left(1+F_n^{\rm HM}(\bm{x},\bm{y}_n)\right).
\end{equation}
Similarly, for minimization, an AM/QM upper bound
\begin{equation}
    S_n(\bm{x})\le F_n^{\rm AM/QM}(\bm{x},\bm{y}_n)
\end{equation}
yields the majorization
\begin{equation}
    Q(\bm{x})=J(\bm{x})+\sum_{n=1}^{N}w_n
\log\left(1+F_n^{\rm AM/QM}(\bm{x},\bm{y}_n)\right).
\end{equation}
Because the auxiliary variables $\bm{y}_n$ are updated according to Equation (\ref{y_n_closedform}) at one feasible point, these logarithmic surrogates are tight and first-order consistent at that feasible point. Hence, the resulting gradient-based SA method can be used to find an \(\epsilon\)-stationary point under standard smoothness and boundedness assumptions.

Different from the logarithmic quadratic transform, which moves the ratio outside the logarithm by an auxiliary variable, the above extension keeps the logarithm but applies the proposed inequality-based bound to the positive coupled term inside it. Therefore, it is most naturally suited to gradient-based SA methods. Convexity of the transformed logarithmic surrogate requires additional problem-specific conditions. Extending the proposed HM-GM-AM-QM transform to full matrix-ratio/products optimization remains an interesting direction for future work.

\subsection{Comparison to Existing Methods}\label{sec_comp_existingmethods}
\begin{table*}[t]
\centering
\small
\setlength{\tabcolsep}{5pt}
\renewcommand{\arraystretch}{1.22}
\caption{Representative FP transforms and their target problems.}
\label{tab:fp_target_problem}
\begin{tabularx}{\textwidth}{|l|X|}
\hline
Method & Target problem \\
\hline
Charnes--Cooper\cite{charnes1962programming} &
$\displaystyle \max_{\bm{x}\in\mathcal{X}} \frac{A(\bm{x})}{B(\bm{x})}$ \\
\hline

Dinkelbach\cite{dinkelbach1967nonlinear} &
$\displaystyle \max_{\bm{x}\in\mathcal{X}} \frac{A(\bm{x})}{B(\bm{x})}$, $\displaystyle \max_{\bm{x}\in\mathcal{X}}\min_n \frac{A_n(\bm{x})}{B_n(\bm{x})}$ \\
\hline

Quadratic transform\cite{shen2018fractional} &
$\displaystyle \max_{\bm{x}\in\mathcal{X}} \sum_{n=1}^{N}\frac{A_n(\bm{x})}{B_n(\bm{x})}$ \\
\hline

Inverse quadratic transform\cite{chen2023inverse} &
$\displaystyle \min_{\bm{x}\in\mathcal{X}} \sum_{n=1}^{N}\frac{A_n(\bm{x})}{B_n(\bm{x})}$ \\
\hline

\textbf{Proposed AM/QM-GM transform} &
$\displaystyle \min_{\bm{x}\in\mathcal{X}} J(\bm{x})+\sum_{n=1}^{N}\prod_{k=1}^{K} f_n^{(k)}(\bm{x})$ \\
\hline

\textbf{Proposed HM-GM transform} &
$\displaystyle \max_{\bm{x}\in\mathcal{X}} J(\bm{x})+\sum_{n=1}^{N}\prod_{k=1}^{K} f_n^{(k)}(\bm{x})$ \\
\hline
\end{tabularx}
\end{table*}

\begin{table*}[t]
\centering
\footnotesize
\setlength{\tabcolsep}{3.5pt}
\renewcommand{\arraystretch}{1.25}
\caption{Comparison of representative FP methods under different problem classes.}
\label{tab:fp_class_comparison}
\begin{tabularx}{\textwidth}{
>{\raggedright\arraybackslash}m{0.3\textwidth}|
>{\centering\arraybackslash}m{0.13\textwidth}|
>{\centering\arraybackslash}m{0.13\textwidth}|
>{\centering\arraybackslash}m{0.16\textwidth}|
>{\centering\arraybackslash}m{0.17\textwidth}}
\toprule
Problem class & Charnes--Cooper & Dinkelbach & (Inverse) Quadratic transform & \textbf{Proposed HM-GM-AM-QM transform} \\
\midrule
\makecell[l]{Single-ratio:
$\displaystyle \max_{\bm{x}\in\mathcal{X}} \frac{A(\bm{x})}{B(\bm{x})}$} &
Global optimum & Global optimum & Global optimum & Global optimum \\
\hline

\makecell[l]{Max-min-ratios:
$\displaystyle \max_{\bm{x}\in\mathcal{X}}\min_n \frac{A_n(\bm{x})}{B_n(\bm{x})}$} &
-- & Global optimum & Global optimum & Global optimum \\
\hline

\makecell[l]{Max sum-of-ratios:
$\displaystyle \max_{\bm{x}\in\mathcal{X}} \sum_{n=1}^{N}\frac{A_n(\bm{x})}{B_n(\bm{x})}$} &
-- & -- & Stationary point & Stationary point \\
\hline

\makecell[l]{Min sum-of-ratios:
$\displaystyle \min_{\bm{x}\in\mathcal{X}} \sum_{n=1}^{N}\frac{A_n(\bm{x})}{B_n(\bm{x})}$} &
-- & -- & Stationary point & Stationary point \\
\hline

\makecell[l]{Max sum-of-log-ratios:\\
$\displaystyle \max_{\bm{x}\in\mathcal{X}} \sum_{n=1}^{N}\log\left(1+\frac{A_n(\bm{x})}{B_n(\bm{x})}\right)$} &
-- & -- & Stationary point & Stationary point \\
\hline

\makecell[l]{Min sum-of-log-ratios:\\
$\displaystyle \min_{\bm{x}\in\mathcal{X}} \sum_{n=1}^{N}\log\left(1+\frac{A_n(\bm{x})}{B_n(\bm{x})}\right)$} &
-- & -- & -- & Stationary point \\
\hline

\makecell[l]{Max/min general \\ sum-of-products/ratios:\\
$\displaystyle \max/\min_{\bm{x}\in\mathcal{X}}\; J(\bm{x})$\\$+
\sum_{n=1}^{N}\prod_{k=1}^{K} f_n^{(k)}(\bm{x})$} &
-- & -- & -- & Stationary point \\
\hline

\makecell[l]{Max/min general \\ sum-of-log-products/ratios:\\
$\displaystyle \max/\min_{\bm{x}\in\mathcal{X}}\; J(\bm{x})$\\$+
\sum_{n=1}^{N}\log\left(1+\prod_{k=1}^{K} f_n^{(k)}(\bm{x})\right)$} &
-- & -- & -- & Stationary point \\
\hline

Convergence rate &
Not iterative & Superlinear & No general explicit rate & Sublinear for gradient-based variants \\
\bottomrule
\end{tabularx}
\end{table*}

Tables~\ref{tab:fp_target_problem} and~\ref{tab:fp_class_comparison} compare the proposed transforms with representative FP methods. Classical Charnes--Cooper and Dinkelbach methods provide strong global-optimality guarantees for single-ratio problems under the concave-convex condition. Dinkelbach-type methods can also be extended to max-min-ratio problems. However, these methods are mainly designed for single-ratio or max-min-ratio structures, and are not directly applicable to general sum-of-ratios, sum-of-products, or logarithmic coupled-term problems.

Quadratic-transform-based methods significantly broaden the scope of FP by decoupling multiple fractional terms. The quadratic transform is suitable for sum-of-ratios maximization, while the inverse quadratic transform addresses sum-of-ratios minimization. With additional reformulations, quadratic-transform-based methods can also handle sum-of-log-ratios maximization. These methods generally guarantee convergence to a stationary point for multi-ratio problems, which is consistent with the nonconvex nature of sum-of-ratios optimization.

In contrast, the proposed HM-GM-AM-QM transform is developed from inequality-induced upper and lower bounds. For minimization problems, the GM-AM and GM-QM inequalities provide upper-bound surrogates, while for maximization problems, the HM-GM inequality provides lower-bound surrogates. This enables the proposed framework to handle not only fractional terms but also general multiplicative terms with an arbitrary number of coupled functions. Since ratios can be regarded as special cases of products by including reciprocal terms, the proposed transform extends FP from ratio-specific optimization to broader product/ratio optimization.

Furthermore, by combining the monotonicity of the logarithm with the proposed upper- or lower-bound surrogates, the framework can also be extended to general sum-of-log-products/ratios problems. The gradient-based variants provide an explicit sublinear convergence guarantee to an $\epsilon$-stationary point, which is useful when the transformed subproblem is difficult to solve exactly. Therefore, compared with existing FP methods, the proposed transform offers a more general inequality-based SA framework for coupled multiplicative and logarithmic structures, while retaining stationary-point guarantees for both minimization and maximization problems.

\subsection{Application in Optimization Constraints}\label{sec_app_constraint}
\begin{table*}[t]
\centering
\scriptsize
\setlength{\tabcolsep}{2.2pt}
\renewcommand{\arraystretch}{1.22}
\caption{Examples of non-convex constraints transformed by the proposed HM-GM-AM-QM transform.}
\label{tab:constraint_examples}
\begin{tabularx}{\textwidth}{
p{0.23\textwidth}|
p{0.21\textwidth}|
p{0.34\textwidth}|
p{0.08\textwidth}|
p{0.12\textwidth}}
\hline
Original non-convex constraint & Auxiliary update at $x^{i}$ & Transformed surrogate constraint & Transform & Convexity \\
\hline

$\displaystyle x_1x_2-\Gamma\le 0,\quad x_1,x_2>0$ &
$\displaystyle y^{i}=\frac{x_2^{i}}{x_1^{i}}$ &
$\displaystyle 
\frac{1}{2}\left(y^{i}x_1^2+\frac{x_2^2}{y^{i}}\right)-\Gamma\le 0$ &
AM-GM &
Convex \\
\hline

$\displaystyle \big((x-a)^2+b\big)\big(x^2+c\big)-\Gamma\le 0,\quad b,c>0$ &
$\displaystyle 
y^{i}=\frac{(x^{i})^2+c}{(x^{i}-a)^2+b}$ &
$\displaystyle 
\frac{1}{2}\left[
y^{i}\big((x-a)^2+b\big)^2+
\frac{(x^2+c)^2}{y^{i}}\right]-\Gamma\le 0$ &
AM-GM &
Convex \\
\hline

$\displaystyle e^{x_1}(x_2^2+a)-\Gamma\le 0,\quad a>0$ &
$\displaystyle 
y^{i}=\frac{(x_2^{i})^2+a}{e^{x_1^{i}}}$ &
$\displaystyle 
\frac{1}{2}\left[
y^{i}e^{2x_1}
+\frac{(x_2^2+a)^2}{y^{i}}\right]-\Gamma\le 0$ &
AM-GM &
Convex \\
\hline

$\displaystyle \frac{x_1^2+a}{b-x_2^2}-\Gamma\le 0,\quad a,b>0,\ |x_2|<\sqrt{b}$ &
$\displaystyle 
y^{i}=\newline\frac{1}{2\big((x_1^{i})^2+a\big)\big(b-(x_2^{i})^2\big)}$ &
$\displaystyle 
y^{i}(x_1^2+a)^2+
\frac{1}{4y^{i}(b-x_2^2)^2}
-\Gamma\le 0$ &
AM-GM &
Convex \\
\hline

$\displaystyle \frac{e^{x_1}}{a+\sqrt{x_2}}-\Gamma\le 0,\quad a>0,\ x_2>0$ &
$\displaystyle 
y^{i}=\frac{1}{2e^{x_1^{i}}(a+\sqrt{x_2^{i}})}$ &
$\displaystyle 
y^{i}e^{2x_1}
+\frac{1}{4y^{i}(a+\sqrt{x_2})^2}
-\Gamma\le 0$ &
AM-GM &
Convex \\
\hline

$\displaystyle x_1x_2x_3-\Gamma\le 0,\quad x_1,x_2,x_3>0$ &
$\displaystyle 
y_1^{i}=\left(\frac{x_2^{i}}{x_1^{i}}\right)^{3/2},\newline
y_2^{i}=(y_1^{i})^{1/3}\frac{x_3^{i}}{x_2^{i}}$ &
$\displaystyle 
\frac{1}{3}\left[
y_1^{i}y_2^{i}x_1^3+
\frac{y_2^{i}}{y_1^{i}}x_2^3+
\frac{x_3^3}{(y_2^{i})^2}
\right]-\Gamma\le 0$ &
AM-GM &
Convex \\
\hline

$\displaystyle (x_1x_2)^2-\Gamma\le 0,\quad x_1,x_2>0$ &
$\displaystyle y^{i}=\frac{x_2^{i}}{x_1^{i}}$ &
$\displaystyle 
\frac{1}{2}\left[
(y^{i})^2x_1^4+
\frac{x_2^4}{(y^{i})^2}
\right]-\Gamma\le 0$ &
QM-GM &
Convex \\
\hline

$\displaystyle 
\left|\|x-u_1\|-\|x-u_2\|\right|-r\le 0$
\newline
or
$\displaystyle 
\|x-u_1\|^2+\|x-u_2\|^2
-2\|x-u_1\|\|x-u_2\|-r^2\le 0$ &
$\displaystyle 
y^{i}=
-\frac{\|x^{i}-u_2\|}
{2\|x^{i}-u_1\|}$ &
$\displaystyle 
\|x-u_1\|^2+\|x-u_2\|^2
+2y^{i}\|x-u_1\|^2
+\frac{\|x-u_2\|^2}{2y^{i}}
-\Big((r^{i})^2+2r^{i}(r-r^{i})\Big)
\le 0$ &
Negative AM-GM &
Maybe convex \\
\hline
\end{tabularx}
\end{table*}
In this section, we present how our proposed HM-GM-AM-QM transform is applied in optimization constraints. Table~\ref{tab:constraint_examples} shows several concrete non-convex constraints that can be transformed by the proposed transforms. For constraints of the form \(h(\bm{x})-\Gamma\le 0\), where $\Gamma$ is a constant number, the proposed transform constructs a surrogate upper bound \(Q(\bm{x};\bm{x}^\prime)\) such that
\[
h(\bm{x})\le Q(\bm{x};\bm{x}^\prime).
\]
Thus, the transformed constraint
\[
Q(\bm{x};\bm{x}^\prime)-\Gamma\le 0
\]
is a conservative approximation of the original constraint. The auxiliary variables are updated according to Equation (\ref{y_n_closedform}) at the current point \(\bm{x}^\prime\), so the surrogate is tight at \(\bm{x}^\prime\).

The first seven rows give convex surrogate constraints under the stated domains. These examples include bilinear terms, products of convex functions, exponential-polynomial products, fractional terms with concave denominators, and higher-order multiplicative terms. Hence, the proposed transform can be used not only for objective-function approximation, but also for constructing tractable convex constraints inside an SCA procedure.

The last row corresponds to the distance-difference constraint. After squaring the absolute-value constraint, the non-convexity appears through the negative product term \(-2\|\bm{x}-\bm{u}_1\|\|\bm{x}-\bm{u}_2\|\). The negative GM-AM transform can be used to obtain an SCA surrogate, together with the first-order approximation of \(r^2\). Note that this transformed constraint is not generally convex because the coefficients induced by the negative auxiliary variable may be negative. We present an example here where the proposed transform can be safely used.

\textbf{Case Study:}
Consider the quantum source position optimization problem in a quantum network. A quantum source is located at $\bm{q} \in \mathbb{R}^2$ in the quantum network. We use $\bm{u}_n^{(q)} \in \mathbb{R}^2$ to denote the location of quantum node $n$, where $n \in \mathcal{N}$. We denote each couple of nodes as $(n,n^\prime)$, where $n^\prime \neq n$. The total number of node couples is $M^{(q)} = \frac{N(N-1)}{2}$, and $m^{(q)} \in \{1,2,\cdots,M^{(q)}\}$ is used to denote $m^{(q)}$-th node couple.

The goal is to derive the optimal entanglement distribution by optimizing the position of the source $\bm{q}$. The detailed optimization problem can be found in Section \uppercase\expandafter{\romannumeral5} in \cite{iacovelli2024probability}. Here, we only consider the sub-problem 2, quantum source position, which appears in Section \uppercase\expandafter{\romannumeral6}.B. The quantum source position is given as follows:
\begin{subequations}\label{prob7}
\begin{align}
\min\limits_{\bm{q}} & \sum_{m^{(q)}=1}^{M^{(q)}} \Big(\alpha_{m^{(q)}}^{-1} 10^{\frac{\eta}{10}(\lVert \bm{q} - \bm{u}_n^{(q)} \rVert + \lVert \bm{q} - {\bm{u}_{n^\prime}^{(q)}} \rVert)} \cdot 10^{\beta \left| \lVert \bm{q} - \bm{u}_n^{(q)} \rVert - \lVert \bm{q} - {\bm{u}_{n^\prime}^{(q)}} \rVert \right|}\Big)
\\
\text{s.t.} &\quad \bm{q} \in \mathbb{R}^2,\label{constr_q}
\end{align}
\end{subequations}
where $\alpha_{m^{(q)}}$, $\eta$, and $\beta$ are constant parameters detailed in \cite{iacovelli2024probability}. The non-convex part in Problem (\ref{prob7}) is 
\begin{equation}
    \left| \lVert \bm{q} - \bm{u}_n^{(q)} \rVert - \lVert \bm{q} - {\bm{u}_{n^\prime}^{(q)}} \rVert \right|.\nonumber
\end{equation}
In the following section, we present how to transform this non-convex term into a convex term.

\textit{Problem Transformation:}
We first introduce an additional variable $\bm{r}=[r_1,\cdots,r_{M^{(q)}}]^\intercal$. We transform the optimization (\ref{prob7}) to the following equivalent optimization:
\begin{subequations}\label{prob8}
\begin{align}
\min\limits_{\bm{q},\bm{r}} & \sum_{m^{(q)}=1}^{M^{(q)}} \alpha_{m^{(q)}}^{-1} 10^{\frac{\eta}{10}(\lVert \bm{q} - \bm{u}_n^{(q)} \rVert + \lVert \bm{q} - {\bm{u}_{n^\prime}^{(q)}} \rVert)+\beta r_m}
\\
\text{s.t.} &\quad \text{(\ref{constr_q})},\nonumber \\
&\quad \left| \lVert \bm{q} - \bm{u}_n^{(q)} \rVert - \lVert \bm{q} - {\bm{u}_{n^\prime}^{(q)}} \rVert \right| \leq r_m,\label{constr_r}
\end{align}
\end{subequations}
where the constraint (\ref{constr_r}) is still non-convex. Let's square both sides of the constraint (\ref{constr_r}) as
\begin{align}
    \lVert \bm{q} - \bm{u}_n^{(q)} \rVert^2 + \lVert \bm{q} - \bm{u}_{n^\prime}^{(q)} \rVert^2 - 2  \lVert \bm{q} - \bm{u}_n^{(q)} \rVert \lVert \bm{q} - \bm{u}_{n^\prime}^{(q)} \rVert \leq r_m^2.
\end{align}
The first-order Taylor expansion is used to conduct the SCA technique in \cite{iacovelli2024probability}. Now, we will use our proposed AM bound to conduct the SCA method. We first analyze the term 
\begin{equation}
    - 2  \lVert \bm{q} - \bm{u}_n^{(q)} \rVert \lVert \bm{q} - \bm{u}_{n^\prime}^{(q)} \rVert \nonumber
\end{equation}
at the left side. Note that since this term is negative, the AM upper bound would be the lower bound for it. By using the proposed AM bound and introducing the additional variable $\bm{y}$, we can obtain
\begin{align}
    - 2  \lVert \bm{q} - \bm{u}_n^{(q)} \rVert \lVert \bm{q} - \bm{u}_{n^\prime}^{(q)} \rVert \geq \frac{4 \lVert \bm{q} - \bm{u}_n^{(q)} \rVert^2 y_n + \frac{\lVert \bm{q} - \bm{u}_{n^\prime}^{(q)} \rVert^2}{y_n}}{2},
\end{align}
where iff 
\begin{equation}
    y_n = -\frac{\lVert \bm{q} - \bm{u}_{n^\prime}^{(q)} \rVert}{2  \lVert \bm{q} - \bm{u}_n^{(q)} \rVert},
\end{equation}
the equal sign can be achieved. When we choose a feasible point $\bm{q}_0$, the left side of the constraint (\ref{constr_r}) would be
\begin{align}\label{eq_rm_left}
    \lVert \bm{q} - \bm{u}_n^{(q)} \rVert^2 + \lVert \bm{q} - \bm{u}_{n^\prime}^{(q)} \rVert^2 -  \lVert \bm{q} - \bm{u}_n^{(q)} \rVert^2 \cdot \frac{\lVert \bm{q}_0 - \bm{u}_{n^\prime}^{(q)} \rVert}{  \lVert \bm{q}_0 - \bm{u}_n^{(q)} \rVert} -\frac{\lVert \bm{q} - \bm{u}_{n^\prime}^{(q)} \rVert^2 \lVert \bm{q}_0 - \bm{u}_n^{(q)} \rVert}{\lVert \bm{q}_0 - \bm{u}_{n^\prime}^{(q)} \rVert}
    \leq r_m^2,
\end{align}
where the left side of Equation (\ref{eq_rm_left}) by using the AM bounds can achieve the same function by using the first-order Taylor expansion in \cite{iacovelli2024probability}. This finding demonstrates that our AM bounds can be used as the lower or upper bound of functions in constraints. 

However, we can't obtain a helpful AM bound of $r_m^2$. The reason is that when we apply our bounds to $r_m^2$, the additional variable would be $\frac{r_m}{r_m}$, e.g., 1, which means that the additional variable based on our construction idea would be a constant number. Thus, all HM, AM, and QM bounds are all $r_m^2$. This is the limitation of our proposed bounds, and we will discuss this later. For $r_m^2$, we can use the first-order Taylor expansion to get the lower bound of $r_m^2$. If we fix 
\begin{equation}
    y_n = -\frac{\lVert \bm{q} - \bm{u}_{n^\prime}^{(q)} \rVert}{2  \lVert \bm{q} - \bm{u}_n^{(q)} \rVert},
\end{equation}
the final transformed optimization would be 
\begin{subequations}\label{prob9}
\begin{align}
\min\limits_{\bm{q},\bm{r}} & \sum_{m^{(q)}=1}^{M^{(q)}} \alpha_{m^{(q)}}^{-1} 10^{\frac{\eta}{10}(\lVert \bm{q} - \bm{u}_n^{(q)} \rVert + \lVert \bm{q} - {\bm{u}_{n^\prime}^{(q)}} \rVert)+\beta r_m}
\\
\text{s.t.} &\quad \text{(\ref{constr_q})},\nonumber \\
&\quad \lVert \bm{q} - \bm{u}_n^{(q)} \rVert^2 + \lVert \bm{q} - \bm{u}_{n^\prime}^{(q)} \rVert^2 +  \lVert \bm{q} - \bm{u}_n^{(q)} \rVert^2 \cdot 2 y_n \nonumber \\
&\quad +\frac{\lVert \bm{q} - \bm{u}_{n^\prime}^{(q)} \rVert^2}{2 y_n}
    \leq r_{m,0}^2 + 2 r_{m,0}(r_m - r_{m,0}),
\end{align}
\end{subequations}
where $\bm{r}_0 = [r_{1,0},\cdots,r_{M,0}]^\intercal$ is the local point of the expansion. The transformed Problem (\ref{prob9}) is the same convex problem in \cite{iacovelli2024probability}. The subsequent optimization process is consistent with that in \cite{iacovelli2024probability}, so we omit it here.

From transforming $- 2  \lVert \bm{q} - \bm{u}_n^{(q)} \rVert \lVert \bm{q} - \bm{u}_{n^\prime}^{(q)} \rVert$, we conceive the following finding for our proposed bounds:
\begin{lemma}\label{lemma2}
For the multiplicative term $-A_n(\bm{x})B_n(\bm{x})$, functions $A_n(\bm{x}): \mathbb{R}^M \rightarrow \mathbb{R}_{++}$, $B_n(\bm{x}): \mathbb{R}^M \rightarrow \mathbb{R}_{++}$, $\forall n \in \mathcal{N}$, we obtain that
    \begin{align}
    \frac{2}{\frac{1}{\left(A_n(\bm{x})\right)^2 y_n}+\frac{y_n}{\left(B_n(\bm{x})\right)^2}}
    &\geq -\sqrt{\left(A_n(\bm{x})\right)^2 y_n \cdot \frac{\left(B_n(\bm{x})\right)^2}{y_n}}\nonumber \\  
    &\geq\frac{\left(A_n(\bm{x})\right)^2 y_n+\frac{\left(B_n(\bm{x})\right)^2}{y_n}}{2}\nonumber \\  
    &\geq -\sqrt{\frac{\left(A_n(\bm{x})\right)^4 y_n^2+\frac{\left(B_n(\bm{x})\right)^4}{y_n^2}}{2}},
\end{align}
where 
    $y_n = -\frac{B_n(\bm{x})}{A_n(\bm{x})}$,
the equality holds. In this case, the upper and lower bounds are reversed.
\end{lemma}
\begin{proof}
    Multiply each side of Inequality (\ref{eq_inequality_two_functions}) by a minus sign. Then let the new $y_n = -\frac{B_n(\bm{x})}{A_n(\bm{x})}$. We can easily obtain the \textbf{Lemma \ref{lemma2}}.
\end{proof}
\begin{remark}
    Obviously, \textbf{Lemma \ref{lemma2}} also applies to the HM-GM, QM-GM transforms, and the general multiplicative terms
        $-\sum_{n=1}^N \left(\prod_{k=1}^K f^{(k)}_n(\bm{x})\right)$.
\end{remark}

\subsection{Numerical Results}\label{sec_simple_results}
In this section, we present some simple numerical simulations to show the effectiveness of the proposed SA methods. 
\subsubsection{Minimization Problem}
First, consider the following minimization problem (\ref{prob11}):
\begin{subequations}\label{prob11}
\begin{align}
\min\limits_{x}\quad &
x+\frac{x}{\ln x}+ \frac{x}{\ln x}e^x
\\
\text{s.t.} \quad &  1 < x \leq 10.
\end{align}
\end{subequations}
Problem (\ref{prob11}) is non-convex because of the existence of $\frac{x}{\ln x}$ and $\frac{x}{\ln x}e^x$. We analyze this optimization by the proposed bounds.
Let
\begin{align}
    &\tilde{f}_1^{(1)} = y_1^{(1)} x^2, \tilde{f}_1^{(2)} = \frac{1}{y_1^{(1)} \ln^2 x}, \tilde{f}_2^{(1)} = y_2^{(1)} y_2^{(2)} x^3, \tilde{f}_2^{(2)} = \frac{y_2^{(2)}}{y_2^{(1)}\ln^3 x}, \tilde{f}_2^{(3)} = \frac{e^{3x}}{(y_2^{(2)})^2},
\end{align}
where 
\begin{align}
    &y_1^{(1)} = \frac{1}{x\ln x},  
    y_2^{(1)} = x^{-\frac{3}{2}} (\ln x)^{-\frac{3}{2}},   
    y_2^{(2)} = x^{-\frac{1}{2}}(\ln x)^{\frac{1}{2}}e^x.
\end{align}
The AM and QM surrogates are given as
\begin{align}
    &F^{AM} = x + \frac{1}{2} (\tilde{f}_1^{(1)} + \tilde{f}_1^{(2)}) + \frac{1}{3}(\tilde{f}_2^{(1)} + \tilde{f}_2^{(2)} + \tilde{f}_2^{(3)} ) ,\\
    &F^{QM} = x + \sqrt{\frac{(\tilde{f}_1^{(1)})^2 + (\tilde{f}_1^{(2)})^2}{2}}  + \sqrt{\frac{(\tilde{f}_2^{(1)})^2 + (\tilde{f}_2^{(2)})^2 + (\tilde{f}_2^{(3)})^2 }{3}}.
\end{align}
\begin{figure}[tbp]
\centering
\includegraphics[width=1\textwidth]{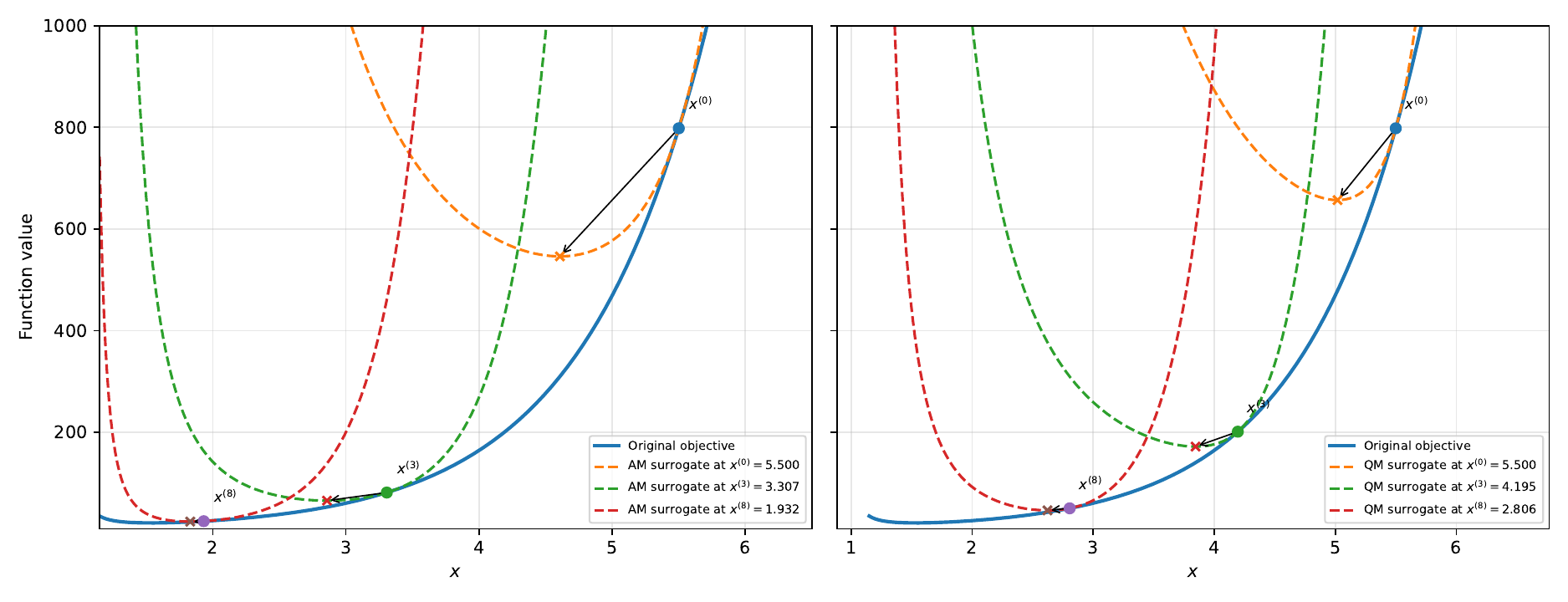}
\vspace{-6pt}\caption{SA procedures for minimization problem (\ref{prob11}) at feasible points based on AM and QM surrogates.}
\label{fig_sur_amqm}
\end{figure}

\textbf{Simulation settings:}
We evaluate the proposed AM-GM- and QM-GM-based SA methods on Problem~(\ref{prob11}). The feasible interval is set as \(1<x\le 10\), and all algorithms are initialized at \(x^{0}=5.5\). For the surrogate illustration, we plot the original objective together with the AM and QM surrogate functions constructed at three representative SA iterates, so as to visualize the tangency between the surrogate and the original objective at the current point.

For convergence comparison, we consider four methods: AM-SA, QM-SA, gradient-based AM-SA, and gradient-based QM-SA. In AM-SA and QM-SA, the transformed surrogate subproblem is solved accurately by running inner gradient descent until the inner gradient norm is below \(10^{-9}\), with a maximum of 3000 inner steps. In the gradient-based variants, only \(M=3\) inner gradient steps are performed at each outer iteration, which represents an efficient inexact update of the surrogate subproblem. All inner gradient updates use backtracking line search with initial stepsize \(0.5\), Armijo parameter \(10^{-4}\), and minimum stepsize \(10^{-14}\). The outer iteration is terminated when the stationarity gap satisfies $|\nabla \Phi(x^{i})|\le 10^{-6}$,
or when the maximum number of outer iterations (set as 100) is reached. The convergence curves report the stationarity gap \(|\nabla \Phi(x^{i})|\) versus the outer iteration index. To compare computational efficiency, each method is independently repeated 30 times, and we report the average running time, standard deviation, average number of outer iterations, average number of inner gradient steps, final objective value, and final stationarity gap.

\begin{figure}[tbp]
\centering
\includegraphics[width=0.6\textwidth]{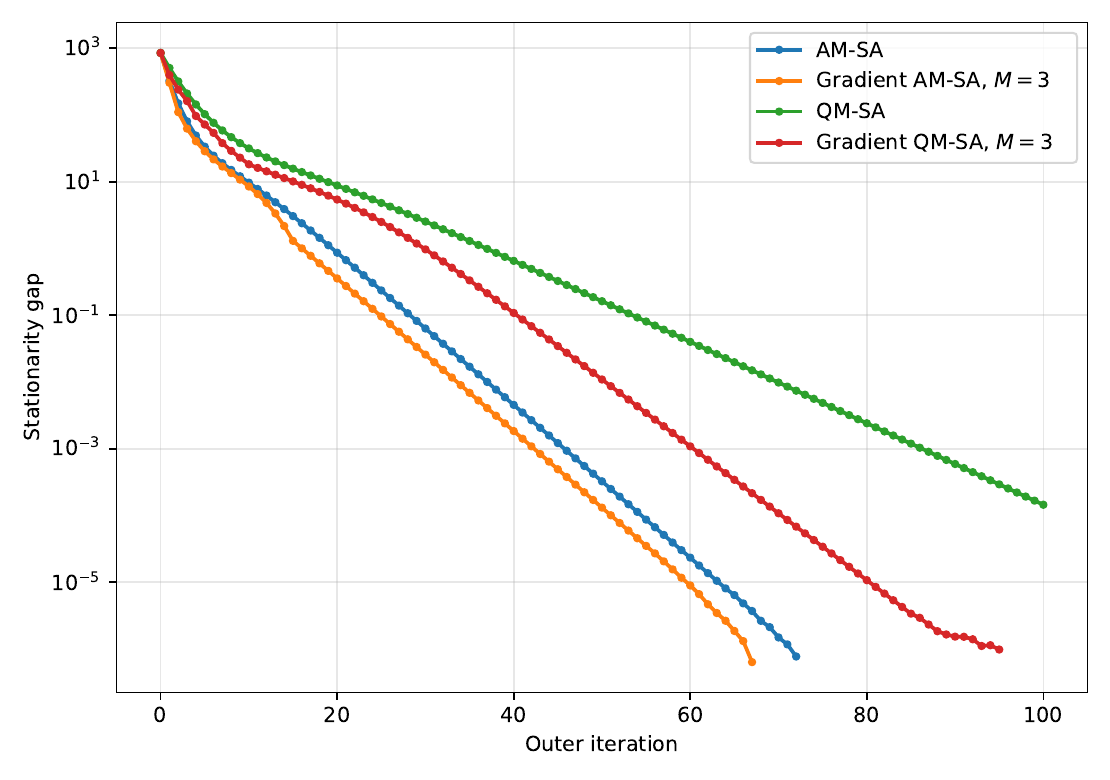}
\vspace{-6pt}\caption{Convergence of the SA method with the proposed AM/QM bound to solve Problem (\ref{prob11}).}
\label{fig_amqm_convergence}
\end{figure}

\begin{table}[t]
\centering
\caption{Time performance comparison of AM-SA and QM-SA methods.}
\label{tab:sca_comparison}
\resizebox{\textwidth}{!}{%
\begin{tabular}{lrrrrrr}
\toprule
Method & Avg. time (s) & Std. time (s) & Avg. outer iters & Avg. inner steps & Avg. final objective & Avg. final stationarity gap \\
\midrule
AM-SA & 5.2524e-01 & 3.3489e-03 & 7.2000e+01 & 2.3800e+03 & 2.1742e+01 & 7.6978e-07 \\
Gradient AM-SA, $M=3$ & 4.5559e-02 & 2.6408e-04 & 6.7000e+01 & 2.0100e+02 & 2.1742e+01 & 6.3410e-07 \\
QM-SA & 1.8305e+00 & 1.3523e-02 & 1.0000e+02 & 6.5740e+03 & 2.1742e+01 & 1.4467e-04 \\
Gradient QM-SA, $M=3$ & 8.1262e-02 & 6.0345e-04 & 9.5000e+01 & 2.8500e+02 & 2.1742e+01 & 9.7886e-07 \\
\bottomrule
\end{tabular}%
}
\end{table}

In Fig. \ref{fig_sur_amqm}, we show the proposed surrogates tightly wrapped around the original function in Problem (\ref{prob11}), and they are all tangent to the chosen feasible points. Note that the AM upper bound function $F^{AM}$ and QM upper bound function $F^{QM}$ are convex over $x$ under the given constraint. Next, we conduct the SA algorithm with the $F^{AM}$ and $F^{QM}$ to find the optimum of Problem (\ref{prob11}).
We transformed Problem (\ref{prob11}) to a new convex \mbox{Problem (\ref{prob11_trans})} as
\begin{subequations}\label{prob11_trans}
\begin{align}
\min\limits_{x, \bm{y}}\quad & F^{AM/QM}
\\
\text{s.t.} \quad &  1 < x \leq 10.
\end{align}
\end{subequations}



Fig.~\ref{fig_amqm_convergence} shows the convergence behavior in terms of the stationarity gap. All methods reduce the stationarity gap steadily, confirming that the proposed SA framework can effectively approach a stationary point. The AM-based methods converge faster than the QM-based methods, which is consistent with the tighter AM surrogate observed in Fig.~\ref{fig_sur_amqm}. Moreover, the gradient-based AM-SA with only \(M=3\) inner gradient steps closely follows the convergence behavior of AM-SA, which means that accurately solving each surrogate subproblem is not always necessary.

In Table~\ref{tab:sca_comparison}, we show that all methods reach almost the same final objective value, around \(21.742\), confirming the effectiveness of both AM- and QM-based SA methods. However, the gradient-based variants are significantly faster. Gradient AM-SA reduces the average running time from \(0.525\) s to \(0.0456\) s, while gradient QM-SA reduces it from \(1.83\) s to \(0.0813\) s. This corresponds to about \(11.5\times\) and \(22.5\times\) speedup, respectively.

\subsubsection{Maximization Problem}

Next, consider the following maximization problem:
\begin{subequations}\label{prob_max}
\begin{align}
\max_{x}\quad &
-\rho(x-c)^2+\frac{\ln x}{x}+\frac{\ln x}{x}e^{-x}
\\
\text{s.t.}\quad & 1<x\le 10,
\end{align}
\end{subequations}
where \(\rho>0\) and \(c>1\). In the simulation, we set \(\rho=10^{-3}\) and \(c=4\). Problem~\eqref{prob_max} is non-concave because it contains the coupled fractional/product terms \(\ln x/x\) and \((\ln x/x)e^{-x}\).
Let
\begin{align}
    &\tilde{f}_1^{(1)}=y_1^{(1)}(\ln x)^2,
    \tilde{f}_1^{(2)}=\frac{1}{y_1^{(1)}x^2},
    \tilde{f}_2^{(1)}=y_2^{(1)}y_2^{(2)}(\ln x)^3,
    \tilde{f}_2^{(2)}=\frac{y_2^{(2)}}{y_2^{(1)}x^3},
    \tilde{f}_2^{(3)}=\frac{e^{-3x}}{(y_2^{(2)})^2},
\end{align}
where
\begin{align}
    &y_1^{(1)}=\frac{1}{x\ln x},y_2^{(1)}=\left(\frac{1}{x\ln x}\right)^{3/2},y_2^{(2)}=\sqrt{\frac{x}{\ln x}}e^{-x}.
\end{align}
Then the HM lower-bound surrogate is
\begin{align}
F^{HM}
=-\rho(x-c)^2
+\frac{2}{\frac{1}{\tilde{f}_1^{(1)}}+\frac{1}{\tilde{f}_1^{(2)}}}
+
\frac{3}{\frac{1}{\tilde{f}_2^{(1)}}+\frac{1}{\tilde{f}_2^{(2)}}+\frac{1}{\tilde{f}_2^{(3)}}}.
\end{align}
Based on the HM lower-bound surrogate, we can convert the Problem (\ref{prob_max}) into the following problem:
\begin{subequations}\label{prob_max_convert}
\begin{align}
\max_{x,\bm{y}}\quad &
F^{HM} \\
\text{s.t.}\quad & 1<x\le 10.
\end{align}
\end{subequations}
At each SA iteration, the auxiliary variables are updated at the current point, and the resulting lower-bound surrogate is maximized. Simulation settings are the same as those in the minimization problem.

\begin{figure}[tbp]
\centering
\includegraphics[width=0.6\textwidth]{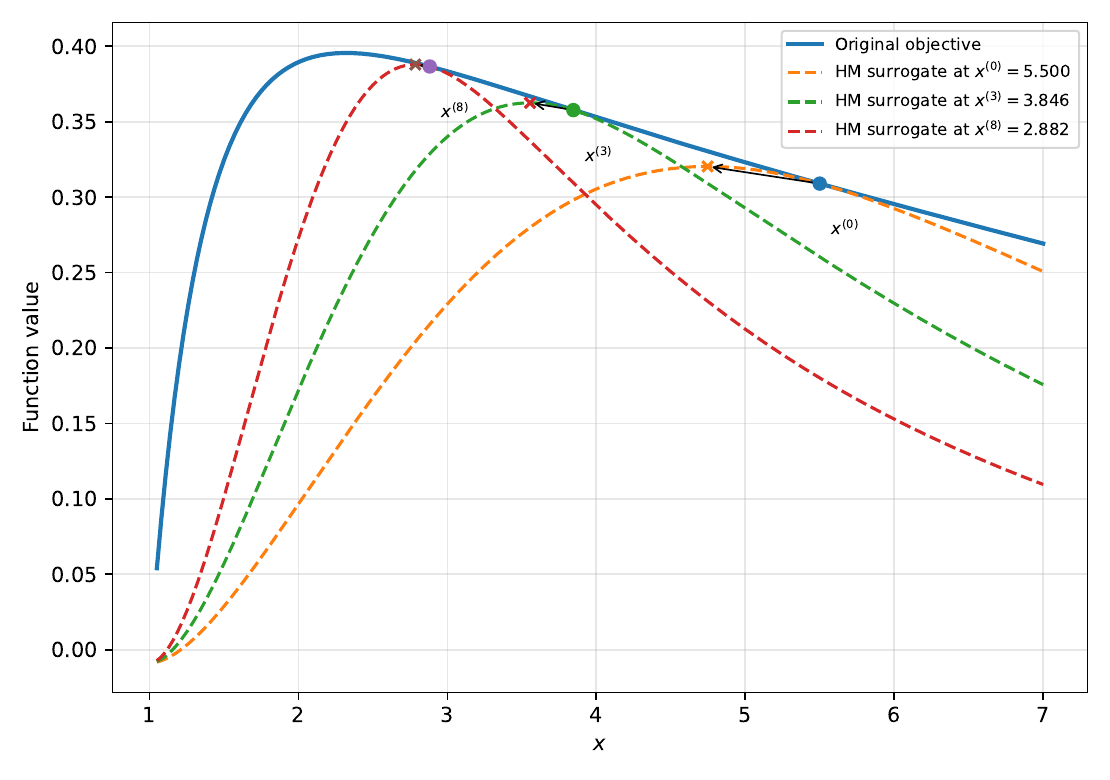}
\vspace{-6pt}\caption{SA procedures for maximization problem (\ref{prob_max}) at feasible points based on the HM surrogate.}
\label{fig_sur_hm}
\end{figure}
Fig.~\ref{fig_sur_hm} illustrates the SA procedure based on the proposed HM surrogate for Problem~(\ref{prob_max}). The HM surrogate is constructed as a lower-bound function for the maximization problem. At each selected feasible point, the surrogate is tight and tangent to the original objective, which verifies the value-matching and first-order consistency properties required by SA. As the iteration proceeds, the tangent point moves from the initial point toward the neighborhood of the maximizer, showing how the HM-based minorization guides the update of \(x\).

\begin{figure}[tbp]
\centering
\includegraphics[width=0.6\textwidth]{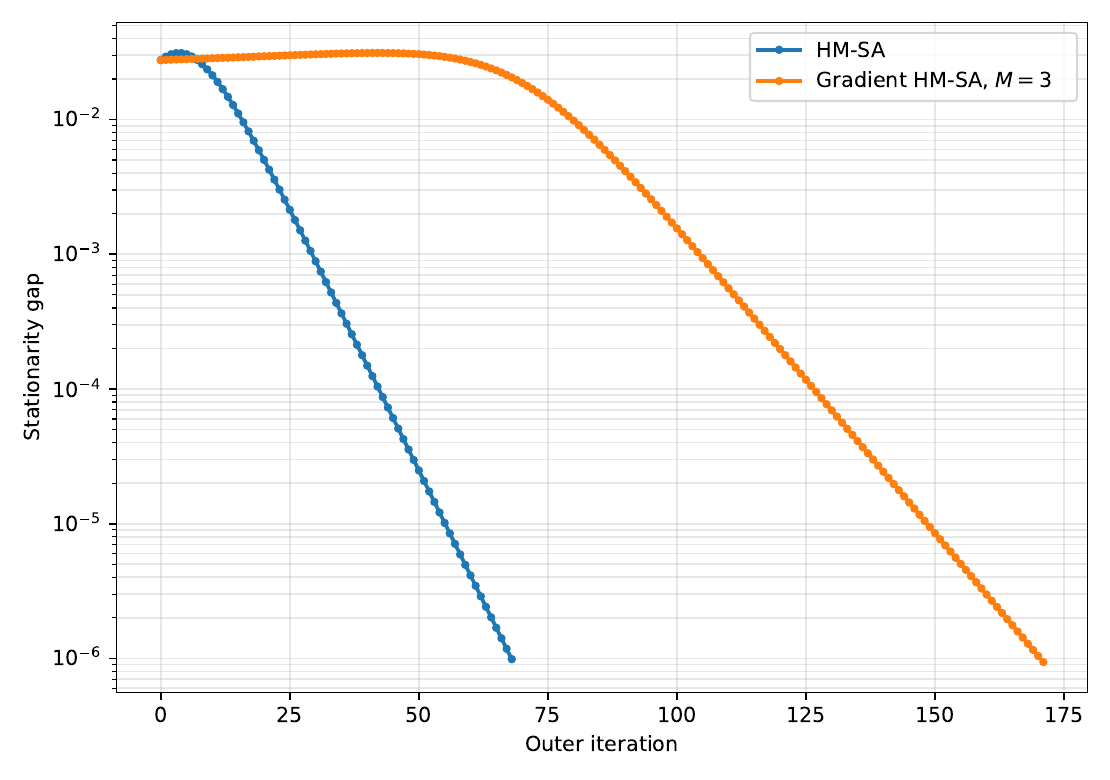}
\vspace{-6pt}\caption{Convergence of the SA method with the proposed HM surrogate to solve Problem (\ref{prob_max_convert}).}
\label{fig_hm_convergence}
\end{figure}

\begin{table}[t]
\centering
\caption{Time performance comparison of HM-SA methods.}
\label{tab:sca_time_hm}
\resizebox{\textwidth}{!}{%
\begin{tabular}{lrrrrrr}
\toprule
              Method &  Avg. time (s) &  Std. time (s) &  Avg. outer iters &  Avg. inner steps &  Avg. final objective &  Avg. final stationarity gap \\
\midrule
              HM-SA &     3.0619e-01 &     1.6609e-03 &        6.8000e+01 &        5.1080e+03 &            3.9557e-01 &                   9.8693e-07 \\
Gradient HM-SA, M=3 &     3.2560e-02 &     1.6656e-04 &        1.7100e+02 &        5.1300e+02 &            3.9557e-01 &                   9.3861e-07 \\
\bottomrule
\end{tabular}%
}
\end{table}

In Fig.~\ref{fig_hm_convergence}, the convergence behavior in terms of the stationarity gap is presented. Both HM-SA and gradient HM-SA steadily reduce the stationarity gap and finally reach the prescribed tolerance around \(10^{-6}\). HM-SA converges in fewer outer iterations because it solves each surrogate subproblem more accurately. In contrast, gradient HM-SA uses only \(M=3\) local gradient steps at each outer iteration, so it requires more outer iterations but has much lower per-iteration cost.

Table~\ref{tab:sca_time_hm} further confirms this trade-off. Both methods reach the same final objective value, \(3.9557\times 10^{-1}\), and comparable final stationarity gaps below \(10^{-6}\). However, gradient HM-SA reduces the average running time from \(3.0619\times 10^{-1}\) s to \(3.2560\times 10^{-2}\) s, achieving about \(9.4\times\) speedup. This is mainly because the total number of inner gradient steps is reduced from \(5.1080\times 10^{3}\) to \(5.1300\times 10^{2}\). These results show that the gradient-based HM-SA method provides a more efficient implementation while maintaining almost the same solution quality.

\section{Applications and Numerical Studies}\label{sec_application}
In this section, we present several application scenarios where the proposed transforms can be used. Note that if the transformed problem is convex, then we will write its SA method as SCA method.
\subsection{Minimization of Transmission Energy in Wireless Communications}\label{sec_minimization_transmission_energy}
We will show how our proposed AM upper bound can be applied to minimize transmission energy in wireless communication between mobile users and the edge server. 

\subsubsection{Problem Statement}
Consider a system consisting of $N$ mobile users and one server. $n$ is used as the indices for a specific user, where $n \in \mathcal{N}$. Frequency-division multiple access (FDMA) is considered in this system so that communication between users and the server would not interfere. The transmission rate from the user $n$ to the server is $r_n = b_n \log_2 (1+\frac{g_n p_n}{b_n \sigma^2})$ based on the Shannon formula, where $b_n$ is the allocated bandwidth between the server and user $n$, $p_n$ is the transmit power of user $n$, $g_n$ is the channel attenuation between the server and the user $n$, and $\sigma^2$ is the noise power spectral density. Denote $d_n$ as the data that the user $n$ offloaded to the server, $b_{max}$ as the total allocated bandwidth between the users and the server, and $p_{max}$ as the maximum available transmit power of each user. Given this information, the minimization problem would be 
\begin{subequations}\label{prob5}
\begin{align}
\min\limits_{\bm{b},\bm{p}} &\quad \sum_{n=1}^N \frac{p_n d_n}{b_n \log_2 (1+\frac{g_n p_n}{b_n \sigma^2})}
\\
\text{s.t.} &\quad  \sum_{n=1}^N b_n \leq b_{max}, \label{constr_b_max}\\
&\quad\, p_n \leq p_{max}.\label{constr_p_max}
\end{align}
\end{subequations}
$p_n$ and $b_n$ are assumed to be positive. Generally, this sum-of-ratios minimization problem is non-convex, and NP-complete \cite{freund2001solving}. Thus, it is difficult to solve directly. 

\subsubsection{SCA Method Based on the Proposed AM-GM Transform}
Based on our proposed AM upper bound (\ref{bounds_K}), we can transform it into the following optimization:
\begin{subequations}\label{prob6}
\begin{align}
\min\limits_{\bm{b},\bm{p},\bm{y}} &\quad \sum_{n=1}^N d_n^2 p_n^2 y_n + \frac{1}{4 y_n \left(b_n \log_2 (1+\frac{g_n p_n}{b_n \sigma^2})\right)^2}
\tag{\ref{prob6}}\\
\text{s.t.} &\quad  \text{(\ref{constr_b_max})},\text{(\ref{constr_p_max})}.\nonumber
\end{align}
\end{subequations}
where if given one feasible point $(\bm{b},\bm{p})$, $y_n$ can be fixed as 
\begin{align}
    y_n = \frac{1}{2 d_n p_n b_n \log_2 (1+\frac{g_n p_n}{b_n \sigma^2})}.\label{eq_y_energymin}
\end{align}
Since $b_n \log_2 (1+\frac{g_n p_n}{b_n \sigma^2})$ is jointly concave of $(b_n,p_n)$ \cite{qian2024user}, $\frac{1}{\left(b_n \log_2 (1+\frac{g_n p_n}{b_n \sigma^2})\right)^2}$ is jointly convex of $(b_n,p_n)$ according to the scalar composition rule in \cite{boyd2004convex}. Therefore, the optimization problem (\ref{prob6}) is a convex optimization, which can be solved by common convex tools, e.g., CVX \cite{grant2014cvx}. At least a stationary point of Problem (\ref{prob5}) is guaranteed by solving Problem (\ref{prob6}). 






\subsubsection{Parameter Setting}
The number of mobile users $N$ is set as 40. We consider denoting $g_n$ as $h_n l_n$, where $h_n$ is the large-scale slow-fading component capturing effects of path loss and shadowing, and $l_n$ is the small-scale Rayleigh fading. $h_n$ is given as $128.1+37.6\log_2 d^{(o)}_n$ in \cite{qian2024user}, where $d^{(o)}_n$ is the Euclidean distance between the user $n$ and the server. Gaussian noise power $\sigma^2$ is $-134$ dBm. The maximum bandwidth $b_{max}$ is assumed to be $10$ MHz. The maximum transmit power $p_{max}$ of each mobile user is $10$ W. The data size of the mobile user $d_n$ is randomly selected from $[500\text{KB},2000\text{KB}]$. The Mosek optimization tool in Matlab is used to conduct the simulations. The tolerant error gap $\epsilon$ is set as $10^{-4}$. We set the starting point as the value under the average allocation policy.

\subsubsection{Numerical Results}
\begin{figure}[tbp]
\centering
\includegraphics[width=0.6\textwidth]{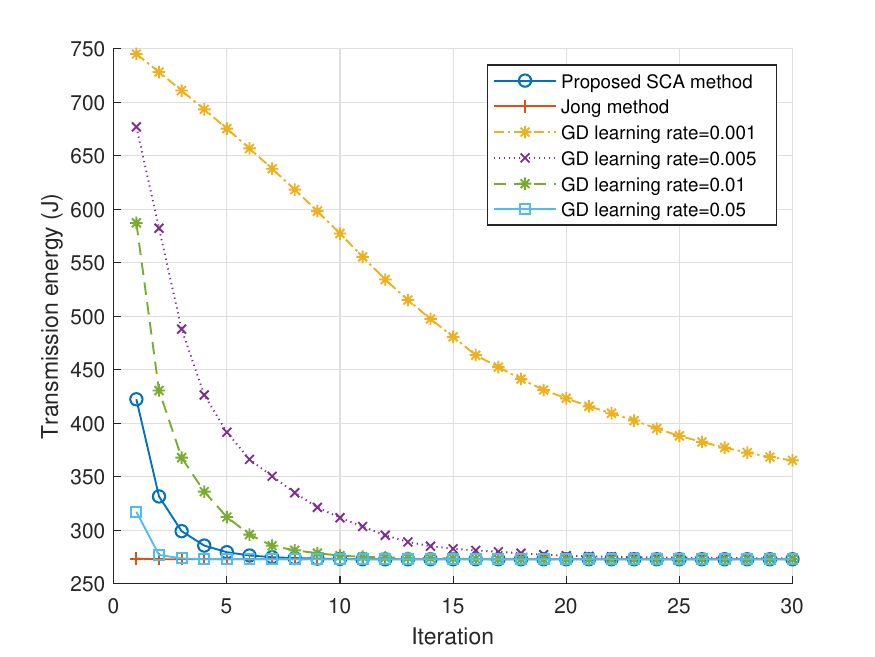}
\vspace{-6pt}\caption{Convergence of different algorithms.}
\label{fig_convergence_algorithms}
\end{figure}
Apart from the proposed SCA method, we also consider the following benchmarks: 1. The classical gradient descent (GD) method \cite{bottou2012stochastic}; 2. The method proposed in \cite{jong2012efficient}, where a global optimum is guaranteed to converge for minimization problem with only sum-of-ratios in the $\frac{\text{convex}}{\text{concave}}$ form. For clarity, we denote this method as the Jong method. The solution derived by the Jong method is considered the global optimum. In Fig. \ref{fig_convergence_algorithms}, the convergence of these algorithms is presented.

In Fig. \ref{fig_convergence_algorithms}, all algorithms can converge to the global optimum because the ratio term in Problem (\ref{prob5}) is actually a pseudoconvex function. Thus, the stationary point is equivalent to the optimal point. As for convergence speed, the Jong method converges fastest, followed by the GD method with a learning rate of 0.05 and the proposed SCA method. Although the Jong method converges fastest, it introduces one more additional variable than our proposed SCA method. Details about the Jong method can be found in \cite{jong2012efficient}. The GD with a large enough learning rate can converge faster than the proposed SCA method in Problem (\ref{prob5}). However, the convergence of the GD method is generally not guaranteed. Yet, a stationary point is guaranteed to be obtained by using the proposed SCA method with the AM bound. Besides, the Jong method can only solve specific sum-of-ratios (i.e., $\frac{\text{convex}}{\text{concave}}$) optimization and our proposed SCA method can be applied to the optimization of sum-of-ratios (products) with other general functions, e.g., Problem (\ref{prob11}).

\subsection{Age-of-Information Minimization in Status Update Networks}\label{sec_aoi}

We next consider an information-freshness optimization problem in a status update system. Age of information (AoI) is a widely used metric to quantify the freshness of received information. A smaller AoI implies that the destination has more up-to-date information about the monitored source.

\subsubsection{Problem Statement}

Consider a status update system with \(K\) source nodes and one edge server. Source \(k\) generates update packets with rate \(\lambda_k\), and the server processes received packets with service rate \(\mu\). Following the priority-based M/M/1 queueing model, define
\begin{equation}
\rho_k = \frac{\lambda_k}{\mu},
\qquad
\hat{\rho}_k = \sum_{j=1}^{k-1}\rho_j .
\end{equation}
The average AoI of source \(k\) can be written as
\begin{equation}
\bar{\Delta}_k
=
\frac{\hat{\rho}_k^2+3\hat{\rho}_k+1}{\mu(1+\hat{\rho}_k)}
+
\frac{(\hat{\rho}_k+1)^2}{\mu\rho_k}.
\end{equation}
Then the sum-AoI minimization problem is formulated as
\begin{subequations}\label{prob:aoi}
\begin{align}
\min_{\bm{\lambda}} \quad &
\sum_{k=1}^{K}
\left(
\frac{\hat{\rho}_k^2+3\hat{\rho}_k+1}{\mu(1+\hat{\rho}_k)}
+
\frac{(\hat{\rho}_k+1)^2}{\mu\rho_k}
\right) \\
\text{s.t.}\quad &
\lambda_{\min}\le \lambda_k \le \lambda_{\max},\quad \forall k,\\
&
\sum_{k=1}^{K}\lambda_k \le \eta\mu,
\end{align}
\end{subequations}
where \(0<\eta<1\) is used to guarantee queue stability. Problem~\eqref{prob:aoi} is a sum-of-ratios minimization problem and is generally non-convex.

\subsubsection{Problem Transformation}

For notational simplicity, define
\begin{equation}
A_{k,1}(\bm{\lambda})=\hat{\rho}_k^2+3\hat{\rho}_k+1,
\qquad
B_{k,1}(\bm{\lambda})=\mu(1+\hat{\rho}_k),
\end{equation}
and
\begin{equation}
A_{k,2}(\bm{\lambda})=(\hat{\rho}_k+1)^2,
\qquad
B_{k,2}(\bm{\lambda})=\mu\rho_k .
\end{equation}
Then the AoI objective can be rewritten as
\begin{equation}
\sum_{k=1}^{K}
\left(
\frac{A_{k,1}(\bm{\lambda})}{B_{k,1}(\bm{\lambda})}
+
\frac{A_{k,2}(\bm{\lambda})}{B_{k,2}(\bm{\lambda})}
\right).
\end{equation}
Using the proposed AM-GM transform, each ratio term can be upper-bounded as
\begin{equation}
\frac{A_{k,l}(\bm{\lambda})}{B_{k,l}(\bm{\lambda})}
\le
y_{k,l}\left(A_{k,l}(\bm{\lambda})\right)^2
+
\frac{1}{4y_{k,l}\left(B_{k,l}(\bm{\lambda})\right)^2},
\qquad l\in\{1,2\},
\end{equation}
where the equality holds when
\begin{equation}
y_{k,l}^{i}
=
\frac{1}{2A_{k,l}(\bm{\lambda}^{i})B_{k,l}(\bm{\lambda}^{i})}.
\end{equation}
At iteration \(i\), with fixed \(y_{k,l}^{i}\), Problem~\eqref{prob:aoi} is approximated by
\begin{subequations}\label{prob:aoi_trans}
\begin{align}
\min_{\bm{\lambda}} \quad &
\sum_{k=1}^{K}\sum_{l=1}^{2}
\left[
y_{k,l}^{i}
\left(A_{k,l}(\bm{\lambda})\right)^2
+
\frac{1}{4y_{k,l}^{i}
\left(B_{k,l}(\bm{\lambda})\right)^2}
\right]\\
\text{s.t.}\quad &
\lambda_{\min}\le \lambda_k \le \lambda_{\max},\quad \forall k,\\
&
\sum_{k=1}^{K}\lambda_k \le \eta\mu .
\end{align}
\end{subequations}
This transformed problem is a convex surrogate under the convex-concave ratio structure. Therefore, the proposed AM-SCA method can be used to solve Problem~\eqref{prob:aoi}.

\subsubsection{SCA Method Based on the Proposed AM-GM Transform}

The SCA procedure is summarized as follows. First, initialize a feasible update-rate vector \(\bm{\lambda}^{0}\). At the \(i\)-th iteration, update the auxiliary variables by
\begin{equation}
y_{k,l}^{i}
=
\frac{1}{2A_{k,l}(\bm{\lambda}^{i})B_{k,l}(\bm{\lambda}^{i})},
\qquad k=1,\ldots,K,\quad l=1,2 .
\end{equation}
Then solve the transformed problem~\eqref{prob:aoi_trans} to obtain \(\bm{\lambda}^{i+1}\). The procedure is repeated until the objective value converges. Alternatively, the gradient-based AM-SCA method can be used by replacing the exact solution of Problem~\eqref{prob:aoi_trans} with several gradient steps.

\subsubsection{Parameter Setting}

Unless otherwise specified, the service rate is set as \(\mu=1\). The number of sources is varied from \(K=3\) to \(K=10\). The stability parameter is set as \(\eta=0.95\), and the update-rate constraints are given by \(\lambda_{\min}=10^{-3}\) and \(\lambda_{\max}=\eta\mu\). The initial point is set as the equal-rate allocation
\begin{equation}
\lambda_k^{0}=\frac{\eta\mu}{K},\quad \forall k .
\end{equation}
For the gradient-based AM-SCA method, the number of inner gradient steps is set as \(M=3\). The stopping tolerance is set as \(10^{-4}\), and the maximum number of SCA iterations is set as 100.

We compare the proposed method with the following baselines. The equal-rate scheme sets all update rates to the same value and performs a one-dimensional search over the common rate. The max-rate scheme sets all update rates to the largest feasible equal rate. We also include the inverse quadratic transform as a representative FP-based baseline for sum-of-ratios minimization.

\subsubsection{Numerical Results}
\begin{figure}[tbp]
\centering
\includegraphics[width=0.6\textwidth]{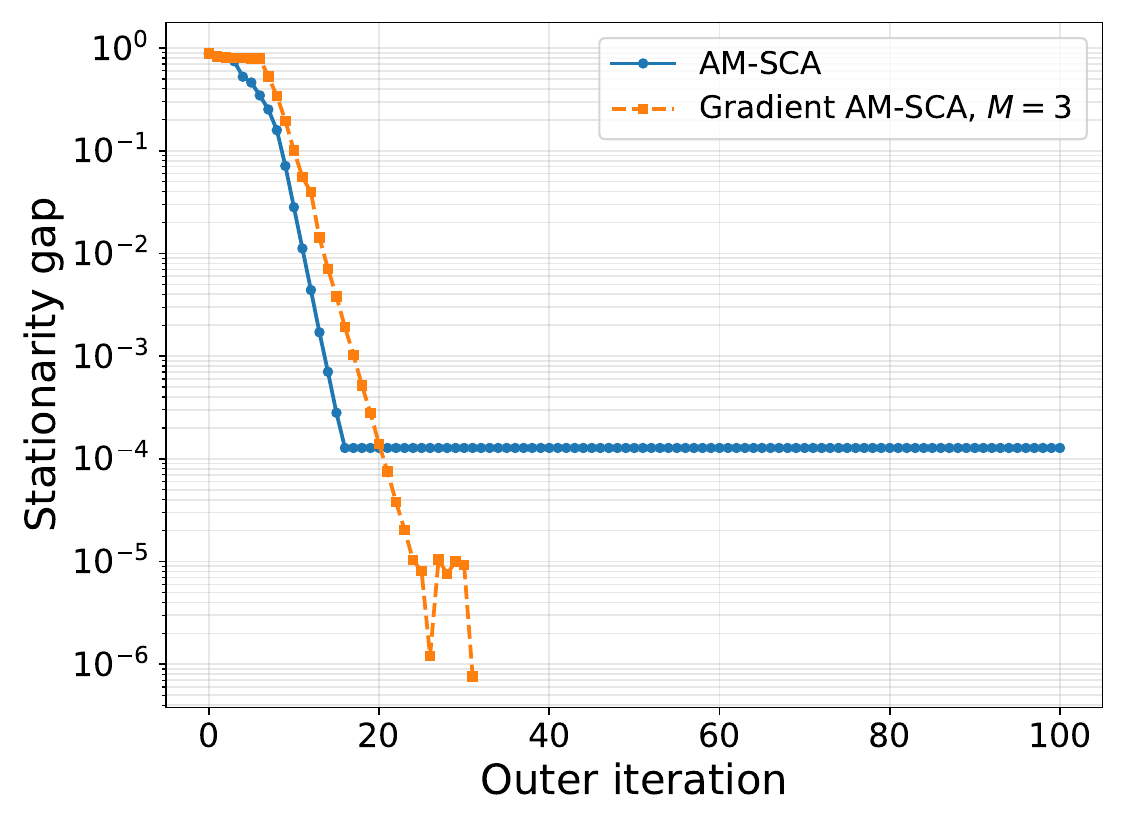}
\vspace{-6pt}\caption{Convergence of the proposed AM-SCA methods.}
\label{fig_aoi_convergence}
\end{figure}
\begin{figure}[tbp]
\centering
\includegraphics[width=0.6\textwidth]{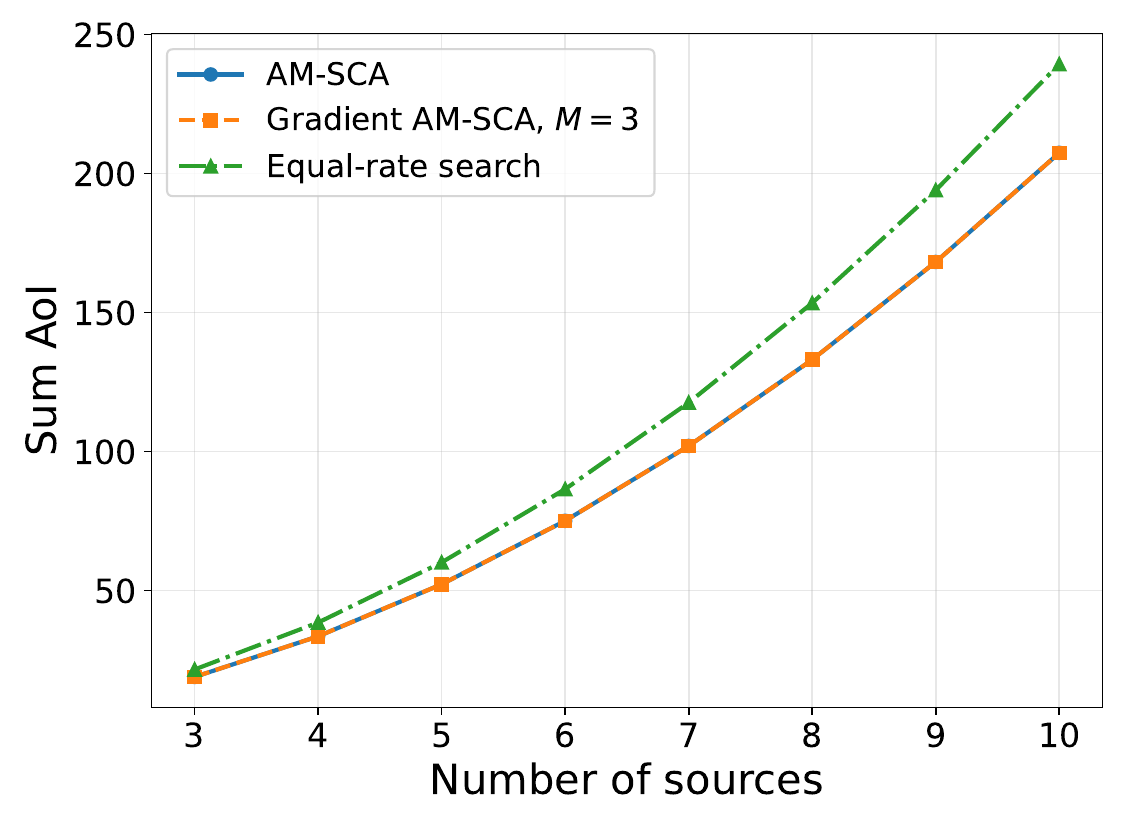}
\vspace{-6pt}\caption{Sum AoI under different numbers of source nodes.}
\label{fig_aoi_k}
\end{figure}
Fig.~\ref{fig_aoi_convergence} shows the convergence behavior of the proposed AM-SCA methods. Both AM-SCA and gradient AM-SCA reduce the stationarity gap effectively, confirming that the proposed AM-GM transform can guide the iterates toward a stationary point. The gradient-based AM-SCA with only \(M=3\) inner gradient steps reaches a comparable, and even smaller, final stationarity gap, which indicates that accurately solving each surrogate subproblem is not always necessary in this example. Fig.~\ref{fig_aoi_k} further compares the sum AoI under different numbers of source nodes. The proposed AM-SCA and gradient AM-SCA achieve nearly identical AoI values, and both consistently outperform the equal-rate search baseline. The performance gap becomes more visible as the number of sources increases, showing that the proposed SCA methods can better exploit the heterogeneous update-rate structure than the simple equal-rate allocation.

\subsection{Semantic Information Utility Maximization in Wireless Networks}\label{sec_semantic}

We next consider a semantic communication system, where users transmit semantic information to an edge server over wireless links. Different from conventional bit-level transmission, the utility of semantic communication depends not only on the physical transmission rate, but also on the semantic compression and semantic recovery quality. This naturally leads to coupled multiplicative and fractional terms.

\subsubsection{Problem Statement}

Consider a system with \(N\) users and one edge server. Let \(b_n\), \(p_n\), and \(s_n\) denote the bandwidth, transmit power, and semantic encoding level of user \(n\), respectively. The transmission rate of user \(n\) is given by
\begin{equation}
R_n(b_n,p_n)
=
b_n\log_2\left(1+\frac{g_np_n}{b_n\sigma^2}\right),
\end{equation}
where \(g_n\) is the channel gain and \(\sigma^2\) is the noise power spectral density. The semantic recovery quality is modeled as
\begin{equation}
Q_n(s_n)=1-e^{-\kappa_n s_n},
\end{equation}
where \(\kappa_n>0\) measures the semantic extraction efficiency. The required semantic data size after encoding is denoted by
\begin{equation}
D_n(s_n)=D_n^{0}(1+\theta_n s_n),
\end{equation}
where \(D_n^{0}>0\) is the basic data size and \(\theta_n>0\) captures the additional semantic description overhead.

The semantic information utility of user \(n\) is defined as
\begin{equation}
\log\left(
1+
\frac{Q_n(s_n)R_n(b_n,p_n)}{D_n(s_n)}
\right).
\end{equation}
Then, the semantic information utility maximization problem is formulated as
\begin{subequations}\label{prob_semantic}
\begin{align}
\max_{\bm{b},\bm{p},\bm{s}}\quad
&
\sum_{n=1}^{N}
w_n
\log\left(
1+
\frac{Q_n(s_n)R_n(b_n,p_n)}{D_n(s_n)}
\right) \\
\text{s.t.}\quad
&
\sum_{n=1}^{N} b_n \le B_{\max}, \\
&
\sum_{n=1}^{N} p_n \le P_{\max}, \\
&
0<b_n\le b_{\max},\quad 0<p_n\le p_{\max},\quad \forall n,\\
&
s_{\min}\le s_n\le s_{\max},\quad \forall n .
\end{align}
\end{subequations}
Problem~\eqref{prob_semantic} is non-convex because the objective contains the logarithm of a coupled product-ratio term involving semantic quality, transmission rate, and semantic data size.

\subsubsection{Problem Transformation}

Define
\begin{equation}
S_n(\bm{b},\bm{p},\bm{s})
=
\frac{Q_n(s_n)R_n(b_n,p_n)}{D_n(s_n)}.
\end{equation}
This term can be regarded as a product of three positive functions:
\begin{equation}
S_n(\bm{b},\bm{p},\bm{s})
=
f_n^{(1)}(s_n)
f_n^{(2)}(b_n,p_n)
f_n^{(3)}(s_n),
\end{equation}
where
\begin{equation}
f_n^{(1)}(s_n)=Q_n(s_n),\quad
f_n^{(2)}(b_n,p_n)=R_n(b_n,p_n),\quad
f_n^{(3)}(s_n)=\frac{1}{D_n(s_n)}.
\end{equation}

For the maximization problem, we use the HM-GM lower bound. At iteration \(i\), define
\begin{align}
\tilde f_n^{(1)}
&=
y_n^{(1)}y_n^{(2)}
\left(f_n^{(1)}(s_n)\right)^3,\\
\tilde f_n^{(2)}
&=
\frac{y_n^{(2)}}{y_n^{(1)}}
\left(f_n^{(2)}(b_n,p_n)\right)^3,\\
\tilde f_n^{(3)}
&=
\frac{
\left(f_n^{(3)}(s_n)\right)^3
}
{
\left(y_n^{(2)}\right)^2
}.
\end{align}
Then the HM lower bound of \(S_n(\bm{b},\bm{p},\bm{s})\) is
\begin{equation}
F_n^{\mathrm{HM}}
=
\frac{3}{
\frac{1}{\tilde f_n^{(1)}}+
\frac{1}{\tilde f_n^{(2)}}+
\frac{1}{\tilde f_n^{(3)}}
}.
\end{equation}
Equation (\ref{y_n_closedform}) gives
\begin{align}
y_n^{(1)}
&=
\sqrt{
\frac{
f_n^{(2)}(b_n,p_n)
}{
f_n^{(1)}(s_n)
}
},\\
y_n^{(2)}
&=
\left(y_n^{(1)}\right)^{1/3}
\left(
\frac{
f_n^{(3)}(s_n)
}{
f_n^{(2)}(b_n,p_n)
}
\right)^{1/3}.
\end{align}
At the current feasible point
\((\bm{b}^{i},\bm{p}^{i},\bm{s}^{i})\), the auxiliary variables are updated by the above equations.
Since \(\log(1+z)\) is increasing, we have
\begin{equation}
\log\left(1+F_n^{\mathrm{HM}}\right)
\le
\log\left(1+S_n(\bm{b},\bm{p},\bm{s})\right).
\end{equation}
Therefore, the original problem can be approximated by the following SA subproblem:
\begin{subequations}\label{prob_semantic_trans}
\begin{align}
\max_{\bm{b},\bm{p},\bm{s},\bm{y}}\quad
&
\sum_{n=1}^{N}
w_n
\log\left(
1+
F_n^{\mathrm{HM}}(\bm{b},\bm{p},\bm{s},\bm{y}_n^{i})
\right) \\
\text{s.t.}\quad
&
\sum_{n=1}^{N} b_n \le B_{\max}, \\
&
\sum_{n=1}^{N} p_n \le P_{\max}, \\
&
0<b_n\le b_{\max},\quad 0<p_n\le p_{\max},\quad \forall n,\\
&
s_{\min}\le s_n\le s_{\max},\quad \forall n .
\end{align}
\end{subequations}
The surrogate objective is a tight lower bound of the original objective at the current iterate, and thus can be used in an SA procedure.

\subsubsection{SA Method Based on the Proposed HM-GM Transform}

The proposed HM-GM transform provides a tight lower-bound surrogate for the coupled positive term inside the logarithm. However, the resulting transformed objective is not guaranteed to be concave in general. Therefore, instead of solving the transformed subproblem exactly, we apply the proposed gradient-based HM-SA method.

At iteration \(i\), the auxiliary variables are first updated according to Equation (\ref{y_n_closedform}) at the current feasible point. Then, with the auxiliary variables fixed, several projected gradient ascent steps are performed on the transformed surrogate objective. The updated point is used to construct the next surrogate. Since the surrogate is tight and first-order consistent at the current iterate, the resulting gradient-based SA method can be used to approach an \(\epsilon\)-stationary point under the standard smoothness and boundedness assumptions discussed in Section~\ref{sec_max_sca}.

\subsubsection{Parameter Setting}

Unless otherwise specified, the number of users is set as \(N=20\). The total bandwidth is \(B_{\max}=10\) MHz, and the total power budget is \(P_{\max}=10\) W. The maximum bandwidth and power of each user are set as \(b_{\max}=2\) MHz and \(p_{\max}=1\) W. The noise power spectral density is set as \(-134\) dBm/Hz. The channel gain \(g_n\) is generated according to a distance-dependent path-loss model with small-scale Rayleigh fading. The semantic parameters \(\kappa_n\), \(D_n^0\), and \(\theta_n\) are randomly generated from prescribed positive intervals. The semantic encoding level satisfies \(s_{\min}\le s_n\le s_{\max}\), where \(s_{\min}=0.1\) and \(s_{\max}=1\). The initial point is set by equal bandwidth allocation, equal power allocation, and \(s_n=(s_{\min}+s_{\max})/2\).

We compare the proposed gradient HM-SA methods with the following baselines:
\begin{itemize}
\item \textbf{Equal allocation}: bandwidth and power are equally allocated among users, and the semantic encoding level is fixed at the midpoint.
\item \textbf{Communication-only allocation}: the semantic encoding level is fixed, and only bandwidth and power are optimized.
\item \textbf{Direct gradient ascent}: projected gradient ascent is directly applied to the original objective without using the proposed HM-GM surrogate.
\end{itemize}

\subsubsection{Numerical Results}
\begin{figure}[tbp]
\centering
\includegraphics[width=0.6\textwidth]{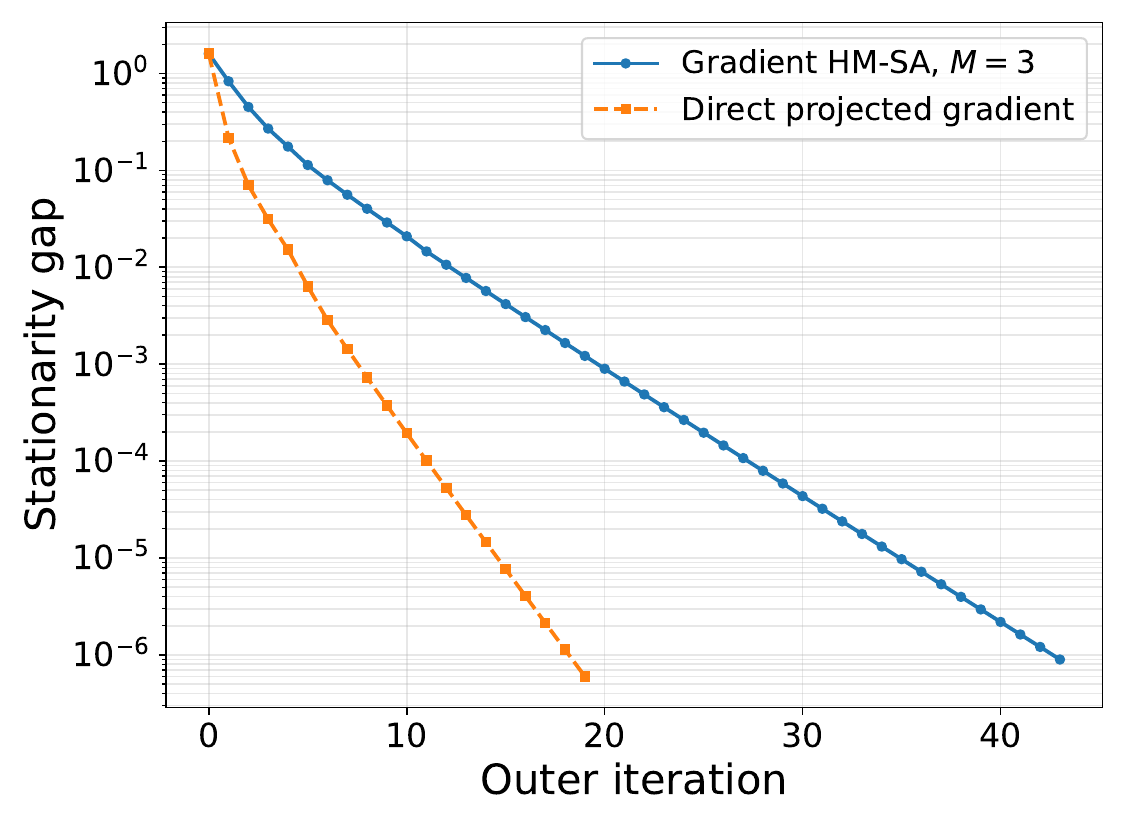}
\vspace{-6pt}\caption{Convergence of the proposed HM-SA methods and gradient baseline.}
\label{fig_semantic_convergence}
\end{figure}
\begin{figure}[tbp]
\centering
\includegraphics[width=0.6\textwidth]{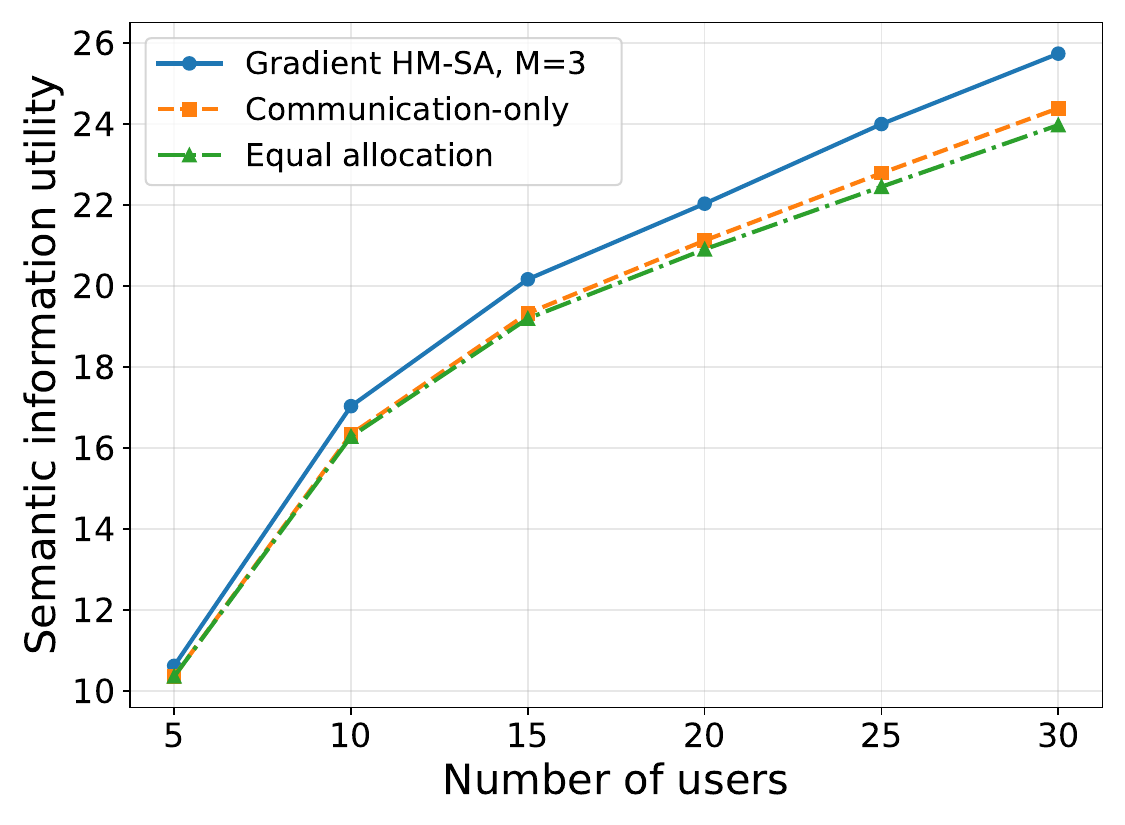}
\vspace{-6pt}\caption{Semantic information utility under different numbers of users.}
\label{fig_semantic_user}
\end{figure}
Fig.~\ref{fig_semantic_convergence} shows the stationarity-gap behavior of the gradient-based HM-SA method and the direct projected-gradient baseline. Since the transformed logarithmic HM surrogate is not guaranteed to be concave, the proposed method is implemented through an inexact gradient-based SA update. Both methods reduce the stationarity gap effectively and reach a small stationarity level. Direct projected gradient converges faster in this smooth example because it directly updates the original objective. In contrast, gradient HM-SA optimizes a conservative lower-bound surrogate of the coupled product-ratio term, and thus its updates are more cautious. Nevertheless, the stationarity gap of gradient HM-SA decreases steadily to around \(10^{-6}\), confirming that the proposed HM-GM surrogate can provide a valid and effective optimization procedure for the considered non-concave semantic-utility problem.

Fig.~\ref{fig_semantic_user} compares the achieved semantic information utility under different numbers of users. The proposed gradient HM-SA method consistently achieves higher utility than both the communication-only and equal-allocation baselines. This demonstrates the benefit of jointly optimizing semantic encoding, bandwidth, and transmit power, rather than optimizing only communication resources or using fixed equal allocation. The performance gain becomes more visible as the number of users increases, because resource competition and semantic-resource coupling become stronger. These results show that the proposed HM-GM transform can be used for information-centric utility maximization involving logarithmic product-ratio structures.

\subsection{Reliability-Aware Multi-Hop Information Delivery}\label{sec_reliability}

We next consider a reliability-aware information delivery problem in a multi-hop communication network. In many information services, e.g., sensor monitoring, emergency message dissemination, and industrial control, information is delivered through multiple wireless links. The end-to-end reliability of a route is naturally modeled as the product of the successful transmission probabilities of all links on that route. Therefore, the resulting utility maximization problem contains sum-of-products terms and fits the proposed HM-GM-based SA framework.

\subsubsection{Problem Statement}

Consider a communication network with \(L\) wireless links and \(R\) information flows. Each flow \(r\) is transmitted through a fixed route \(\mathcal{P}_r\), where \(\mathcal{P}_r\subseteq\{1,\ldots,L\}\) denotes the set of links used by flow \(r\). Let \(p_\ell\) denote the transmit power allocated to link \(\ell\). The successful transmission probability of link \(\ell\) is modeled as
\begin{equation}
P_\ell(p_\ell)=1-\exp(-a_\ell p_\ell),
\end{equation}
where \(a_\ell>0\) captures the channel condition of link \(\ell\). The end-to-end delivery reliability of flow \(r\) is then given by
\begin{equation}
\mathcal{R}_r(\bm{p})=
\prod_{\ell\in\mathcal{P}_r} P_\ell(p_\ell).
\end{equation}
The reliability-aware information utility maximization problem is formulated as
\begin{subequations}\label{prob_reliability}
\begin{align}
\max_{\bm{p}}\quad &
\sum_{r=1}^{R} w_r
\prod_{\ell\in\mathcal{P}_r} P_\ell(p_\ell) \\
\text{s.t.}\quad &
\sum_{\ell=1}^{L} p_\ell \le P_{\max},\\
&
0\le p_\ell\le p_{\max},\quad \forall \ell,
\end{align}
\end{subequations}
where \(w_r>0\) is the priority weight of flow \(r\). Problem~\eqref{prob_reliability} is non-convex because the objective contains products of link reliability functions.

\subsubsection{Problem Transformation}

For each route \(r\), denote the number of links on the route as \(K_r=|\mathcal{P}_r|\). To simplify notation, we rewrite the route reliability as
\begin{equation}
\mathcal{R}_r(\bm{p})
=
\prod_{k=1}^{K_r} f_r^{(k)}(\bm{p}),
\end{equation}
where \(f_r^{(k)}(\bm{p})\) denotes the success probability of the \(k\)-th link on route \(\mathcal{P}_r\).
Since Problem~\eqref{prob_reliability} is a maximization problem, we apply the proposed HM-GM transform to construct a lower-bound surrogate for each product term. Following the general construction, define the transformed terms as
\begin{equation}
\tilde f_r^{(1)}
=
f_r^{(1)}(\bm{p})\prod_{k=1}^{K_r-1} y_r^{(k)},
\end{equation}
\begin{equation}
\tilde f_r^{(k)}
=
f_r^{(k)}(\bm{p})
\frac{\prod_{i=k}^{K_r-1}y_r^{(i)}}{\left(y_r^{(k-1)}\right)^{k-1}},
\quad k=2,\ldots,K_r-1,
\end{equation}
and
\begin{equation}
\tilde f_r^{(K_r)}
=
\frac{f_r^{(K_r)}(\bm{p})}{\left(y_r^{(K_r-1)}\right)^{K_r-1}}.
\end{equation}
Then the HM lower bound of the route reliability is
\begin{equation}
F_r^{\mathrm{HM}}(\bm{p},\bm{y}_r)
=
\frac{K_r}{\sum_{k=1}^{K_r}\frac{1}{\tilde f_r^{(k)}}}.
\end{equation}
At the current feasible point \(\bm{p}^{i}\), the auxiliary variables are updated by Equation (\ref{y_n_closedform}) of the proposed transform:
\begin{equation}
y_r^{(1,i)}
=
\sqrt{
\frac{
f_r^{(2)}(\bm{p}^{i})
}{
f_r^{(1)}(\bm{p}^{i})
}
},
\end{equation}
and
\begin{equation}
y_r^{(k-1,i)}
=
\sqrt[k]{
\left(y_r^{(k-2,i)}\right)^{k-2}
\frac{
f_r^{(k)}(\bm{p}^{i})
}{
f_r^{(k-1)}(\bm{p}^{i})
}
},
\quad k=3,\ldots,K_r .
\end{equation}
Therefore, at iteration \(i\), Problem~\eqref{prob_reliability} is approximated by
\begin{subequations}\label{prob_reliability_trans}
\begin{align}
\max_{\bm{p}}\quad &
\sum_{r=1}^{R} w_r
F_r^{\mathrm{HM}}(\bm{p},\bm{y}_r^{i}) \\
\text{s.t.}\quad &
\sum_{\ell=1}^{L} p_\ell \le P_{\max},\\
&
0\le p_\ell\le p_{\max},\quad \forall \ell .
\end{align}
\end{subequations}
The transformed objective is a tight lower-bound surrogate of the original objective at the current point. However, it is not guaranteed to be concave in general. Therefore, we use the gradient-based HM-SA method to solve this problem.

\subsubsection{SA Method Based on the Proposed HM-GM Transform}

The proposed gradient-based HM-SA method proceeds as follows. First, initialize a feasible power allocation vector \(\bm{p}^{0}\). At iteration \(i\), update the auxiliary variables \(\bm{y}_r^{i}\) for all routes according to Equation (\ref{y_n_closedform}). Then, with \(\bm{y}_r^{i}\) fixed, perform several projected gradient ascent steps on the surrogate objective in Problem~\eqref{prob_reliability_trans}. The obtained point is denoted as \(\bm{p}^{i+1}\). Since the HM surrogate is tight and first-order consistent at \(\bm{p}^{i}\), the proposed gradient-based SA method can be used to approach an \(\epsilon\)-stationary point under the smoothness and boundedness assumptions discussed in Section~III.

\subsubsection{Parameter Setting}

In the simulation, the number of wireless links is set as \(L=20\), and the number of information flows is varied from \(R=5\) to \(R=30\). Each flow is assigned a route with \(K_r\) links randomly selected from the link set. The channel parameter \(a_\ell\) is generated from a uniform distribution over a positive interval, and the flow weight \(w_r\) is randomly generated around one. The total power budget is \(P_{\max}=10\) W, and the maximum power of each link is \(p_{\max}=1\) W. The initial power allocation is set as the equal-power allocation. For the gradient-based HM-SA method, the number of inner projected gradient steps is set as \(M=3\). The maximum number of outer iterations is set as 100, and the stopping tolerance is set according to the stationarity gap.

We compare the proposed method with two baselines. The first is the equal-power allocation scheme, where all links are assigned the same feasible power. The second is direct projected gradient ascent, which directly optimizes the original non-convex objective without using the proposed HM-GM surrogate.

\subsubsection{Numerical Results}
\begin{figure}[tbp]
\centering
\includegraphics[width=0.6\textwidth]{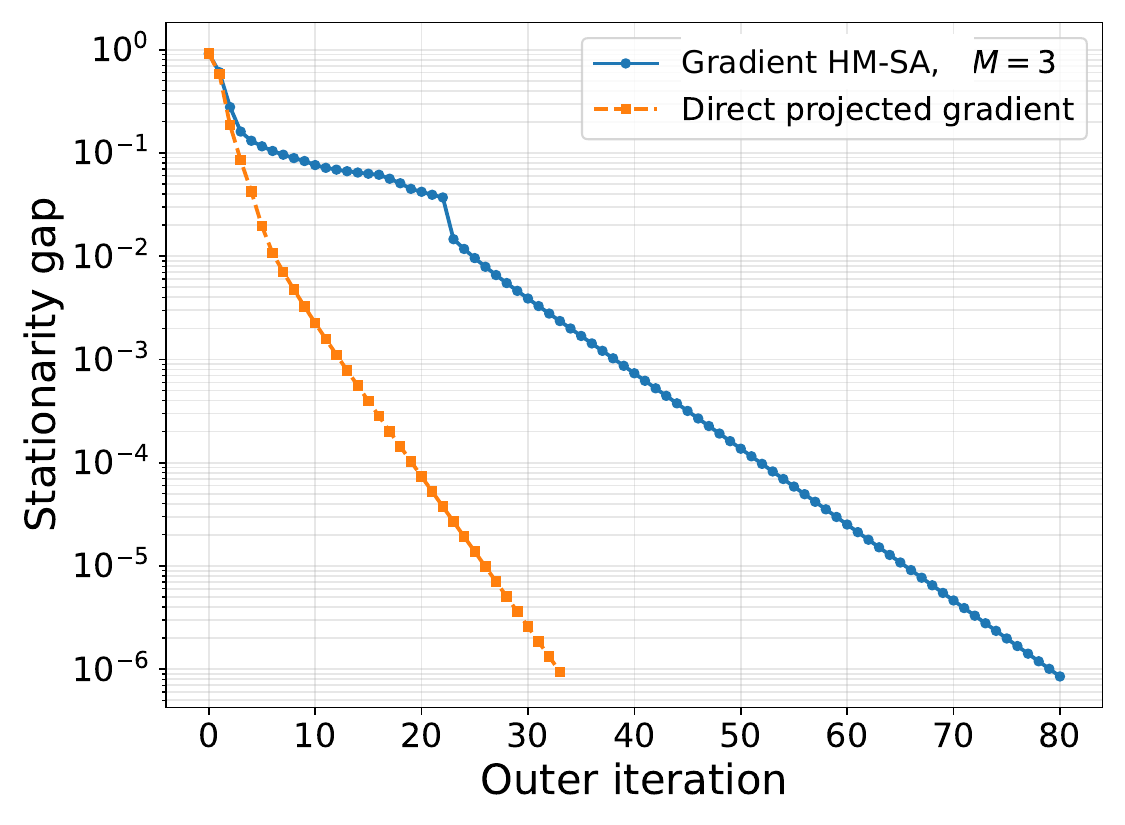}
\vspace{-6pt}\caption{Convergence of the proposed HM-SA methods and gradient baseline.}
\label{fig_reliability_convergence}
\end{figure}
\begin{figure}[tbp]
\centering
\includegraphics[width=0.6\textwidth]{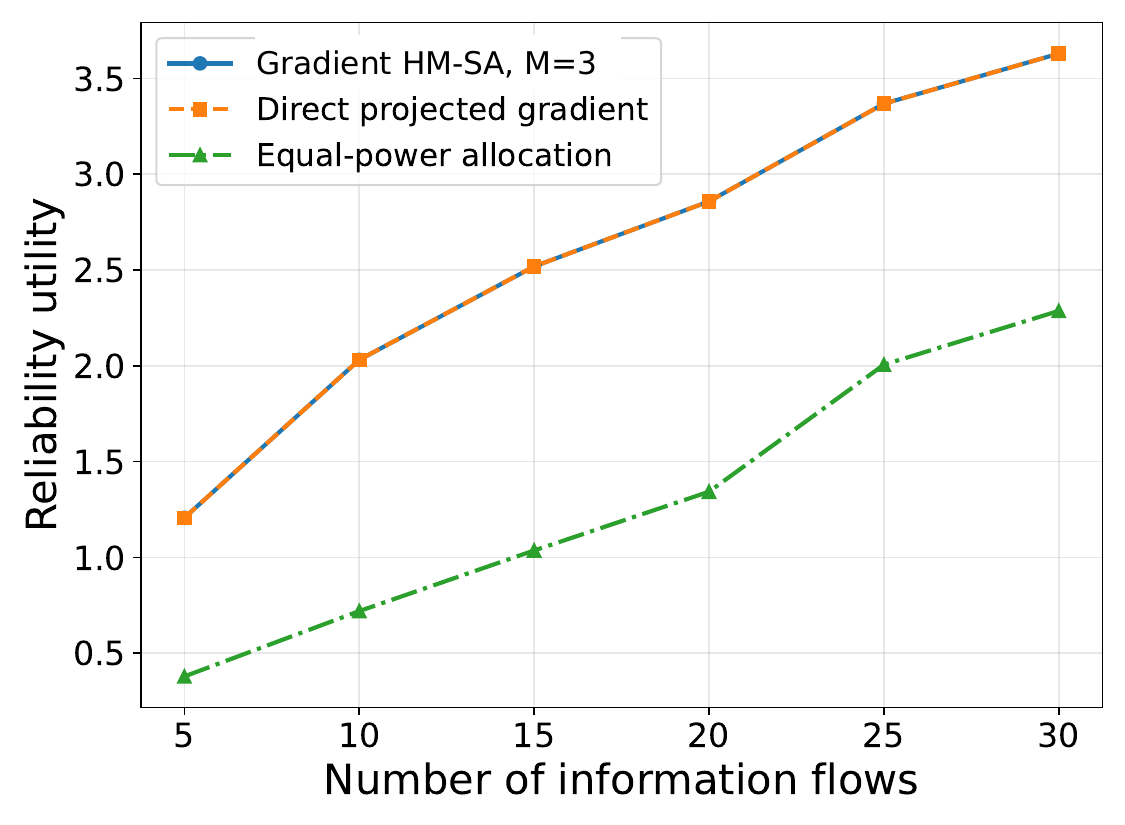}
\vspace{-6pt}\caption{Information utility under different numbers of information flows.}
\label{fig_reliability_flow}
\end{figure}
Fig.~\ref{fig_reliability_convergence} shows the convergence behavior of the proposed gradient HM-SA method and the direct projected-gradient baseline. Both methods reduce the stationarity gap and eventually reach a small stationarity level. The direct projected-gradient method converges faster in this example, which is reasonable because it directly optimizes the original smooth objective. In contrast, the proposed gradient HM-SA method optimizes a conservative HM lower-bound surrogate of the product reliability term, and therefore its updates are more cautious. Nevertheless, the stationarity gap of gradient HM-SA still decreases steadily and reaches around \(10^{-6}\), confirming the effectiveness of the proposed surrogate-based optimization procedure.

In Fig.~\ref{fig_reliability_flow}, we compare the reliability utility under different numbers of information flows. The proposed gradient HM-SA method consistently improves upon the equal-power allocation baseline, showing that adaptive power allocation is important for multi-hop information delivery. The direct projected-gradient method achieves higher utility in this smooth low-dimensional example, but it does not use the product-aware surrogate structure. By contrast, the proposed HM-SA method provides a principled SA interpretation for sum-of-products reliability maximization, where the end-to-end reliability is explicitly handled through a tight HM-GM lower bound. These results demonstrate that the proposed transform is applicable to information reliability optimization and can effectively improve the reliability utility over simple non-adaptive allocation schemes.

\subsection{Cooperative Edge Caching for Information Availability}\label{sec_cache}

We next consider a cooperative edge caching problem in content-centric communication networks. In such systems, user requests can be served by nearby edge caches. If the requested content is not stored in any accessible cache, the request is forwarded to the remote server, which increases backhaul traffic and delivery delay. Therefore, improving information availability at the edge is an important objective in information-centric networking.

\subsubsection{Problem Statement}

Consider a network with \(M\) edge caches, \(U\) users, and \(F\) contents. Let \(q_{m,f}\in[0,1]\) denote the caching probability of content \(f\) at cache \(m\). For user \(u\), denote by \(\mathcal{A}_u\) the set of accessible edge caches. The probability that user \(u\)'s request for content \(f\) cannot be served by any accessible cache is
$\prod_{m\in\mathcal{A}_u}(1-q_{m,f})$.
Let \(\pi_{u,f}\) denote the request probability of user \(u\) for content \(f\). The expected cache-miss probability minimization problem is formulated as
\begin{subequations}\label{prob_cache}
\begin{align}
\min_{\bm{q}}\quad &
\frac{1}{U}\sum_{u=1}^{U}\sum_{f=1}^{F}
\pi_{u,f}
\prod_{m\in\mathcal{A}_u}(1-q_{m,f}) \\
\text{s.t.}\quad
&
\sum_{f=1}^{F} q_{m,f}\le C_m,\quad \forall m,\\
&
0\le q_{m,f}\le 1-\epsilon_q,\quad \forall m,f,
\end{align}
\end{subequations}
where \(C_m\) is the cache capacity of edge cache \(m\), and \(\epsilon_q>0\) is a small constant used to keep \(1-q_{m,f}\) away from zero. Problem~\eqref{prob_cache} is non-convex because the objective contains products of cache-miss probabilities across multiple accessible caches.

\subsubsection{Problem Transformation}
For each user-content pair \((u,f)\), define \(K_u=|\mathcal{A}_u|\), and index the caches in \(\mathcal{A}_u\) as \(m_{u,1},m_{u,2},\ldots,m_{u,K_u}\). Then the cache-miss probability can be rewritten as
\begin{equation}
\prod_{m\in\mathcal{A}_u}(1-q_{m,f})
=
\prod_{k=1}^{K_u} f_{u,f}^{(k)}(\bm{q}),
\end{equation}
where
\begin{equation}
f_{u,f}^{(k)}(\bm{q})=1-q_{m_{u,k},f}.
\end{equation}
Since Problem~\eqref{prob_cache} is a minimization problem, we use the proposed AM-GM transform to construct an upper-bound surrogate for each product term. At iteration \(i\), define
\begin{equation}
\tilde f_{u,f}^{(1)}
=
\left(f_{u,f}^{(1)}(\bm{q})\right)^{K_u}
\prod_{k=1}^{K_u-1}y_{u,f}^{(k)},
\end{equation}
\begin{equation}
\tilde f_{u,f}^{(k)}
=
\left(f_{u,f}^{(k)}(\bm{q})\right)^{K_u}
\frac{\prod_{j=k}^{K_u-1}y_{u,f}^{(j)}}
{\left(y_{u,f}^{(k-1)}\right)^{k-1}},
\quad k=2,\ldots,K_u-1,
\end{equation}
and
\begin{equation}
\tilde f_{u,f}^{(K_u)}
=
\frac{\left(f_{u,f}^{(K_u)}(\bm{q})\right)^{K_u}}
{\left(y_{u,f}^{(K_u-1)}\right)^{K_u-1}}.
\end{equation}
Then, the AM upper bound of the product term is
\begin{equation}
F_{u,f}^{\rm AM}(\bm{q},\bm{y}_{u,f})
=
\frac{1}{K_u}
\sum_{k=1}^{K_u}
\tilde f_{u,f}^{(k)}.
\end{equation}
Based on Equation (\ref{y_n_closedform}), it gives
\begin{equation}
y_{u,f}^{(1,i)}
=
\left(
\frac{
f_{u,f}^{(2)}(\bm{q}^{i})
}{
f_{u,f}^{(1)}(\bm{q}^{i})
}
\right)^{\frac{K_u}{2}},
\end{equation}
and
\begin{equation}
y_{u,f}^{(k-1,i)}
=
\left[
\left(y_{u,f}^{(k-2,i)}\right)^{k-2}
\left(
\frac{
f_{u,f}^{(k)}(\bm{q}^{i})
}{
f_{u,f}^{(k-1)}(\bm{q}^{i})
}
\right)^{K_u}
\right]^{\frac{1}{k}},
\quad k=3,\ldots,K_u .
\end{equation}
With fixed auxiliary variables, Problem~\eqref{prob_cache} is approximated by
\begin{subequations}\label{prob_cache_trans}
\begin{align}
\min_{\bm{q}}\quad &
\frac{1}{U}\sum_{u=1}^{U}\sum_{f=1}^{F}
\pi_{u,f}
F_{u,f}^{\rm AM}(\bm{q},\bm{y}_{u,f}^{i}) \\
\text{s.t.}\quad
&
\sum_{f=1}^{F} q_{m,f}\le C_m,\quad \forall m,\\
&
0\le q_{m,f}\le 1-\epsilon_q,\quad \forall m,f.
\end{align}
\end{subequations}
For fixed \(\bm{y}^{i}\), each transformed term contains \((1-q_{m,f})^{K_u}\), which is convex over the feasible interval. Therefore, Problem~\eqref{prob_cache_trans} is a convex optimization problem with linear constraints.

\subsubsection{SCA Method Based on the Proposed AM-GM Transform}

The proposed AM-SCA method for Problem~\eqref{prob_cache} proceeds as follows. First, initialize a feasible caching probability matrix \(\bm{q}^{0}\). At iteration \(i\), the auxiliary variables \(\bm{y}_{u,f}^{i}\) are updated according to Equation (\ref{y_n_closedform}) at \(\bm{q}^{i}\). Then, the convex surrogate problem~\eqref{prob_cache_trans} is solved to obtain \(\bm{q}^{i+1}\). The procedure is repeated until convergence.

In this application, the AM surrogate admits a simple closed-form update after collecting the coefficients of each \((1-q_{m,f})^{K_u}\) term. Specifically, for each cache \(m\), the transformed subproblem can be written as
\begin{equation}
\min_{\bm{q}}
\sum_{f=1}^{F} A_{m,f}^{i}(1-q_{m,f})^{K}
\end{equation}
subject to
\begin{equation}
\sum_{f=1}^{F}q_{m,f}\le C_m,\quad
0\le q_{m,f}\le 1-\epsilon_q,
\end{equation}
where \(A_{m,f}^{i}\ge 0\) is determined by the request probabilities, the user-cache association, and the auxiliary variables at iteration \(i\). In the simulation, each user accesses \(K=3\) caches. The KKT condition gives
\begin{equation}
q_{m,f}^{i+1}
=
\left[
1-
\sqrt{
\frac{\lambda_m}{3A_{m,f}^{i}}
}
\right]_{0}^{1-\epsilon_q},
\end{equation}
where \([\cdot]_{0}^{1-\epsilon_q}\) denotes projection onto \([0,1-\epsilon_q]\), and \(\lambda_m\ge0\) is chosen such that the cache-capacity constraint is satisfied when it is active. This closed-form update makes the AM-SCA method efficient for the considered caching problem.

\subsubsection{Parameter Setting}

In the simulation, the number of edge caches is set as \(M=6\), the number of users is set as \(U=12\), and the number of contents is set as \(F=20\). Each user can access \(K=3\) nearby caches. The content popularity follows a Zipf distribution with skewness parameter \(\zeta=0.8\). The cache capacity is set to the same value \(C\) for all caches, and \(C\) is varied from \(2\) to \(8\). The upper bound of the caching probability is set as \(q_{m,f}\le 0.95\), i.e., \(\epsilon_q=0.05\), to avoid numerical instability when \(1-q_{m,f}\) is close to zero. The initial caching probability is set by equal allocation. The maximum number of outer iterations is set as 150, and the stopping tolerance is set according to the projected stationarity gap.

We compare the proposed AM-SCA method with three baselines. The first is the direct projected gradient, which directly applies projected gradient descent to the original non-convex cache-miss objective. The second is popularity-based caching, where each cache gives priority to the most popular contents requested by its accessible users. The third is equal caching, where each cache allocates the same caching probability to all contents.

\subsubsection{Numerical Results}

\begin{figure}[tbp]
\centering
\includegraphics[width=0.6\textwidth]{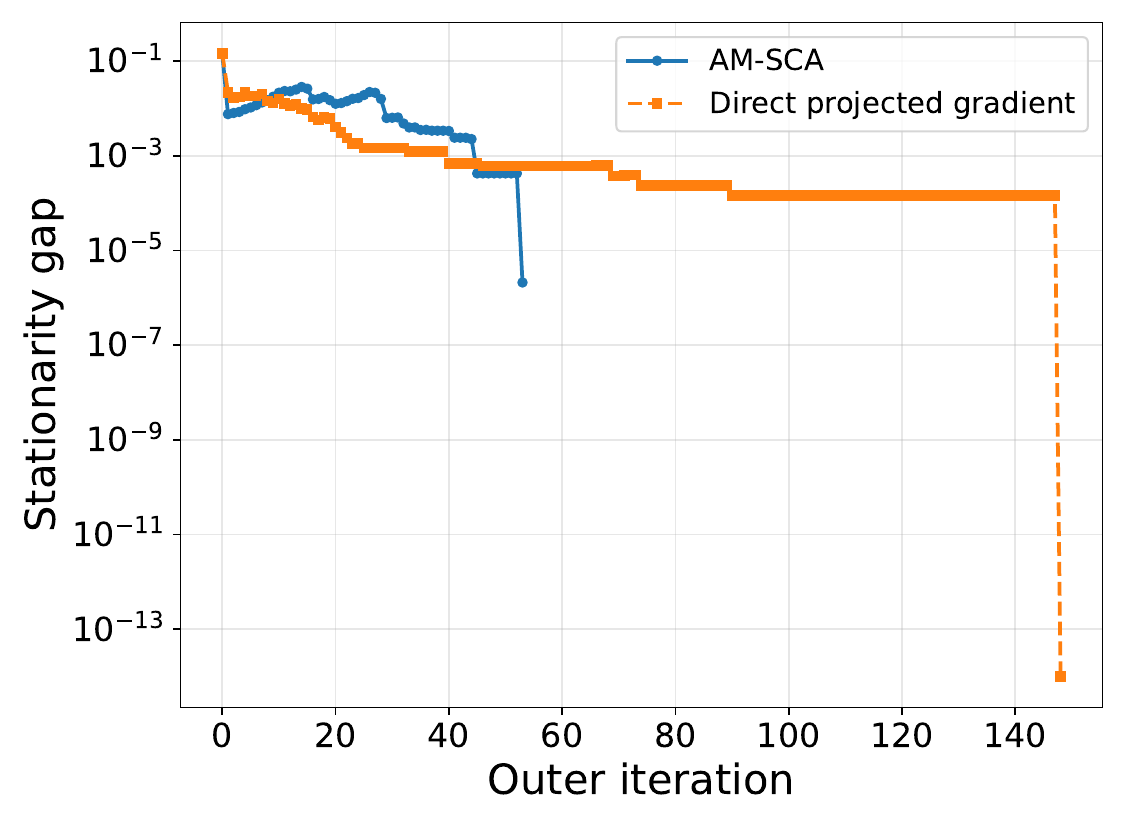}
\vspace{-6pt}
\caption{Convergence of the proposed AM-SCA method and gradient baseline.}
\label{fig_cache_convergence}
\end{figure}

\begin{figure}[tbp]
\centering
\includegraphics[width=0.6\textwidth]{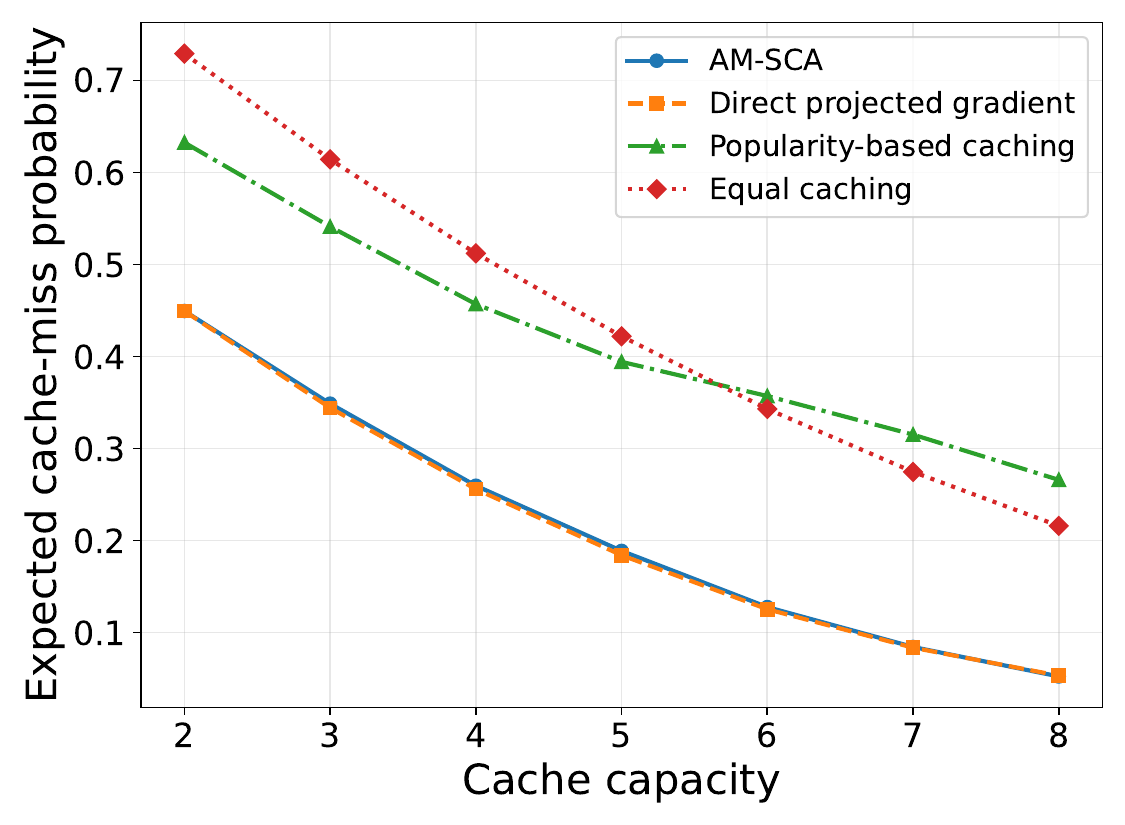}
\vspace{-6pt}
\caption{Expected cache-miss probability under different cache capacities.}
\label{fig_cache_capacity}
\end{figure}

Fig.~\ref{fig_cache_convergence} presents the stationarity-gap convergence of AM-SCA and direct projected gradient. AM-SCA quickly reduces the stationarity gap and reaches a small stationarity level within fewer outer iterations. Direct projected gradient also converges, but it requires more iterations to reach a comparable stationarity level. This shows that the AM-GM surrogate provides an effective structured update for the cooperative caching problem.

In Fig.~\ref{fig_cache_capacity}, we show the expected cache-miss probability under different cache capacities. As the cache capacity increases, the cache-miss probability decreases for all methods. The proposed AM-SCA method achieves performance comparable to direct projected gradient and clearly outperforms popularity-based caching and equal caching. This improvement comes from jointly optimizing the caching probabilities across users, contents, and accessible cache sets. In contrast, the heuristic baselines only use content popularity or uniform allocation, and therefore cannot fully exploit the cooperative caching structure. These results demonstrate that the proposed AM-GM transform is useful for information-availability optimization problems whose objectives contain sum-of-products terms.

\subsection{Multi-Modal Product-Loss Learning}\label{sec_ml_prodloss}

We next consider a machine learning application in multi-modal supervised learning. In many learning tasks, each sample can be described by multiple modalities or feature views, such as text, image, audio, and sensor measurements. A common goal is to learn predictors from multiple modalities jointly. When the learning objective is designed to penalize samples that are difficult across multiple modalities, product-type loss functions naturally arise.

\subsubsection{Problem Statement}

Consider a supervised learning dataset with \(N\) training samples. Each sample \(i\) has \(K\) modalities, and the feature vector of modality \(k\) is denoted by \(a_i^{(k)}\). For simplicity, we consider binary classification with label \(z_i\in\{-1,+1\}\). Let \(w_k\) be the linear classifier associated with modality \(k\). The logistic loss of modality \(k\) on sample \(i\) is
\begin{equation}
\ell_i^{(k)}(w_k)
=
\log\left(1+\exp\left(-z_i (a_i^{(k)})^\top w_k\right)\right)+\epsilon_\ell,
\end{equation}
where \(\epsilon_\ell>0\) is a small constant to ensure strict positivity. To jointly learn all modality-specific classifiers, we consider the following product-loss minimization problem:
\begin{subequations}\label{prob_ml_product}
\begin{align}
\min_{\bm{w}}\quad
&
\frac{1}{N}
\sum_{i=1}^{N}
\prod_{k=1}^{K}
\ell_i^{(k)}(w_k)
+
\frac{\lambda}{2}
\sum_{k=1}^{K}
\|w_k\|_2^2 \\
\text{s.t.}\quad
&
w_k\in\mathcal{W}_k,\quad k=1,\ldots,K,
\end{align}
\end{subequations}
where \(\lambda>0\) is the regularization parameter and \(\mathcal{W}_k\) is a compact convex feasible set. Problem~\eqref{prob_ml_product} is non-convex because the empirical loss contains products of modality-specific loss functions.

The product-loss form is useful when the learning objective emphasizes samples that remain difficult across multiple modalities. If at least one modality already fits a sample well, the product loss becomes smaller. In contrast, if all modalities incur large losses on a sample, the product term becomes large and encourages joint improvement across modalities.

\subsubsection{Problem Transformation}

For each sample \(i\), define
\begin{equation}
f_i^{(k)}(w_k)=\ell_i^{(k)}(w_k),\quad k=1,\ldots,K.
\end{equation}
Then the empirical product loss can be written as
\begin{equation}
\prod_{k=1}^{K}
\ell_i^{(k)}(w_k)
=
\prod_{k=1}^{K}
f_i^{(k)}(w_k).
\end{equation}
Since Problem~\eqref{prob_ml_product} is a minimization problem, we apply the proposed AM-GM transform to construct an upper-bound surrogate for each product term. At iteration \(t\), define
\begin{equation}
\tilde f_i^{(1)}
=
\left(f_i^{(1)}(w_1)\right)^K
\prod_{k=1}^{K-1}y_i^{(k)},
\end{equation}

\begin{equation}
\tilde f_i^{(k)}
=
\left(f_i^{(k)}(w_k)\right)^K
\frac{\prod_{j=k}^{K-1}y_i^{(j)}}
{\left(y_i^{(k-1)}\right)^{k-1}},
\quad k=2,\ldots,K-1,
\end{equation}
and
\begin{equation}
\tilde f_i^{(K)}
=
\frac{
\left(f_i^{(K)}(w_K)\right)^K
}
{
\left(y_i^{(K-1)}\right)^{K-1}
}.
\end{equation}
The AM upper-bound surrogate of the product loss is
\begin{equation}
F_i^{\rm AM}(\{w_k\},y_i)
=
\frac{1}{K}
\sum_{k=1}^{K}
\tilde f_i^{(k)}.
\end{equation}
The auxiliary variables are updated at the current iterate \(\{w_k^{t}\}_{k=1}^{K}\) according to Equation (\ref{y_n_closedform}):
\begin{equation}
y_i^{1,t}
=
\left(
\frac{
f_i^{(2)}(w_2^{t})
}{
f_i^{(1)}(w_1^{t})
}
\right)^{\frac{K}{2}},
\end{equation}
and
\begin{equation}
y_i^{k-1,t}
=
\left[
\left(y_i^{k-2,t}\right)^{k-2}
\left(
\frac{
f_i^{(k)}(w_k^{t})
}{
f_i^{(k-1)}(w_{k-1}^{t})
}
\right)^K
\right]^{\frac{1}{k}},
\quad k=3,\ldots,K.
\end{equation}
With fixed auxiliary variables, Problem~\eqref{prob_ml_product} is approximated by
\begin{subequations}\label{prob_ml_product_trans}
\begin{align}
\min_{\bm{w}}\quad
&
\frac{1}{N}
\sum_{i=1}^{N}
F_i^{\rm AM}(\{w_k\},y_i^{t})
+
\frac{\lambda}{2}
\sum_{k=1}^{K}
\|w_k\|_2^2 \tag{\ref{prob_ml_product_trans}}\\
\text{s.t.}\quad
&
w_k\in\mathcal{W}_k,\quad k=1,\ldots,K.
\end{align}
\end{subequations}
For fixed \(y_i^{t}\), each surrogate term contains \(\left(\ell_i^{(k)}(w_k)\right)^K\). Since the logistic loss is convex and non-negative, and \(x^K\) is convex and non-decreasing over \(x\ge0\), \(\left(\ell_i^{(k)}(w_k)\right)^K\) is convex in \(w_k\). Therefore, Problem~\eqref{prob_ml_product_trans} is a convex surrogate problem when \(\mathcal{W}_k\) is convex.

\subsubsection{SCA Method Based on the Proposed AM-GM Transform}

The AM-SCA method for Problem~\eqref{prob_ml_product} proceeds as follows. First, initialize a feasible set of classifiers \(\{w_k^{0}\}_{k=1}^{K}\). At iteration \(t\), update the auxiliary variables \(y_i^{t}\) for all training samples according to Equation (\ref{y_n_closedform}). Then solve the convex surrogate problem~\eqref{prob_ml_product_trans} to obtain \(\{w_k^{t+1}\}_{k=1}^{K}\). This procedure is repeated until convergence.

Since the surrogate is an upper bound of the original product loss and is tight at the current iterate, the resulting SCA procedure follows the proposed AM-GM-based framework. Moreover, if solving Problem~\eqref{prob_ml_product_trans} exactly is expensive for large-scale datasets, the gradient-based AM-SCA method can be used by performing several projected gradient steps on the surrogate objective. This gives an efficient implementation for large-scale multi-modal learning problems.

\subsubsection{Parameter Setting}

For numerical evaluation, one may consider a multi-view classification dataset or construct a two-view dataset by splitting the input features into two non-overlapping groups. Unless otherwise specified, the number of modalities is set as \(K=2\) or \(K=3\). The regularization parameter is chosen from a small validation set, and the positivity constant is set as \(\epsilon_\ell=10^{-6}\). The feasible set \(\mathcal{W}_k\) can be chosen as an \(\ell_2\)-norm ball, i.e., \(\mathcal{W}_k=\{w_k:\|w_k\|_2\le R_w\}\). The classifiers are initialized by training each modality independently with the standard logistic loss.

We compare the proposed AM-SCA method with the following baselines:
\begin{itemize}
\item \textbf{Independent training}: each modality-specific classifier is trained separately using the standard logistic loss.
\item \textbf{Average-loss training}: all modality losses are averaged and optimized jointly.
\item \textbf{Direct projected gradient}: projected gradient descent is directly applied to the original product-loss objective without using the proposed AM-GM surrogate.
\end{itemize}

\subsubsection{Numerical Results}
\begin{figure}[tbp]
\centering
\includegraphics[width=0.6\textwidth]{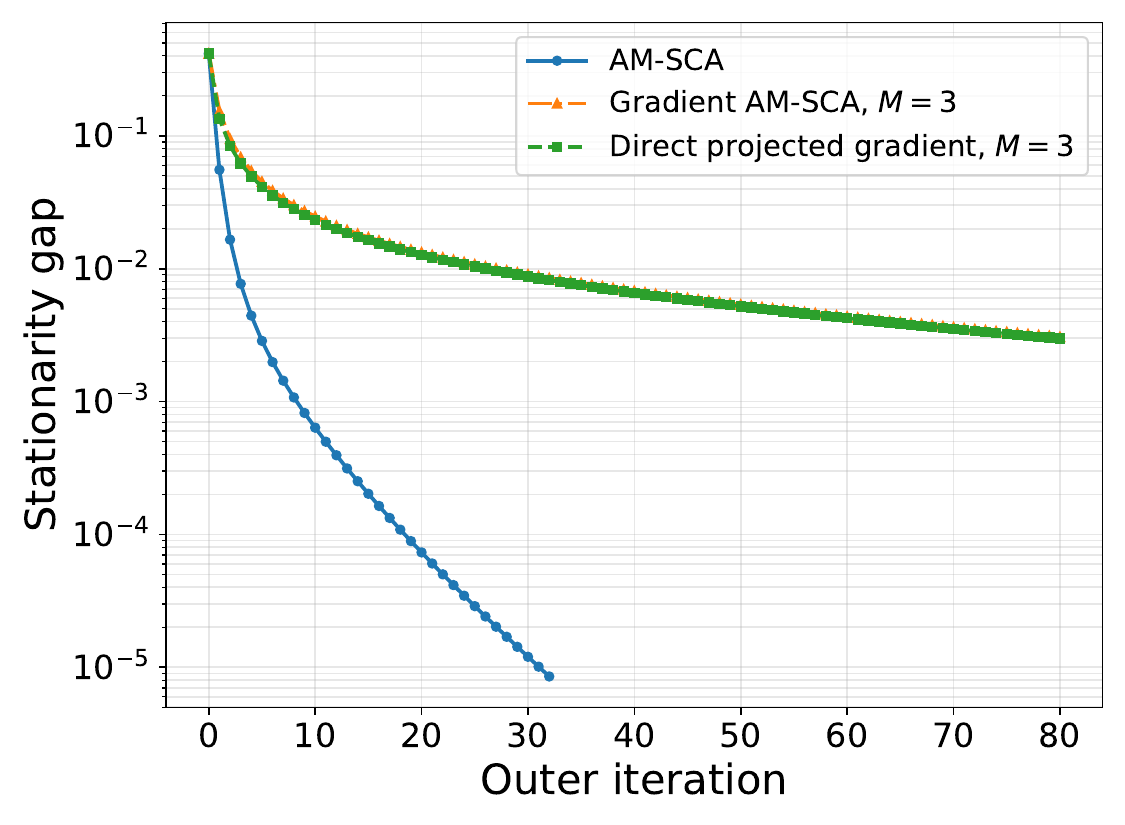}
\vspace{-6pt}
\caption{Convergence of AM-SCA, gradient AM-SCA, and direct projected gradient for multi-modal product-loss learning.}
\label{fig_ml_stationarity}
\end{figure}

\begin{figure}[tbp]
\centering
\includegraphics[width=0.6\textwidth]{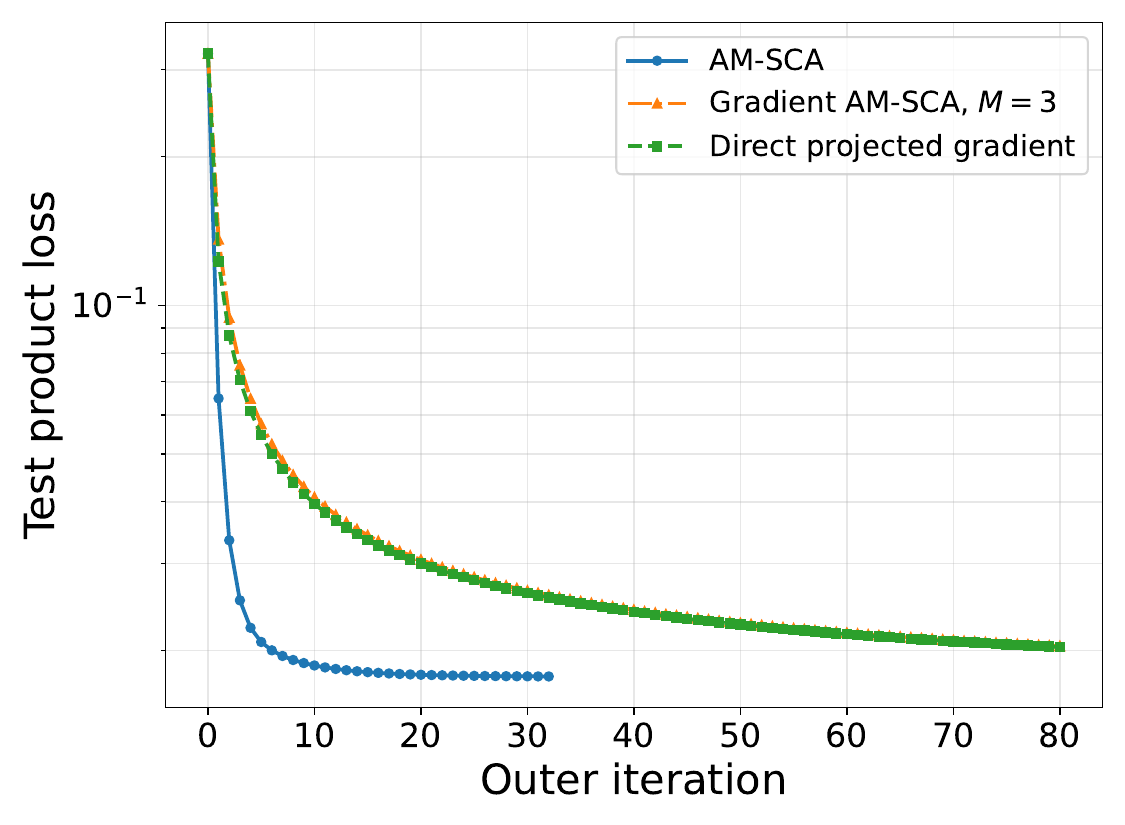}
\vspace{-6pt}
\caption{Test product loss under AM-SCA, gradient AM-SCA, and direct projected gradient.}
\label{fig_ml_product_loss}
\end{figure}

\begin{figure}[tbp]
\centering
\includegraphics[width=0.6\textwidth]{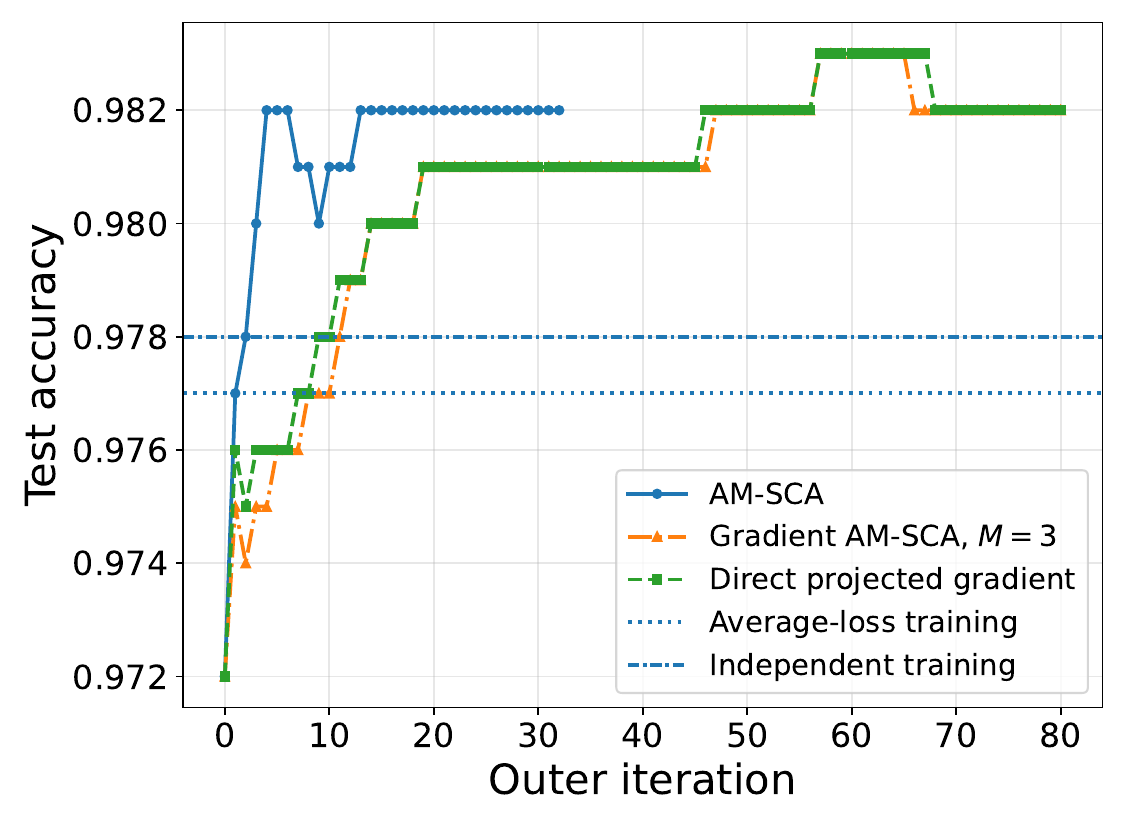}
\vspace{-6pt}
\caption{Test accuracy under different learning methods.}
\label{fig_ml_accuracy}
\end{figure}

\begin{figure}[tbp]
\centering
\includegraphics[width=0.6\textwidth]{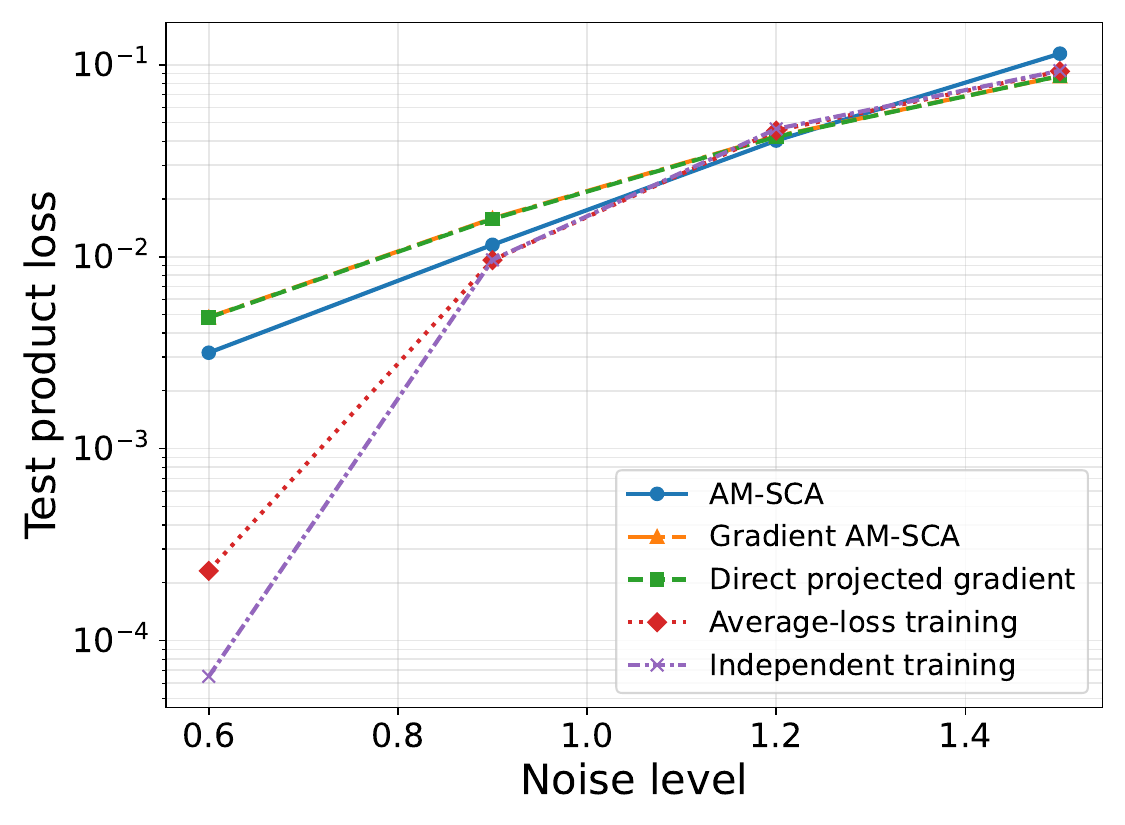}
\vspace{-6pt}
\caption{Test product loss under different noise levels.}
\label{fig_ml_noise}
\end{figure}

Fig.~\ref{fig_ml_stationarity} shows the stationarity-gap convergence of AM-SCA, gradient AM-SCA, and direct projected gradient. AM-SCA decreases the stationarity gap much faster than the two gradient-based methods and reaches a small stationarity level within fewer outer iterations. This is because AM-SCA solves the transformed convex surrogate subproblem more accurately at each outer iteration. In contrast, gradient AM-SCA uses only \(M=3\) projected gradient steps for each surrogate, and direct projected gradient applies the same number of steps directly to the original non-convex product-loss objective. Therefore, their stationarity gaps decrease more slowly. This result confirms that the proposed AM-GM surrogate provides an effective optimization procedure for the considered product-loss learning problem.

In Fig.~\ref{fig_ml_product_loss}, the test product loss during training is presented. AM-SCA reduces the test product loss rapidly and reaches a lower value than gradient AM-SCA and direct projected gradient. This is consistent with the convergence result in Fig.~\ref{fig_ml_stationarity}, since the proposed AM-SCA method directly minimizes a tight upper-bound surrogate of the product-loss objective. Gradient AM-SCA and direct projected gradient also reduce the product loss, but their convergence is slower because they only perform a few local gradient steps at each outer iteration.

In Fig.~\ref{fig_ml_accuracy}, we give the test accuracy of each method. All methods achieve high and similar classification accuracy. AM-SCA improves the accuracy rapidly at the early iterations, while direct projected gradient and gradient AM-SCA can reach slightly higher final accuracy in this particular setting. This is reasonable because the proposed method is designed to optimize the product-loss objective rather than classification accuracy directly. Therefore, stationarity gap and test product loss are the primary metrics for this application, whereas test accuracy is reported only as a reference metric.

Fig.~\ref{fig_ml_noise} compares the test product loss under different noise levels. As the noise level increases, the test product loss increases for all methods, reflecting the higher difficulty of the multi-modal learning task. The proposed AM-SCA and gradient AM-SCA methods remain competitive with the direct projected gradient across different noise levels. These results show that the proposed AM-GM transform can be used to handle machine-learning objectives with sum-of-products structures, and that AM-SCA provides a particularly effective optimization procedure when the transformed surrogate can be solved accurately.

\subsection{Clean-Adversarial Product-Loss Learning}\label{sec_ml_clean}

We next consider a robust learning problem in which the model is trained using both clean samples and their perturbed counterparts. In standard adversarial or augmentation-based training, the clean loss and the perturbed loss are often combined additively. Here, we consider a product-loss formulation, where a sample receives a larger penalty only when it is difficult under both the clean and perturbed inputs. This leads to a sum-of-products objective and can be handled by the proposed AM-GM-based SCA framework.

\subsubsection{Problem Statement}

Consider a binary classification dataset
\(\{(a_i,z_i)\}_{i=1}^{N}\), where \(a_i\in\mathbb{R}^{d}\) is the feature vector and \(z_i\in\{-1,+1\}\) is the label. Let \(\hat a_i\) denote a perturbed version of \(a_i\), which can be generated by data augmentation or by an adversarial perturbation method. For a linear classifier \(w\), the clean and perturbed logistic losses are defined as
\begin{equation}
\ell_i^{\rm c}(w)
=
\log\left(1+\exp(-z_i a_i^\top w)\right)+\epsilon_\ell,
\end{equation}
and
\begin{equation}
\ell_i^{\rm p}(w)
=
\log\left(1+\exp(-z_i \hat a_i^\top w)\right)+\epsilon_\ell,
\end{equation}
where \(\epsilon_\ell>0\) ensures strict positivity. We formulate the clean-adversarial product-loss learning problem as
\begin{subequations}\label{prob_clean_adv_product}
\begin{align}
\min_{w}\quad
&
\frac{1}{N}
\sum_{i=1}^{N}
\ell_i^{\rm c}(w)\ell_i^{\rm p}(w)
+
\frac{\lambda}{2}\|w\|_2^2 \tag{\ref{prob_clean_adv_product}}\\
\text{s.t.}\quad
&
w\in\mathcal{W},
\end{align}
\end{subequations}
where \(\lambda>0\) is the regularization parameter and \(\mathcal{W}\) is a compact convex set. Problem~\eqref{prob_clean_adv_product} is non-convex because the objective contains products of two convex loss functions.

This product-loss formulation emphasizes samples that remain hard under both clean and perturbed inputs. If either the clean loss or the perturbed loss is small, the product loss is reduced. In contrast, samples that are difficult in both cases receive a larger penalty.

\subsubsection{Problem Transformation}

For each sample \(i\), define
\begin{equation}
f_i^{(1)}(w)=\ell_i^{\rm c}(w),\qquad
f_i^{(2)}(w)=\ell_i^{\rm p}(w).
\end{equation}
Then the empirical loss term is
\begin{equation}
f_i^{(1)}(w)f_i^{(2)}(w).
\end{equation}
Since Problem~\eqref{prob_clean_adv_product} is a minimization problem, we apply the AM-GM transform to construct an upper-bound surrogate. For \(K=2\), the AM-GM bound gives
\begin{equation}
f_i^{(1)}(w)f_i^{(2)}(w)
\le
\frac{1}{2}
\left[
y_i\left(f_i^{(1)}(w)\right)^2
+
\frac{\left(f_i^{(2)}(w)\right)^2}{y_i}
\right],
\end{equation}
where \(y_i>0\) is an auxiliary variable. At iteration \(t\), the equality condition (\ref{y_n_closedform}) gives
\begin{equation}
y_i^{t}
=
\frac{f_i^{(2)}(w^{t})}{f_i^{(1)}(w^{t})}.
\end{equation}
Therefore, with fixed \(y_i^{t}\), Problem~\eqref{prob_clean_adv_product} is approximated by
\begin{subequations}\label{prob_clean_adv_product_trans}
\begin{align}
\min_{w}\quad
&
\frac{1}{2N}
\sum_{i=1}^{N}
\left[
y_i^{t}\left(\ell_i^{\rm c}(w)\right)^2
+
\frac{\left(\ell_i^{\rm p}(w)\right)^2}{y_i^{t}}
\right]
+
\frac{\lambda}{2}\|w\|_2^2 \tag{\ref{prob_clean_adv_product_trans}}\\
\text{s.t.}\quad
&
w\in\mathcal{W}.
\end{align}
\end{subequations}
For fixed \(y_i^{t}\), the transformed objective is convex in the linear logistic-regression setting. This is because the logistic loss is convex and non-negative, and the square function is convex and non-decreasing over the non-negative domain. Hence, Problem~\eqref{prob_clean_adv_product_trans} is a convex surrogate problem when \(\mathcal{W}\) is convex.

\subsubsection{SCA Method Based on the Proposed AM-GM Transform}

The AM-SCA method proceeds as follows. First, initialize a feasible classifier \(w^{0}\). At iteration \(t\), compute the clean and perturbed losses for all samples and update the auxiliary variables \(y_i^{t}\) according to the equality condition (\ref{y_n_closedform}). Then solve the convex surrogate problem~\eqref{prob_clean_adv_product_trans} to obtain \(w^{t+1}\). The process is repeated until the projected stationarity gap is below a prescribed tolerance.

When the dataset is large, the convex surrogate subproblem can be solved inexactly by projected gradient descent. This gives a gradient AM-SCA variant, where only a small number of local gradient steps are performed at each outer iteration. In nonlinear models such as neural networks, the transformed subproblem may no longer be convex, but the same gradient AM-SCA procedure can still be used to search for a stationary point.

\subsubsection{Parameter Setting}

In the simulation, we consider a synthetic binary classification dataset. The perturbed sample \(\hat a_i\) is generated by adding bounded random noise or by applying a single-step adversarial perturbation to \(a_i\). The regularization parameter is set as \(\lambda=10^{-3}\), and the positivity constant is set as \(\epsilon_\ell=10^{-4}\). The feasible set is chosen as an \(\ell_2\)-norm ball, i.e.,
\begin{equation}
\mathcal{W}=\{w:\|w\|_2\le R_w\}.
\end{equation}
We compare the proposed AM-SCA method with gradient AM-SCA, direct projected gradient, clean-only training, and additive clean-adversarial training. Direct projected gradient directly minimizes the original product-loss objective without using the AM-GM surrogate, while additive clean-adversarial training minimizes the sum of the clean and perturbed losses.

\subsubsection{Numerical Results}
\begin{figure}[tbp]
\centering
\includegraphics[width=0.6\textwidth]{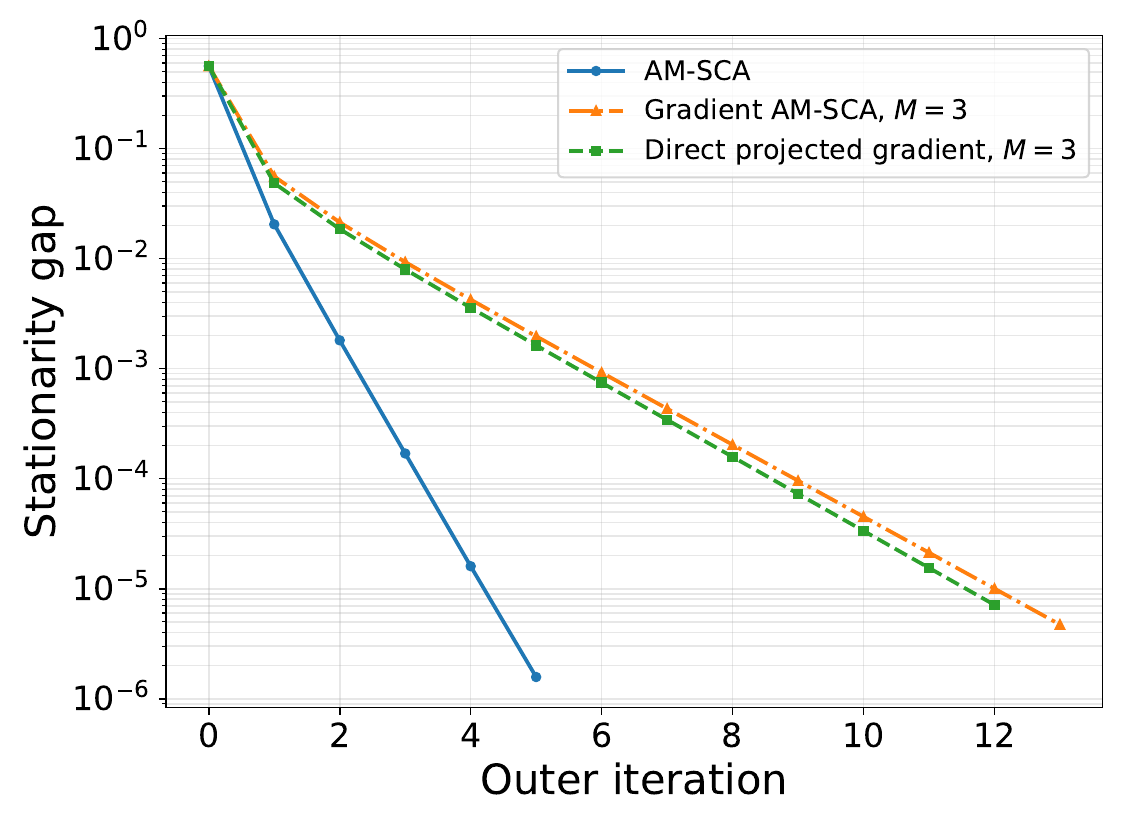}
\vspace{-6pt}
\caption{Convergence of AM-SCA, gradient AM-SCA, and direct projected gradient for clean-adversarial product-loss learning.}
\label{fig_clean_adv_stationarity}
\end{figure}

\begin{figure}[tbp]
\centering
\includegraphics[width=0.6\textwidth]{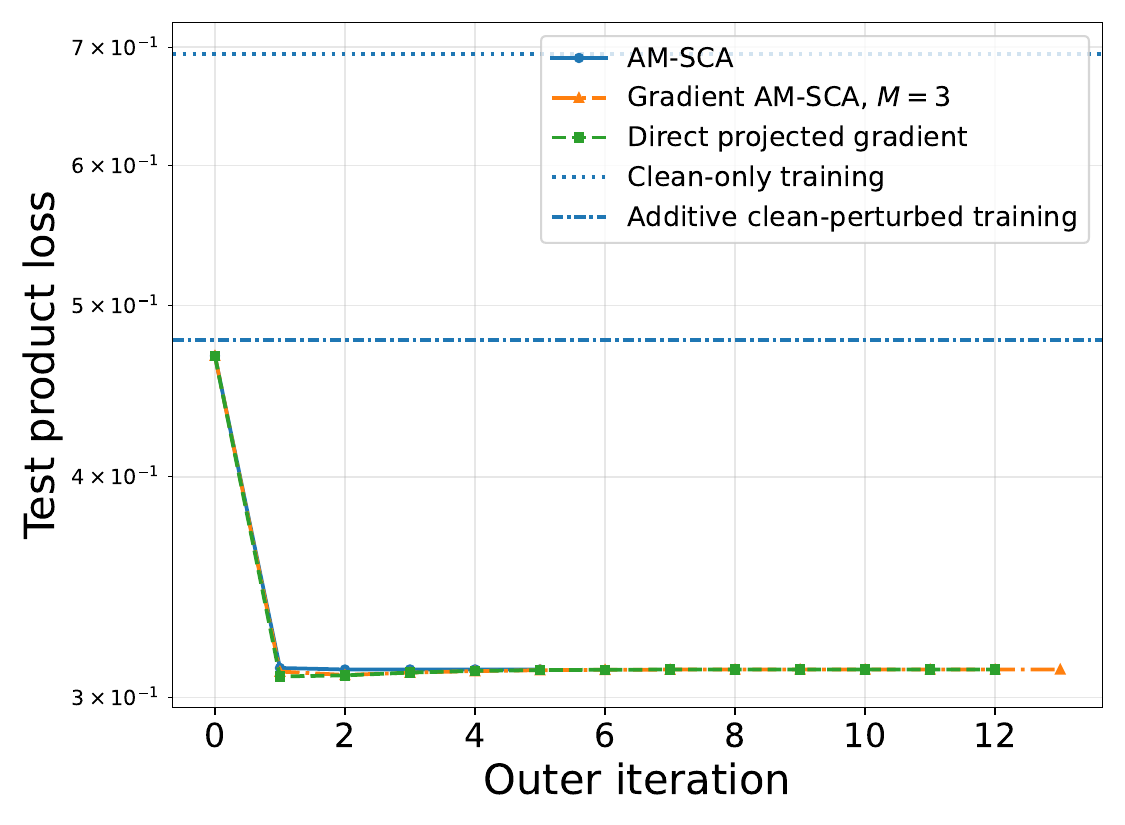}
\vspace{-6pt}
\caption{Test product loss under different optimization methods.}
\label{fig_clean_adv_product_loss}
\end{figure}

\begin{figure}[tbp]
\centering
\includegraphics[width=0.6\textwidth]{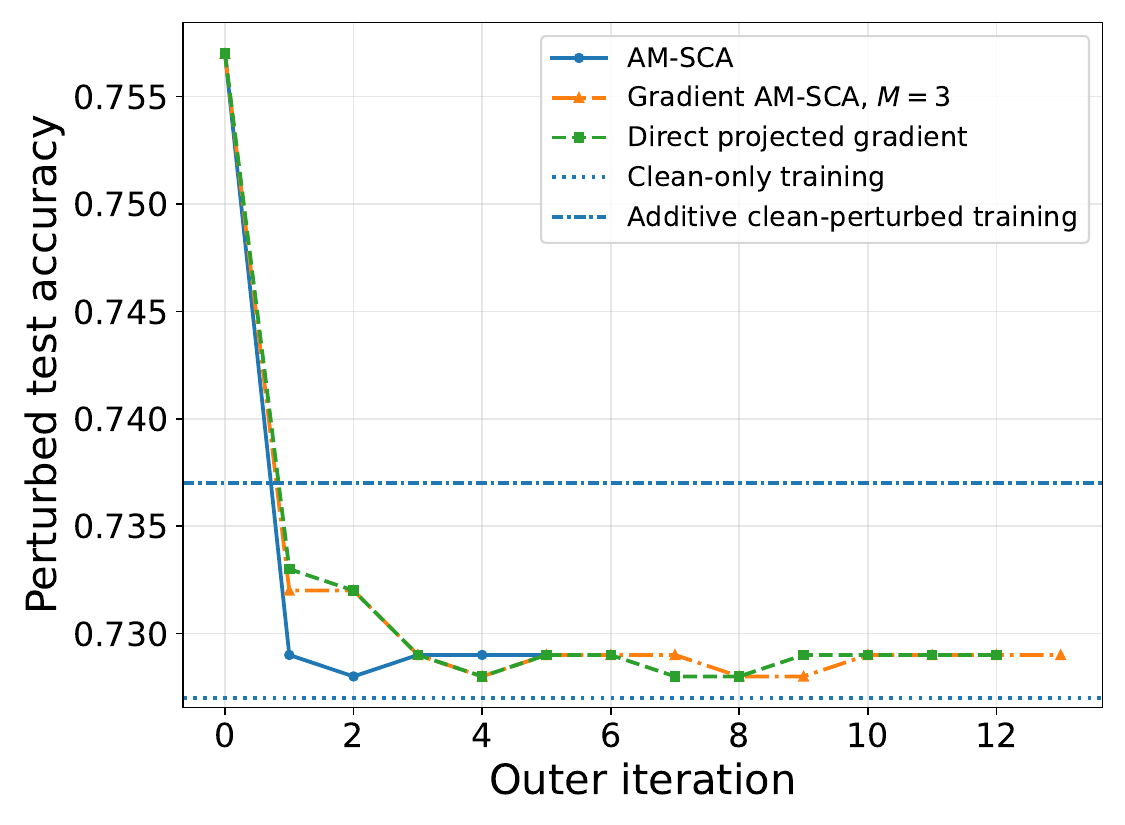}
\vspace{-6pt}
\caption{Perturbed test accuracy under different learning methods.}
\label{fig_clean_adv_accuracy}
\end{figure}

\begin{figure}[tbp]
\centering
\includegraphics[width=0.6\textwidth]{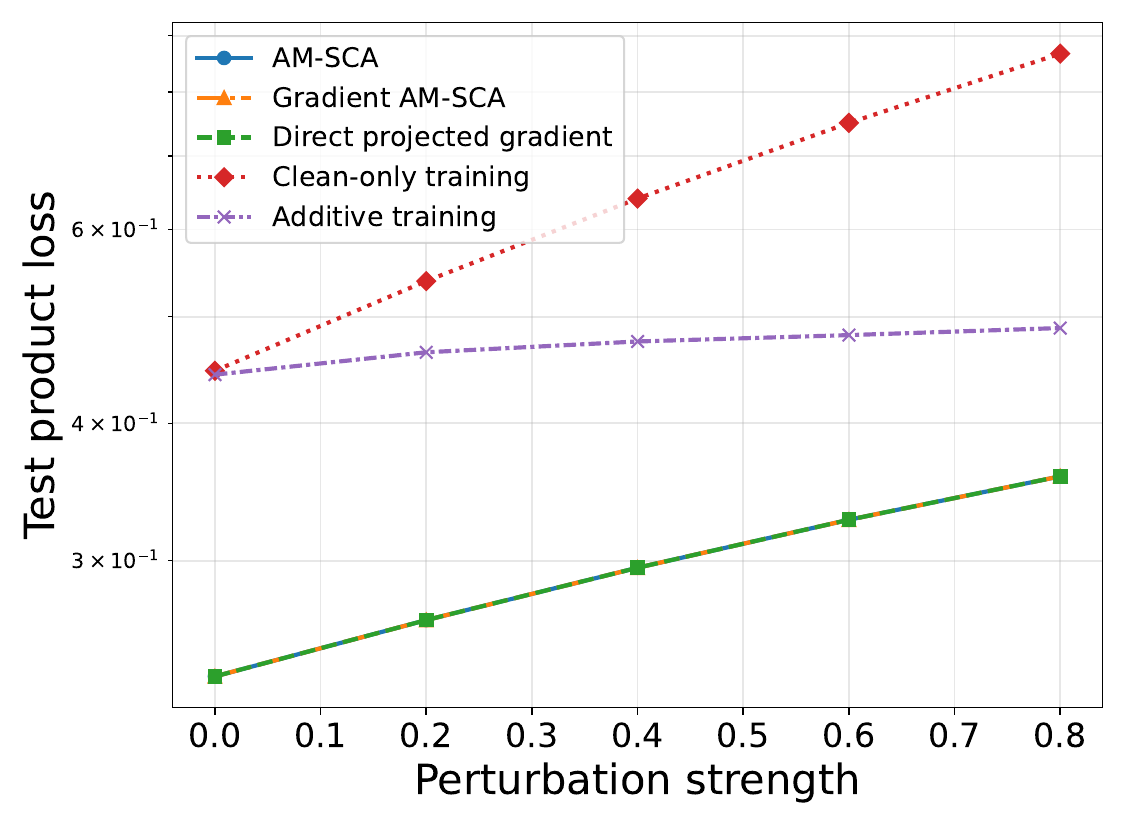}
\vspace{-6pt}
\caption{Test product loss under different perturbation strengths.}
\label{fig_clean_adv_eps}
\end{figure}

Fig.~\ref{fig_clean_adv_stationarity} shows the stationarity-gap convergence of AM-SCA, gradient AM-SCA, and direct projected gradient. AM-SCA reduces the stationarity gap much faster than the two gradient-based methods and reaches the prescribed stationarity level within only a few outer iterations. This is because AM-SCA solves the transformed convex AM surrogate more accurately at each iteration. In contrast, gradient AM-SCA uses only \(M=3\) projected gradient steps for each surrogate subproblem, while direct projected gradient applies the same number of steps directly to the original non-convex product-loss objective. Therefore, their stationarity gaps decrease more slowly. This result verifies the effectiveness of the proposed AM-GM surrogate for solving the clean-adversarial product-loss problem.

Fig.~\ref{fig_clean_adv_product_loss} reports the test product loss during training. AM-SCA decreases the test product loss rapidly and reaches a stable value within fewer iterations. Gradient AM-SCA and direct projected gradient also reduce the product loss, but require more outer iterations. The clean-only and additive clean-perturbed baselines are included as reference methods. Since these baselines do not directly optimize the clean-adversarial product loss, they are less aligned with the proposed product-loss objective.

Fig.~\ref{fig_clean_adv_accuracy} shows the perturbed test accuracy. All methods achieve similar perturbed accuracy, and the differences are relatively small. This is reasonable because the proposed objective is designed to minimize the product of clean and perturbed losses rather than classification error directly. Therefore, the stationarity gap and the test product loss are the primary metrics for this application, while perturbed accuracy is reported as an auxiliary performance indicator.

In Fig.~\ref{fig_clean_adv_eps}, we compare the test product loss under different perturbation strengths. As the perturbation strength increases, the product loss generally increases, which indicates that the learning task becomes more difficult when the perturbed samples deviate further from the clean samples. The proposed AM-SCA and gradient AM-SCA methods remain competitive with the direct projected gradient across different perturbation levels. These results show that the proposed AM-GM transform can be applied to robust learning problems with clean-adversarial product structures.

\subsection{Summary}\label{sec_summary_application}
Table~\ref{tab:application_summary} summarizes representative application problems that can be handled by the proposed HM-GM-AM-QM based SA framework. The proposed AM/QM-SA methods are suitable for minimization problems with sum-of-products or product-type constraints, where upper-bound surrogates are needed. The proposed HM-SA method is suitable for maximization problems with product-type objectives or reliability-type lower-bound constraints, where lower-bound surrogates are required. When the transformed subproblem is convex or concave, the corresponding SCA subproblem can be solved accurately; otherwise, the gradient-based variants provide an efficient way to approach an \(\epsilon\)-stationary point.

{\scriptsize
\begin{longtable}{|p{0.18\textwidth}|p{0.55\textwidth}|p{0.22\textwidth}|}
\caption{Representative application optimization problems that can be handled by the proposed HM-GM-AM-QM based SA methods.}
\label{tab:application_summary}\\
\hline
\textbf{Application} & \textbf{Representative optimization problem} & \textbf{Applicable proposed method} \\
\hline
\endfirsthead

\hline
\textbf{Application} & \textbf{Representative optimization problem} & \textbf{Applicable proposed method} \\
\hline
\endhead

\hline
\endfoot

\hline
\endlastfoot

\multicolumn{3}{|c|}{\textbf{Wireless Communications and Networking}}\\
\hline

Transmit-energy minimization &
\(\displaystyle
\min_{x\in\mathcal{X}}
J(x)+\sum_{n=1}^{N}\prod_{k=1}^{K} f_n^{(k)}(x)
\),
where product terms model coupled power, bandwidth, rate, or delay factors. &
AM/QM-SA;\newline Gradient AM/QM-SA \\ 
\hline

Age-of-information minimization &
\(\displaystyle
\min_{r\in\mathcal{R}}
\sum_{i=1}^{N}
\prod_{k=1}^{K_i} f_{i}^{(k)}(r)
\),
where \(f_i^{(k)}(r)\) captures sampling, service, and transmission delay components. &
AM-SA; Gradient AM-SA \\ 
\hline

Semantic information utility maximization &
\(\displaystyle
\max_{b,p,s}
\sum_{n=1}^{N}
\log\!\left(1+
Q_n(s_n)\frac{R_n(b_n,p_n)}{D_n(s_n)}
\right)
\). &
HM-SA or Gradient HM-SA after logarithmic/product-ratio reformulation \\ 
\hline

Reliability-aware multi-hop routing &
\(\displaystyle
\max_{p\in\mathcal{P}}
\sum_{r=1}^{R} w_r
\prod_{\ell\in\mathcal{P}_r}
P_{\ell}(p_{\ell})
\). &
HM-SA; Gradient HM-SA \\ 
\hline

Cooperative edge caching &
\(\displaystyle
\min_{q}
\frac{1}{U}\sum_{u=1}^{U}\sum_{f=1}^{F}
\pi_{u,f}
\prod_{m\in\mathcal{A}_u}
(1-q_{m,f})
\). &
AM-SA; Gradient AM-SA \\ 
\hline

Multi-connectivity outage minimization &
\(\displaystyle
\min_{x}
\sum_{u=1}^{U}
\prod_{m\in\mathcal{A}_u}
\left(1-P_{u,m}(x)\right)
\),
where outage occurs when all accessible links fail. &
AM-SA; Gradient AM-SA \\ 
\hline

Multi-link success probability maximization &
\(\displaystyle
\max_{p}
\sum_{u=1}^{U}
w_u
\prod_{\ell\in\mathcal{L}_u}
\left(1-e^{-a_{\ell}p_{\ell}}\right)
\). &
HM-SA; Gradient HM-SA \\ 
\hline

Resource allocation with reliability constraints &
\(\displaystyle
\prod_{\ell\in\mathcal{L}_u}
P_{\ell}(p_{\ell})\ge \eta_u,\quad \forall u
\). &
HM-GM lower-bound SA for conservative feasible updates \\ 
\hline

Wireless powered communication &
\(\displaystyle
\max_{\tau,p}
\sum_{u=1}^{U}
w_u
\prod_{k=1}^{K}
f_{u}^{(k)}(\tau,p)
\),
where harvesting time, transmit power, and rate are coupled. &
HM-SA; Gradient HM-SA \\ 
\hline

UAV-assisted relay optimization &
\(\displaystyle
\max_{q,p}
\sum_{u=1}^{U}
w_u
P^{\rm air}_{u}(q,p)
P^{\rm ground}_{u}(q,p)
\). &
HM-SA; Gradient HM-SA \\ 
\hline

Space-air-ground integrated networks &
\(\displaystyle
\max_{x}
\sum_{u=1}^{U}
\prod_{k=1}^{K}
f_{u}^{(k)}(x)
\),
where terms model satellite, aerial, and terrestrial service factors. &
HM-SA; Gradient HM-SA \\ 
\hline

Task offloading success maximization &
\(\displaystyle
\max_{x}
\sum_{u=1}^{U}
w_u
P^{\rm assoc}_u(x)
P^{\rm offload}_u(x)
P^{\rm compute}_u(x)
\). &
HM-SA; Gradient HM-SA \\ 
\hline

Edge computing delay-risk minimization &
\(\displaystyle
\min_{x}
\sum_{u=1}^{U}
D_u(x)E_u(x)C_u(x)
\),
where \(D_u,E_u,C_u\) are delay, energy, and computation-cost factors. &
AM/QM-SA;\newline Gradient AM/QM-SA \\ 
\hline

URLLC reliability-latency optimization &
\(\displaystyle
\min_{x}
\sum_{u=1}^{U}
\left(1-P_u^{\rm rel}(x)\right)L_u(x)
\). &
AM/QM-SA;\newline Gradient AM/QM-SA \\ 
\hline

Secure communication utility maximization &
\(\displaystyle
\max_{p}
\sum_{u=1}^{U}
w_u
R_u^{\rm sec}(p)P_u^{\rm rel}(p)
\). &
HM-SA; Gradient HM-SA \\ 
\hline

Physical-layer secrecy outage minimization &
\(\displaystyle
\min_{p}
\sum_{u=1}^{U}
\prod_{e\in\mathcal{E}_u}
P_{u,e}^{\rm leak}(p)
\). &
AM-SA; Gradient AM-SA \\ 
\hline

\multicolumn{3}{|c|}{\textbf{Information Theory and Information-Centric Systems}}\\
\hline

Sum-of-log-ratios maximization &
\(\displaystyle
\max_{x\in\mathcal{X}}
\sum_{i=1}^{N}
\log\!\left(1+\frac{A_i(x)}{B_i(x)}\right)
\). &
HM-SA after logarithmic ratio reformulation \\ 
\hline

Sum-of-log-products maximization &
\(\displaystyle
\max_{x\in\mathcal{X}}
\sum_{i=1}^{N}
\log\!\left(1+
\prod_{k=1}^{K}f_i^{(k)}(x)
\right)
\). &
HM-SA; Gradient HM-SA \\ 
\hline

Information freshness-reliability maximization &
\(\displaystyle
\max_{x}
\sum_{u=1}^{U}
w_u
F_u(x)R_u(x)
\),
where \(F_u(x)\) is freshness utility and \(R_u(x)\) is delivery reliability. &
HM-SA; Gradient HM-SA \\ 
\hline

Distributed sensing information availability &
\(\displaystyle
\max_{x}
\sum_{s=1}^{S}
w_s
\prod_{m\in\mathcal{M}_s}
P_{m,s}(x)
\). &
HM-SA; Gradient HM-SA \\ 
\hline

Source coding with multiplicative distortion risk &
\(\displaystyle
\min_{x}
\sum_{i=1}^{N}
D_i(x)R_i(x)
\),
where distortion and rate-related risk are coupled multiplicatively. &
AM/QM-SA; \newline Gradient AM/QM-SA \\ 
\hline

Distributed estimation error minimization &
\(\displaystyle
\min_{x}
\sum_{i=1}^{N}
\prod_{k=1}^{K}
{\rm MSE}_{i,k}(x)
\). &
AM/QM-SA;\newline Gradient AM/QM-SA \\ 
\hline

Sensor fusion miss-detection minimization &
\(\displaystyle
\min_{x}
\sum_{i=1}^{N}
\prod_{k=1}^{K}
P_{i,k}^{\rm miss}(x)
\). &
AM-SA; Gradient AM-SA \\ 
\hline

Information bottleneck product regularization &
\(\displaystyle
\min_{\theta}
\mathcal{L}(\theta)
+
\lambda
I(X;Z_{\theta})I(Z_{\theta};Y)^{-1}
\). &
AM/QM-SA or Gradient AM/QM-SA after product-ratio reformulation \\ 
\hline

\multicolumn{3}{|c|}{\textbf{Machine Learning and Artificial Intelligence}}\\
\hline

Multi-modal product-loss learning &
\(\displaystyle
\min_{\{w_k\}}
\frac{1}{N}
\sum_{i=1}^{N}
\prod_{k=1}^{K}
\ell_i^{(k)}(w_k)
+
\frac{\lambda}{2}\sum_{k=1}^{K}\|w_k\|^2
\). &
AM-SA; Gradient AM-SA \\ 
\hline

Clean-adversarial product-loss learning &
\(\displaystyle
\min_{w}
\frac{1}{N}
\sum_{i=1}^{N}
\ell_i^{\rm c}(w)\ell_i^{\rm p}(w)
+
\frac{\lambda}{2}\|w\|^2
\). &
AM-SA; Gradient AM-SA \\ 
\hline

Multi-task product-loss learning &
\(\displaystyle
\min_{\theta}
\sum_{i=1}^{N}
\prod_{t=1}^{T}
\ell_{i,t}(\theta)
+
\lambda R(\theta)
\). &
AM/QM-SA;\newline Gradient AM/QM-SA \\ 
\hline

Domain generalization with product domain risk &
\(\displaystyle
\min_{\theta}
\sum_{i=1}^{N}
\prod_{d=1}^{D}
\ell_{i,d}(\theta)
+
\lambda R(\theta)
\). &
AM-SA for convex models; Gradient AM-SA for neural models \\ 
\hline

Federated learning with multiplicative client risk &
\(\displaystyle
\min_{\theta}
\sum_{m=1}^{M}
\prod_{k=1}^{K_m}
\ell_{m,k}(\theta)
+
\lambda R(\theta)
\). &
AM/QM-SA;\newline Gradient AM/QM-SA \\ 
\hline

Robust learning with group-risk product &
\(\displaystyle
\min_{\theta}
\prod_{g=1}^{G}
\left(
\mathcal{L}_g(\theta)+\epsilon
\right)
\). &
AM/QM-SA;\newline Gradient AM/QM-SA \\ 
\hline

Fairness-aware learning &
\(\displaystyle
\min_{\theta}
\mathcal{L}(\theta)
+
\lambda
\prod_{g=1}^{G}
\left(
{\rm Gap}_g(\theta)+\epsilon
\right)
\). &
AM/QM-SA;\newline Gradient AM/QM-SA \\ 
\hline

Knowledge distillation with coupled losses &
\(\displaystyle
\min_{\theta}
\sum_{i=1}^{N}
\ell_i^{\rm CE}(\theta)
\ell_i^{\rm KD}(\theta)
+
\lambda R(\theta)
\). &
AM-SA; Gradient AM-SA \\ 
\hline

Contrastive learning with product penalties &
\(\displaystyle
\min_{\theta}
\sum_{i=1}^{N}
\ell_i^{\rm align}(\theta)
\ell_i^{\rm uniform}(\theta)
+
\lambda R(\theta)
\). &
AM/QM-SA;\newline Gradient AM/QM-SA \\ 
\hline

Multi-objective neural training &
\(\displaystyle
\min_{\theta}
\prod_{j=1}^{J}
\left(
\mathcal{L}_j(\theta)+\epsilon
\right)
\). &
Gradient AM/QM-SA \\ 
\hline

Adversarial training with clean-robust coupling &
\(\displaystyle
\min_{\theta}
\sum_{i=1}^{N}
\ell_i^{\rm clean}(\theta)
\ell_i^{\rm adv}(\theta)
\). &
AM-SA for convex models; Gradient AM-SA for deep models \\ 
\hline

Ensemble learning with joint error minimization &
\(\displaystyle
\min_{\theta_1,\ldots,\theta_K}
\sum_{i=1}^{N}
\prod_{k=1}^{K}
\ell_i(\theta_k)
\). &
AM-SA; Gradient AM-SA \\ 
\hline

Co-training with view-consistency product loss &
\(\displaystyle
\min_{\theta}
\sum_{i=1}^{N}
\ell_i^{\rm sup}(\theta)
\ell_i^{\rm con}(\theta)
+
\lambda R(\theta)
\). &
AM-SA; Gradient AM-SA \\ 
\hline

Uncertainty-aware learning &
\(\displaystyle
\min_{\theta}
\sum_{i=1}^{N}
\ell_i(\theta)
U_i(\theta)
+
\lambda R(\theta)
\). &
AM/QM-SA; \newline Gradient AM/QM-SA \\ 
\hline

\multicolumn{3}{|c|}{\textbf{AI Security, Privacy, and Robustness}}\\
\hline

Adversarial example generation &
\(\displaystyle
\max_{\delta\in\Delta}
\ell_{\rm cls}(x+\delta)
S_{\rm imper}(x+\delta)
\). &
HM-SA; Gradient HM-SA \\ 
\hline

Universal adversarial perturbation &
\(\displaystyle
\max_{\delta\in\Delta}
\sum_{i=1}^{N}
\prod_{k=1}^{K}
\ell_{i,k}(x_i+\delta)
\). &
HM-SA; Gradient HM-SA \\ 
\hline

Adversarial patch optimization &
\(\displaystyle
\max_{p\in\mathcal{P}}
\sum_{i=1}^{N}
\ell_i^{\rm attack}(p)
T_i^{\rm physical}(p)
\). &
HM-SA; Gradient HM-SA \\ 
\hline

Unlearnable example generation &
\(\displaystyle
\max_{\delta\in\Delta}
\sum_{i=1}^{N}
\ell_i^{\rm train}(\delta)
S_i^{\rm fidelity}(\delta)
\). &
HM-SA; Gradient HM-SA \\ 
\hline

Privacy-utility tradeoff optimization &
\(\displaystyle
\min_{\theta}
\mathcal{L}_{\rm task}(\theta)
\mathcal{R}_{\rm privacy}(\theta)
+
\lambda R(\theta)
\). &
AM/QM-SA;\newline Gradient AM/QM-SA \\ 
\hline

Model inversion risk minimization &
\(\displaystyle
\min_{\theta}
\mathcal{L}_{\rm task}(\theta)
\prod_{a=1}^{A}
\left(
\mathcal{R}_{a}^{\rm inv}(\theta)+\epsilon
\right)
\). &
AM/QM-SA;\newline Gradient AM/QM-SA \\ 
\hline

Robust unlearning evaluation &
\(\displaystyle
\max_{\delta\in\Delta}
\ell_{\rm recover}(\delta)
\ell_{\rm bypass}(\delta)
\). &
HM-SA; Gradient HM-SA \\ 
\hline

Prompt-attack optimization &
\(\displaystyle
\max_{x\in\mathcal{X}}
S_{\rm unsafe}(x)
S_{\rm semantic}(x)
S_{\rm bypass}(x)
\). &
HM-SA; Gradient HM-SA \\ 
\hline

Detector-aware attack optimization &
\(\displaystyle
\max_{\delta}
\ell_{\rm target}(\delta)
\ell_{\rm evade}(\delta)
\). &
HM-SA; Gradient HM-SA \\ 
\hline

\multicolumn{3}{|c|}{\textbf{Distributed Learning and Federated Systems}}\\
\hline

Federated resource-aware training &
\(\displaystyle
\min_{x}
\sum_{m=1}^{M}
\mathcal{L}_m(x)
E_m(x)
T_m(x)
\). &
AM/QM-SA;\newline Gradient AM/QM-SA \\ 
\hline

Client selection with reliability utility &
\(\displaystyle
\max_{s}
\sum_{m=1}^{M}
U_m(s_m)
P_m^{\rm avail}(s_m)
\). &
HM-SA; Gradient HM-SA \\ 
\hline

Federated edge learning efficiency maximization &
\(\displaystyle
\max_{x}
\sum_{m=1}^{M}
\frac{A_m(x)}
{E_m(x)T_m(x)}
\). &
HM-SA after product-ratio reformulation \\ 
\hline

Hierarchical offloading for PEFT training &
\(\displaystyle
\max_{x}
\sum_{m=1}^{M}
P_m^{\rm train}(x)
P_m^{\rm offload}(x)
P_m^{\rm resource}(x)
\). &
HM-SA; Gradient HM-SA \\ 
\hline

Communication-efficient FL scheduling &
\(\displaystyle
\min_{x}
\sum_{m=1}^{M}
D_m(x)E_m(x)Q_m(x)
\). &
AM/QM-SA;\newline Gradient AM/QM-SA \\ 
\hline

\multicolumn{3}{|c|}{\textbf{Optimization with Non-Convex Constraints}}\\
\hline

Product upper-bound constraint &
\(\displaystyle
\prod_{k=1}^{K} f_k(x)\le c
\). &
AM/QM-SA constraint approximation \\ 
\hline

Product lower-bound constraint &
\(\displaystyle
\prod_{k=1}^{K} f_k(x)\ge c
\). &
HM-SA conservative lower-bound approximation \\ 
\hline

Ratio-product constraint &
\(\displaystyle
\prod_{k=1}^{K}
\frac{A_k(x)}{B_k(x)}
\le c
\). &
AM/QM-SA after ratio-to-product reformulation \\ 
\hline

Reliability constraint &
\(\displaystyle
\prod_{k=1}^{K}
P_k(x)\ge \eta
\). &
HM-SA \\ 
\hline

Risk constraint &
\(\displaystyle
\prod_{k=1}^{K}
R_k(x)\le \rho
\). &
AM/QM-SA \\ 
\hline

Chance-surrogate constraint &
\(\displaystyle
\prod_{k=1}^{K}
\left(1-P_k^{\rm fail}(x)\right)\ge 1-\epsilon
\). &
HM-SA \\ 
\hline

\multicolumn{3}{|c|}{\textbf{Energy Systems, Control, and Cyber-Physical Systems}}\\
\hline

Energy-delay product minimization &
\(\displaystyle
\min_{x}
\sum_{i=1}^{N}
E_i(x)D_i(x)
\). &
AM/QM-SA;\newline Gradient AM/QM-SA \\ 
\hline

Control reliability maximization &
\(\displaystyle
\max_{u}
\sum_{i=1}^{N}
\prod_{k=1}^{K}
P_{i,k}^{\rm stable}(u)
\). &
HM-SA; Gradient HM-SA \\ 
\hline

Smart-grid risk minimization &
\(\displaystyle
\min_{x}
\sum_{i=1}^{N}
R_i^{\rm load}(x)R_i^{\rm price}(x)R_i^{\rm outage}(x)
\). &
AM/QM-SA;\newline Gradient AM/QM-SA \\ 
\hline

Demand-response utility maximization &
\(\displaystyle
\max_{x}
\sum_{i=1}^{N}
U_i(x_i)P_i^{\rm participation}(x_i)
\). &
HM-SA; Gradient HM-SA \\ 
\hline

\multicolumn{3}{|c|}{\textbf{Finance, Operations, and General Decision Systems}}\\
\hline

Portfolio product-risk minimization &
\(\displaystyle
\min_{x}
\sum_{i=1}^{N}
R_i^{\rm market}(x)R_i^{\rm liquidity}(x)
\). &
AM/QM-SA;\newline Gradient AM/QM-SA \\ 
\hline

Reliability-aware supply chain design &
\(\displaystyle
\max_{x}
\sum_{r=1}^{R}
w_r
\prod_{k\in\mathcal{S}_r}
P_k^{\rm success}(x)
\). &
HM-SA; Gradient HM-SA \\ 
\hline

Multi-stage failure probability minimization &
\(\displaystyle
\min_{x}
\sum_{r=1}^{R}
\prod_{k\in\mathcal{S}_r}
P_k^{\rm fail}(x)
\). &
AM-SA; Gradient AM-SA \\ 
\hline

Operations cost-efficiency maximization &
\(\displaystyle
\max_{x}
\sum_{i=1}^{N}
\frac{U_i(x)}
{C_i(x)D_i(x)}
\). &
HM-SA after product-ratio reformulation \\ 
\hline

\end{longtable}
}

\section{Limitation of Our Proposed Transforms}\label{sec_limitation}
In this section, we present the following limitation of our proposed transforms. In Equation (\ref{y_n_closedform}), if $f^{(k)}_n(\bm{x}) = \alpha_n^{(k)} f^{(1)}_n(\bm{x})$, where $\alpha_n^{(k)}$ is a strictly positive constant scaling parameter, we can compute the $\bm{y}_n$ as
\begin{align}
    & y^{(1)}_n = \sqrt{\alpha_n^{(2)}},\nonumber \\
    & y^{(2)}_n = (y^{(1)}_n)^{\frac{1}{3}}\cdot \left(\frac{\alpha_n^{(3)}}{\alpha_n^{(2)}}\right)^{\frac{1}{3}},\nonumber \\
    & y^{(k-1)}_n=(y_n^{(1)})^{\prod_{i=1}^{k-2} \frac{i}{i+2}} \cdot \left(\prod_{i=2}^{k-2} \left(\frac{\alpha_n^{(i+1)}(\bm{x})}{\alpha_n^{(i)}(\bm{x})}\right)^{\frac{1}{i+1}\cdot \prod_{j=i+2}^k \frac{j-2}{j}}\right) \cdot \left(\frac{\alpha_n^{(k)}(\bm{x})}{\alpha_n^{(k-1)}(\bm{x})}\right)^{\frac{1}{k}}, \forall k\in \{4,\cdots,K\}.
\end{align}
Therefore, it is evident that the introduced variable $\bm{y}_n$ is a constant, which implies that we can't include any current-iterate information into the $\bm{y}_n$ with any given feasible point $\bm{x}^0$. In other words, the proposed bounds can not be used in the SA/SCA technique under this special case.

\section{Conclusion} \label{secConclusion}

This paper developed an inequality-based transform framework for optimization problems involving multiplicative and fractional terms with an arbitrary number of coupled functions. By exploiting the HM-GM-AM-QM inequality chain, we constructed HM lower-bound surrogates and AM/QM upper-bound surrogates for general product/ratio-type terms. The corresponding auxiliary variables were derived in closed form from the equality conditions, which ensures that the proposed surrogates are tight at the current iterate.

Based on these transforms, we further developed successive approximation (SA) methods for sum-of-products/ratios minimization and maximization problems. For minimization, the AM and QM bounds provide upper-bound surrogates, while for maximization, the HM bound provides a lower-bound surrogate. We clarified that the resulting method reduces to a standard SCA method when the transformed surrogate is convex for minimization or concave for maximization. When such convexity or concavity is not guaranteed, the proposed gradient-based SA variants provide a practical way to update the decision variables through gradient descent or ascent. The convergence and complexity analysis show that the proposed methods can converge to stationary points under standard assumptions, and that the gradient-based variants achieve a sublinear convergence rate to an $\epsilon$-stationary point.

We also showed that the proposed transforms are applicable not only to objective functions but also to non-convex constraints and logarithmic product-ratio structures. Through several case studies, including transmit-energy minimization, age-of-information minimization, semantic utility maximization, reliability-aware information delivery, cooperative edge caching, and product-loss learning, we demonstrated the flexibility of the proposed framework in communication networks, information systems, and machine learning applications. The numerical results verify that the proposed methods can effectively handle a broad class of non-convex optimization problems with multiplicative structures.

Finally, we discussed the limitations of the proposed transforms. In particular, the transformed surrogates are not always convex or concave, and the usefulness of a specific bound depends on the structure of the component functions. Future work may investigate sharper product bounds, matrix-valued extensions of the HM-GM-AM-QM transform, and broader applications in large-scale networked and learning systems.
\bibliographystyle{IEEEtran}
\bibliography{ref}

\clearpage
\begin{appendices}

\section{Proof of Proposition \ref{prop_property}}\label{append_proof_prop_property}

\begin{proof}\label{proof_prop_property}
The first relation follows directly from the AM-GM upper-bound construction. Specifically, for any fixed feasible point $\bm{x}^\prime$, the auxiliary variables $\bm{y}_n(\bm{x}^\prime)$ are chosen according to Equation (\ref{y_n_closedform}). Hence, for any feasible $\bm{x}$, we have
\begin{equation}
    \prod_{k=1}^{K} f_n^{(k)}(\bm{x})
    \leq
    F_n^{\rm AM}(\bm{x};\bm{y}_n(\bm{x}^\prime)),
    \quad n=1,\ldots,N.
\end{equation}
By summing these inequalities over $n$ and adding $J(\bm{x})$, we obtain
\begin{equation}
    \Phi(\bm{x})\leq Q(\bm{x};\bm{x}^\prime).
\end{equation}
Moreover, since $\bm{y}_n(\bm{x}^\prime)$ is updated by Equation (\ref{y_n_closedform}) at $\bm{x}^\prime$, the AM bound is tight at $\bm{x}^\prime$, i.e.,
\begin{equation}
    \prod_{k=1}^{K} f_n^{(k)}(\bm{x}^\prime)
    =
    F_n^{\rm AM}(\bm{x}^\prime;\bm{y}_n(\bm{x}^\prime)).
\end{equation}
Therefore,
\begin{equation}
    \Phi(\bm{x}^\prime)=Q(\bm{x}^\prime;\bm{x}^\prime).
\end{equation}

It remains to prove the first-order consistency. Define
\begin{equation}
    D(\bm{x};\bm{x}^\prime)
    =
    Q(\bm{x};\bm{x}^\prime)-\Phi(\bm{x}).
\end{equation}
From the AM upper-bound property, we have
\begin{equation}
    D(\bm{x};\bm{x}^\prime)\geq 0,\quad \forall \bm{x},
\end{equation}
and from the equality condition (\ref{eq:equality_condition_no_phi}),
\begin{equation}
    D(\bm{x}^\prime;\bm{x}^\prime)=0.
\end{equation}
Thus, $\bm{x}^\prime$ is a minimizer of $D(\bm{x};\bm{x}^\prime)$ with respect to $\bm{x}$. Since all involved functions are differentiable, we have
\begin{equation}
    \nabla_{\bm{x}}D(\bm{x}^\prime;\bm{x}^\prime)=\bm{0}.
\end{equation}
This implies
\begin{equation}
    \nabla_{\bm{x}}Q(\bm{x}^\prime;\bm{x}^\prime)
    =
    \nabla_{\bm{x}}\Phi(\bm{x}^\prime).
\end{equation}
The proof is completed.
\end{proof}

\section{Proof of Theorem \ref{theorem_sca_convergence}}\label{appdix_proof_theorem_sca_convergence}
\begin{proof}
At iteration $i$, the auxiliary variables are updated at the current feasible point $\bm{x}^{i}$, i.e.,
    $\bm{y}^{i}=\bm{y}(\bm{x}^{i}).$
Then the AM-based surrogate anchored at $\bm{x}^{i}$ is
\begin{equation}
    Q(\bm{x};\bm{x}^{i})
    =
    J(\bm{x})
    +
    \sum_{n=1}^{N}
    F_n^{\rm AM}(\bm{x};\bm{y}_n(\bm{x}^{i})).
\end{equation}
Algorithm~\ref{algo:SCA} updates $\bm{x}^{i+1}$ by solving the transformed surrogate problem exactly, namely,
\begin{equation}
    \bm{x}^{i+1}
    \in
    \arg\min_{\bm{x}\in\mathcal{X}}
    Q(\bm{x};\bm{x}^{i}).
\end{equation}
Therefore,
\begin{equation}
    Q(\bm{x}^{i+1};\bm{x}^{i})
    \leq
    Q(\bm{x}^{i};\bm{x}^{i}).
\end{equation}
By Proposition~\ref{prop_property}, we have
\begin{equation}
    \Phi(\bm{x}^{i})
    =
    Q(\bm{x}^{i};\bm{x}^{i}),
\end{equation}
and
\begin{equation}
    \Phi(\bm{x}^{i+1})
    \leq
    Q(\bm{x}^{i+1};\bm{x}^{i}).
\end{equation}
Combining the above relations gives
\begin{equation}
    \Phi(\bm{x}^{i+1})
    \leq
    Q(\bm{x}^{i+1};\bm{x}^{i})
    \leq
    Q(\bm{x}^{i};\bm{x}^{i})
    =
    \Phi(\bm{x}^{i}).
\end{equation}
Hence, $\{\Phi(\bm{x}^{i})\}$ is non-increasing. Since $\mathcal{X}$ is compact and $\Phi(\bm{x})$ is continuous, $\Phi(\bm{x})$ is bounded below over $\mathcal{X}$. Therefore, $\{\Phi(\bm{x}^{i})\}$ converges.

It remains to show the stationarity of any limit point. The update of Algorithm~\ref{algo:SCA} can be interpreted as an exact two-block optimization procedure for the transformed problem~\eqref{prob4}, where the two blocks are $\bm{x}$ and $\bm{y}$. For fixed $\bm{x}$, the update $\bm{y}=\bm{y}(\bm{x})$ is obtained from the equality condition of the AM bound. For fixed $\bm{y}^{i}$, the update of $\bm{x}$ solves the transformed convex subproblem exactly. Therefore, by Lemma~\ref{lemma_mbi}, any cluster point of the generated sequence is blockwise stationary for the transformed problem~\eqref{prob4}.

Let $\bar{\bm{x}}$ be a limit point of $\{\bm{x}^{i}\}$. Since the auxiliary variables are updated according to Equation~\eqref{y_n_closedform}, the corresponding limit point of the auxiliary variables satisfies
\begin{equation}
    \bar{\bm{y}}
    =
    \bm{y}(\bar{\bm{x}}).
\end{equation}
The blockwise stationarity with respect to the $\bm{x}$-block implies
\begin{equation}
    \bar{\bm{x}}
    \in
    \arg\min_{\bm{x}\in\mathcal{X}}
    \left(
    J(\bm{x})
    +
    \sum_{n=1}^{N}
    F_n^{\rm AM}(\bm{x};\bar{\bm{y}}_n)
    \right).
\end{equation}
Since $\bar{\bm{y}}=\bm{y}(\bar{\bm{x}})$, this is equivalent to
\begin{equation}
    \bar{\bm{x}}
    \in
    \arg\min_{\bm{x}\in\mathcal{X}}
    Q(\bm{x};\bar{\bm{x}}).
\end{equation}
Because $Q(\bm{x};\bar{\bm{x}})$ is differentiable and $\mathcal{X}$ is convex, the first-order optimality condition gives
\begin{equation}
    \nabla_{\bm{x}}Q(\bar{\bm{x}};\bar{\bm{x}})^{\top}
    \left(
    \bm{x}-\bar{\bm{x}}
    \right)
    \geq 0,
    \quad
    \forall \bm{x}\in\mathcal{X}.
\end{equation}
By the first-order consistency property in Proposition~\ref{prop_property},
\begin{equation}
    \nabla_{\bm{x}}Q(\bar{\bm{x}};\bar{\bm{x}})
    =
    \nabla \Phi(\bar{\bm{x}}).
\end{equation}
Therefore,
\begin{equation}
    \nabla \Phi(\bar{\bm{x}})^{\top}
    \left(
    \bm{x}-\bar{\bm{x}}
    \right)
    \geq 0,
    \quad
    \forall \bm{x}\in\mathcal{X}.
\end{equation}
This is the first-order stationarity condition of Problem~\eqref{prob3}. Hence, every limit point of $\{\bm{x}^{i}\}$ is a stationary point of Problem~\eqref{prob3}. The proof is completed.
\end{proof}

\section{Proof of Theorem \ref{theorem_sca_gd}}\label{appdix_proof_theorem_sca_gd}

\begin{proof}
For notational simplicity, define the surrogate function at outer iteration $i$ as
\begin{equation}
    Q_i(\bm{x}) = Q(\bm{x};\bm{x}^i).
\end{equation}
The inner projected gradient steps are initialized by
\begin{equation}
    \bm{x}_0^i = \bm{x}^i,
\end{equation}
and updated as
\begin{equation}
    \bm{x}_{j+1}^i
    =
    \Pi_{\mathcal{X}}
    \left(
    \bm{x}_j^i
    -
    \alpha_{i,j}
    \nabla Q_i(\bm{x}_j^i)
    \right),
    \quad j=0,\ldots,M_i-1,
\end{equation}
where $\Pi_{\mathcal{X}}(\cdot)$ denotes the projection onto the feasible set $\mathcal{X}$. The next outer iterate is
\begin{equation}
    \bm{x}^{i+1}=\bm{x}_{M_i}^i.
\end{equation}

Since $Q_i(\bm{x})$ is $L$-smooth and $\alpha_{i,j}\le 1/L$, the standard descent property of gradient descent gives
\begin{equation}
    Q_i(\bm{x}_{j+1}^i)
    \leq
    Q_i(\bm{x}_{j}^i)
    -
    \frac{1}{2\alpha_{i,j}}
    \left\|
    \bm{x}_{j+1}^i-\bm{x}_{j}^i
    \right\|^2.
\end{equation}
Equivalently, by defining the projected gradient mapping of $Q_i$ as
\begin{equation}
    \mathcal{G}_{\alpha}^{Q_i}(\bm{x})
    =
    \frac{1}{\alpha}
    \left[
    \bm{x}
    -
    \Pi_{\mathcal{X}}
    \left(
    \bm{x}
    -
    \alpha\nabla Q_i(\bm{x})
    \right)
    \right],
\end{equation}
we have
\begin{equation}
    Q_i(\bm{x}_{j+1}^i)
    \leq
    Q_i(\bm{x}_{j}^i)
    -
    \frac{\alpha_{i,j}}{2}
    \left\|
    \mathcal{G}_{\alpha_{i,j}}^{Q_i}(\bm{x}_{j}^i)
    \right\|^2.
\end{equation}

In particular, for the first inner step, since $\bm{x}_0^i=\bm{x}^i$, we obtain
\begin{equation}
    Q_i(\bm{x}_{1}^i)
    \leq
    Q_i(\bm{x}^{i})
    -
    \frac{\alpha_{i,0}}{2}
    \left\|
    \mathcal{G}_{\alpha_{i,0}}^{Q_i}(\bm{x}^{i})
    \right\|^2.
\end{equation}
All subsequent inner steps also decrease the surrogate value. Hence,
\begin{equation}
    Q_i(\bm{x}^{i+1})
    =
    Q_i(\bm{x}_{M_i}^i)
    \leq
    Q_i(\bm{x}_{1}^i)
    \leq
    Q_i(\bm{x}^{i})
    -
    \frac{\alpha_{i,0}}{2}
    \left\|
    \mathcal{G}_{\alpha_{i,0}}^{Q_i}(\bm{x}^{i})
    \right\|^2.
\end{equation}

From Proposition~\ref{prop_property}, the surrogate is an upper bound and is tight at the current iterate, i.e.,
\begin{equation}
    \Phi(\bm{x})\leq Q_i(\bm{x}),\quad \forall \bm{x}\in\mathcal{X},
\end{equation}
and
\begin{equation}
    \Phi(\bm{x}^{i})=Q_i(\bm{x}^{i}).
\end{equation}
Therefore,
\begin{align}
    \Phi(\bm{x}^{i+1})
    &\leq Q_i(\bm{x}^{i+1}) \nonumber\\
    &\leq Q_i(\bm{x}^{i})
    -
    \frac{\alpha_{i,0}}{2}
    \left\|
    \mathcal{G}_{\alpha_{i,0}}^{Q_i}(\bm{x}^{i})
    \right\|^2 \nonumber\\
    &=
    \Phi(\bm{x}^{i})
    -
    \frac{\alpha_{i,0}}{2}
    \left\|
    \mathcal{G}_{\alpha_{i,0}}^{Q_i}(\bm{x}^{i})
    \right\|^2.
\end{align}

Moreover, by the first-order consistency in Proposition~\ref{prop_property}, we have
\begin{equation}
    \nabla Q_i(\bm{x}^{i})
    =
    \nabla \Phi(\bm{x}^{i}).
\end{equation}
Thus, the projected gradient mappings of $Q_i$ and $\Phi$ are identical at $\bm{x}^i$, i.e.,
\begin{equation}
    \mathcal{G}_{\alpha_{i,0}}^{Q_i}(\bm{x}^{i})
    =
    \mathcal{G}_{\alpha_{i,0}}^{\Phi}(\bm{x}^{i}).
\end{equation}
Therefore,
\begin{equation}
    \Phi(\bm{x}^{i+1})
    \leq
    \Phi(\bm{x}^{i})
    -
    \frac{\alpha_{i,0}}{2}
    \left\|
    \mathcal{G}_{\alpha_{i,0}}^{\Phi}(\bm{x}^{i})
    \right\|^2.
\end{equation}

Summing the above inequality over $i=0,1,\ldots,I-1$ gives
\begin{equation}
    \sum_{i=0}^{I-1}
    \frac{\alpha_{i,0}}{2}
    \left\|
    \mathcal{G}_{\alpha_{i,0}}^{\Phi}(\bm{x}^{i})
    \right\|^2
    \leq
    \Phi(\bm{x}^{0})-\Phi(\bm{x}^{I})
    \leq
    \Phi(\bm{x}^{0})-\Phi_{\inf},
\end{equation}
where $\Phi_{\inf}$ is a lower bound of $\Phi$ over $\mathcal{X}$. Hence,
\begin{equation}
    \min_{0\leq i\leq I-1}
    \left\|
    \mathcal{G}_{\alpha_{i,0}}^{\Phi}(\bm{x}^{i})
    \right\|^2
    \leq
    \frac{
    2\left(\Phi(\bm{x}^{0})-\Phi_{\inf}\right)
    }{
    \sum_{i=0}^{I-1}\alpha_{i,0}
    }.
\end{equation}
In particular, if $\alpha_{i,0}=\alpha$ for all $i$, then
\begin{equation}
    \min_{0\leq i\leq I-1}
    \left\|
    \mathcal{G}_{\alpha}^{\Phi}(\bm{x}^{i})
    \right\|^2
    \leq
    \frac{
    2\left(\Phi(\bm{x}^{0})-\Phi_{\inf}\right)
    }{
    \alpha I
    }.
\end{equation}
Therefore, an $\epsilon$-stationary point, defined by
\begin{equation}
    \left\|
    \mathcal{G}_{\alpha}^{\Phi}(\bm{x}^{i})
    \right\|
    \leq
    \epsilon,
\end{equation}
can be obtained after at most
\begin{equation}
    I
    \geq
    \frac{
    2\left(\Phi(\bm{x}^{0})-\Phi_{\inf}\right)
    }{
    \alpha\epsilon^2
    }
\end{equation}
outer iterations.

The proof is completed.
\end{proof}

\end{appendices}
\end{document}